\documentclass[reqno,oneside,11pt,letterpaper]{amsart}

\usepackage[usenames,dvipsnames]{xcolor}
\usepackage[colorlinks=true,linkcolor=Blue,citecolor=Blue,linktocpage=true]{hyperref}
\usepackage{amsmath,amssymb,epsf}
\usepackage{amsfonts,amsbsy,amscd,stmaryrd}
\usepackage{latexsym,amsopn}
\usepackage{amsthm,mathtools,tensor}
\usepackage{float,rotating}
\usepackage{graphicx}
\usepackage{mathrsfs,esint} 
\usepackage{autonum}

\allowdisplaybreaks

\let\origmaketitle\maketitle
\def\maketitle{
  \begingroup
  \def\uppercasenonmath##1{} 
  \let\MakeUppercase\relax 
	\origmaketitle
  \endgroup
}

\makeatletter
\@namedef{subjclassname@2020}{%
  ${2020}$ Mathematics Subject Classification}
\makeatother

\numberwithin{equation}{section}
\newcounter{smallarabics}
\newenvironment{arabicenumerate}
{\begin{list}{{\normalfont\textrm{(\arabic{smallarabics})}}}
  {\usecounter{smallarabics}\setlength{\itemindent}{0cm}
   \setlength{\leftmargin}{5ex}\setlength{\labelwidth}{4ex}
   \setlength{\topsep}{0.75\parsep}\setlength{\partopsep}{0ex}
   \setlength{\itemsep}{0.75ex}}}
{\end{list}}
\newcommand{\ben}{\begin{arabicenumerate}}
\newcommand{\een}{\end{arabicenumerate}}
\def\init{\setcounter{equation}{0}}
\newenvironment{notations}
{\begin{list}{{\normalfont\textrm{-}}}
  {\setlength{\itemindent}{0cm}
   \setlength{\leftmargin}{2ex}\setlength{\labelwidth}{4ex}
   \setlength{\topsep}{0.75\parsep}\setlength{\partopsep}{1ex}
   \setlength{\itemsep}{1ex}}
}
{\end{list}}

\newtheorem{theorem}{Theorem}[section]
\newtheorem{assumption}[theorem]{Hypothesis}
\newtheorem{proposition}[theorem]{Proposition}
\newtheorem{lemma}[theorem]{Lemma}
\newtheorem{corollary}[theorem]{Corollary}
\theoremstyle{definition}
\newtheorem{definition}[theorem]{Definition}
\newtheorem{remark}[theorem]{Remark}
\newtheorem{example}[theorem]{Example}
\newcommand{\beq}{\begin{equation}}
\newcommand{\eeq}{\end{equation}}
\newcommand{\bea}{\begin{aligned}}
\newcommand{\eea}{\end{aligned}}
\newcommand{\bex}{\begin{example}}
\newcommand{\eex}{\end{example}}
\def\bel{\begin{lemma}}\def\eel{\end{lemma}}
\def\bet{\begin{theorem}}\def\eet{\end{theorem}}
\def\bed{\begin{definition}}\def\eed{\end{definition}}
\def\ber{\begin{remark}}\def\eer{\end{remark}}
\def\Lemma{Lem.~}
\def\Lemmas{Lem.~}

\def\beproof{\noindent{\bf Proof.}\ }
\renewenvironment{proof}{\beproof}{\qed}

\newcommand{\step}[1]{{\noindent\emph{Step #1.}} }
\def\fantom{\\ &\phantom{=}\,}
\renewcommand{\leq}{\leqslant}\renewcommand{\geq}{\geqslant}
\newcommand{\open}[1]{\mathopen{}\mathclose{\left]#1 \right[}}

\newcommand{\norm}[1]{\left\|{#1}\right\|}
\newcommand{\module}[1]{\left|#1\right|}

\def\st{{ \ |\  }}
\newcommand*{\defeq}{:=}\newcommand*{\eqdef}{=:}
\newcommand{\mat}[4]{\begin{pmatrix}#1 &#2  \\ #3 &#4 \end{pmatrix}}

\newcommand{\col}[2]{\begin{pmatrix}#1 \\#2\end{pmatrix}}

\def\rr{{\mathbb R}}\def\cc{{\mathbb C}}\def\nn{{\mathbb N}}

\def\bS{\mathbb{S}}

 \def\i{{\rm i}}
\def\e{{e}}
\newcommand{\Sig}{\Sigma}

\def\cD{{\mathcal D}}\def\cE{{\mathcal E}}\def\cF{{\mathcal F}}\def\cH{{\mathcal H}}\def\cN{\mathcal{N}}\def\cU{{\mathcal U}}\def\cV{{\mathcal V}}\def\cW{{\mathcal W}}

\def\rx{{\rm x}}\def\rk{{\rm k}}

\renewcommand{\Im}{\operatorname{Im}}
\renewcommand{\Re}{\operatorname{Re}}

\DeclareMathOperator{\Diff}{Diff}
\DeclareMathOperator{\supp}{supp}
\DeclareMathOperator{\Op}{Op}

\DeclareMathOperator{\Ker}{Ker}
\DeclareMathOperator{\Ran}{Ran}
\DeclareMathOperator{\Dom}{Dom}

\DeclareMathOperator{\rs}{rs}

\newcommand{\one}{\boldsymbol{1}}

\newcommand{\traa}[1]{\mskip-6mu\upharpoonright_{#1}}

\def\slim{{\rm s-}\lim}
\def\p{\partial}

\def\cinf{C^\infty}
\def\cinfb{C^{\infty}_{\rm b}}

\def\coinf{{C}^\infty_{\rm c}}

\def\zero{{\rm\textit{o}}}
\def\vol{{\rm vol}}
\def\12{\frac{1}{2}}

\def\wf{{\rm WF}}
\def\WF{{\rm WF}}
\def\dvol{\mathop{}\!d{\rm vol}}

\def\bar{\overline}
\def\dual{\!\cdot \!}
\def\CCR{{\rm CCR}}
\def\medcap{\textstyle\bigcap}

\def\tosim{\xrightarrow{\sim}}
\def\bij{\xrightarrow{\sim}}

\def\calde{Calder\'{o}n }

\def\Riem{{\rm Riem}_{\rg}}
\def\Ric{{\rm Ric}}
\def\rg{{g}}
\def\rh{{h}}
\def\rR{{\rm R}}
\def\rk{{k}}
\def\ru{{\bf u}}
\def\rU{{\bf U}}
\def\bs{{s}}
\def\rer{{r}}

\def\bG{{\rm G}} 
\def\scal{{\rm R}}

\def\tr{{\rm tr}}
\def\div{{\rm div}}

\def\nab{\nabla}

\def\tM{{\tilde M}}
\def\tV{\tilde{V}}

\newcommand\Texp{{\rm  Texp}}
\def\ha{\widehat}
\def\BT{{BT}}

\def\trace{{\rm tr}}
\def\div{{\rm div}}

\def\nab{\nabla}

\def\Diff{{\rm Diff}}
\def\cF{\mathcal{F}}
\def\tV{\tilde{V}}
\def\trh{\tilde{\rh}}
\def\ta{\tilde{a}}
\def\tA{\tilde{A}}
\def\tD{{\tilde{D}}}
\def\teps{\tilde{\epsilon}}
 \def\tr{\tilde{r}}
 \def\tb{\tilde{b}}
 \def\trho{\tilde{\varrho}}
\def\dvol{d{\rm vol}}
\def\cinfsc{C^{\infty}_{\rm sc}}
\def\tsig{\tilde{\sigma}}
\def\ha{}
\def\AT{\mathcal{A}T}

\newcommand{\modif}[1]{{#1}}
\newcommand{\newmodif}[1]{{#1}}

\author{}
\address{Laboratoire de Math\'ematiques d'Orsay, Universit\'e Paris-Saclay, France}
\email{christian.gerard@math.u-psud.fr}
\author{}
\address{Dipartimento di Matematica, Universit\`a di Genova, Italy}
\email{murro@dima.unige.it}
\author[]{\normalsize Christian \textsc{G\'erard}, Simone \textsc{Murro} \& Micha{\l} 
\textsc{Wrochna}}
\address{\modif{Mathematical Institute, Utrecht University, The Netherlands}}
 \email{\modif{m.wrochna@uu.nl}}
\date{November 2023}
\title[\modif{Wick rotation of linearized gravity and \calde projectors}]{\modif{\large Wick rotation of linearized gravity in Gaussian time}\\ \modif{ \large and \calde projectors}}

\begin{document}

\begin{abstract}\modif{Motivated by the quantization of linearized gravity, we consider gauge-fixed linearized Einstein equations and  their Wick rotation near a Cauchy surface. We show that \calde projectors  for the Wick-rotated equations induce Hadamard bi-solutions on the Lorentzian level.  On the other hand, we find smoothing obstructions to gauge-invariance and positivity conditions needed in quantization. These obstructions are primarily due to boundary terms arising in the Wick-rotated theory and depend on the boundary conditions. 
} 
\end{abstract}

\maketitle

\section{Introduction and summary}\label{sec:intro}

\subsection{Introduction} The  quantization of linearized gravity has remained  incomplete  for a long time.  The structure of the classical theory needed for the quantization is nowadays relatively well understood \cite{FH,HS,BFR,BDM} and various possible candidates for physical states are discussed in the physics literature  in some cases of highly symmetric Einstein metrics, it is however unclear if they are indeed positive. The  rigorous construction of physical {states} has remained an open problem: in fact, there is presently no mathematically satisfactory answer even for perturbations of Minkowski space.

Beside the prerequisite of being positive, the main criterion for a state to be physical is the so-called \emph{Hadamard condition}, which is needed for the renormalization of non-linear quantities (see e.g.~\cite{FNW,AFO,FV})  and for the mathematical formulation of perturbative quantum gravity  \cite{BFR,BFH,rejzner}. Unfortunately, the standard deformation argument for the existence of Hadamard states~\cite{FNW} does not apply to linearized gravity. In fact, the spacetime $(M,\rg)$ which serves as the background for the linearization  must  be a solution of the non-linear Einstein equations, so it cannot be arbitrarily deformed. Furthermore, more advanced techniques based on pseudo\-differential calculus~\cite{junker,GOW} fail to preserve the gauge invariance of the equations if applied directly, whereas conformal scattering methods~\cite{DMP} are affected by divergent behaviour at conformal infinity~\cite{BDM}. A pseudodifferential construction of Hadamard states is known in the case of Yang--Mills fields linearized around a possibly non-zero solution~\cite{GW1},  it uses however a spacetime deformation argument and various methods which are specific to differential forms.

Let us explain the problem in more detail. Let $(M,\rg)$ be a globally hyperbolic spacetime with $\dim M =4$, such that $\rg$ solves the non-linear Einstein equations. We consider the two differential operators
$$
P=  - \square-  I \circ d \circ \delta + 2\,\Riem, \quad  K=I\circ d,
$$
acting on  symmetric $(0,2)$-tensors, resp.~$(0,1)$-tensors, 
where: \smallskip
\ben\setlength{\itemindent}{0.65cm}
\item[] $\square$ is the  d'Alember\-tian associated to $\rg$, i.e.~$(\square u)_{ab} =  \nab^{c}\nab_{c}u_{ab}$, 
\item[] $I$ is the trace reversal, i.e.~$(Iu)_{ab}= u_{ab}- \frac{1}{2} \trace_{\rg}(u) \rg_{ab}$,
\item[] $d$ the symmetric differential, i.e.~$(d w)_{a b}= \nab_{(a}u_{b)}$,
\item[]  $\delta$ is the formal adjoint of $d$, i.e.~$(\delta u)_{a}= -2\nab^{c}u_{ca}$,
\item[] $\Riem$ is the Ricci operator, i.e.~$\Riem(u)_{ab}= \rR\indices{_{a}^{cd}_{b}}u_{cd}= \rR\indices{^{c}_{ab}^{d}}u_{cd}$,
\een \smallskip
(see Subsect.~\ref{sec:lg.1.1} for more details on our conventions). Then \modif{the equation} $Pu=0$ \modif{is equivalent to} the \emph{Einstein equations linearized around $\rg$} \modif{(these are usually formulated using the linearized Einstein operator $\12 P \circ I$)}. The identity $PK=0$ means that $P$ is invariant under  linear gauge transformations $u\mapsto u+ Kw$, and it is responsible for the fact that $P$ is not hyperbolic. The \emph{de Donder gauge} or \emph{harmonic gauge} consists in considering the  \modif{``gauge-fixed''} hyperbolic operator  
$$
D_2= P+  K K^\star,
$$ where $K^\star$ is the formal adjoint for a non-positive Hermitian form  involving $I$ (see Subsect.~\ref{sec:lg.2}). Then, solutions of the linearized Einstein equations $Pu=0$ are obtained by solving $D_2 u =0$ with the gauge condition $K^\star u=0$. An auxiliary role is played by the hyperbolic operator $D_1=K^\star K$ \modif{acting on $(0,1)$-tensors}, which satisfies $D_2 K = K D_1$.

In the quantization problem, of particular relevance are \emph{bi-solutions} which correspond to two-point functions of states. In fact, finding a Hadamard state amounts to constructing  a pair of operators $\lambda_{2}^{\pm}$  \modif{acting on symmetric $(0,2)$-tensors on $M$} such that:
$$
\bea
(1) & \quad D_{2}\lambda^\pm_{2}=\lambda^\pm_{2} D_{2} =0  \,\mbox{  and  }   \,  \i(\lambda_2^+ - \lambda_2^-) \mbox{ is the causal propagator of } D_2,\\\
(2) & \quad   (\lambda^\pm_{2})^{\star}=\lambda^\pm_{2} \,\mbox{  and  }   \, \lambda^\pm_{2}: \Ran K|_{\coinf}\to\Ran K|_{\cD'},\\
(3) & \quad\lambda^\pm_{2}  \geq 0 \, \mbox{ on }  \Ker K^{\star}|_{\coinf},\\
(4) & \quad \wf(\lambda_2^\pm)'  \subset \cN^{\pm}\times \cN^{\pm}.
\eea
$$
The second part of (1) corresponds to the \emph{canonical commutations relations}. Condition~(2) is the \emph{gauge invariance} and (3) the \emph{positivity}. Put together, they ensure that one gets a well-defined quantum state on the physical space $K^\star h=0$. Then, the Gelfand--Naimark--Segal construction yields quantum field operators for linearized gravity. In addition, one requires the wavefront set estimate (4): this is the celebrated \emph{Hadamard condition} (see Subsect.~\ref{sss:had} for more details) which ensures the correct short-distance behavior of fields.

The first significant difficulty as compared to the scalar Klein--Gordon equation is that the Hermitian form in  (2)--(3) is \emph{not} positive and in consequence the positivity condition becomes extremely delicate. Furthermore, the use of pseudo\-differential calculus is helpful to get condition (4), but it interacts  badly with conditions (2)--(3). The impossibility of using a spacetime deformation argument makes various  previously developed  methods inapplicable.   

\subsection{\modif{Setting}} In this paper we  \modif{focus on} globally hyperbolic backgrounds $(M,\rg)$ with the following properties I. and II.  First of all, we assume:
\begin{enumerate}
 \item[I.]  $(M,\rg)$ is a Lorentzian manifold of bounded geometry near a Cauchy surface $\Sigma$.
\end{enumerate}
\noindent Roughly speaking, the bounded geometry assumption  means that all relevant geometric quantities associated with $\rg$ are bounded with all their derivatives relatively to some fixed reference Riemannian metric $\hat g$, see Subsect.~\ref{an-bg.1}.  Using Gaussian normal coordinates to $\Sigma$ we obtain  a diffeomorphism   $\chi$ defined in a neighborhood of $\Sigma$  such that
$$
\chi^{*} \rg= - dt^{2}+\rh_{t},
$$ 
where the function $t\mapsto \rh_t$ takes values in Riemannian metrics of bounded geometry and is bounded with all its derivatives. Our second assumption is:
\begin{enumerate}
 \item[II.]   The map $t\mapsto \rh_t$  is  real analytic.
\end{enumerate}
The precise formulation is given in Subsect.~\ref{hippopotamus}. In practice it is often easier to check the stronger condition of \emph{bounded analytic geometry}, which is an analyticity condition in all variables, see Subsect.~\ref{an-bg.2}. We show that it is satisfied in examples such as Minkowski, de Sitter, Kerr--Kruskal and Schwarzschild--de Sitter spacetimes.



\subsection{Plan of \modif{paper}} 
The main idea \modif{discussed in the present paper}, inspired by Euclidean gravity approaches to quantization, is to consider an elliptic operator $\tilde D_2$ which is a \emph{Wick-rotated} version of the hyperbolic operator $D_2$ near a fixed Cauchy surface $\Sigma$.  We then construct the Cauchy data of $\lambda_2^\pm$ in terms of the \emph{Calder\'on projectors} $\tilde c_2^\pm$ for  $\tilde D_2$. Denoting by $s\modif{=\i^{-1} t}$ the imaginary time variable and fixing suitable boundary conditions,  this means that $\tilde c_2^\pm$ is the projection to Cauchy data of $L^2$ solutions of the elliptic problem $\tilde D_2 u=0$ in $\{ \pm s >0\}$. To put it briefly, Cauchy data for an   elliptic problem are used as Cauchy data in Lorentzian signature.

The motivation is that this type of construction is known to carry good positivity and wavefront set properties in the real analytic, scalar case \cite{GW2,schapira}. Furthermore $\tilde c_2^\pm$ are given by a direct formula in terms of the inverse of $\tilde D_2$, so one could hope to prove the gauge invariance condition (2)   using algebraic identities satisfied by $\tilde D_2$.

However, it is not a priori clear how to do the Wick rotation in a way that guarantees the invertibility of a boundary value problem associated with the elliptic operator $\tilde D_2$, nor how to prove the Hadamard condition without assuming full analyticity. Worst of all, the  {Calder\'on projectors} $\tilde c_2^\pm$  come with some positivity properties indeed (as expected from the scalar case), but with respect to an unphysical  Euclidean version of the physical Hermitian form, so there is no particular reason for the positivity condition (3) to be satisfied.

In the present paper we tackle \modif{part of these issues in the following order.}

\begin{notations}
\item In Sect.~\ref{sec:gauge} we recall the general structure of linear gauge theories on Lorentzian spacetimes in the formalism due to Hack--Schenkel \cite{HS}, which in particular serves us to justify conditions (1)--(4). In Sect.~\ref{an-bg} we introduce the notions of bounded geometry which  enter the assumptions I.~and II.~and we derive various examples. As an intermediate step, in Thm.~\ref{cauchykow} we show   that it suffices to check analytic bounded geometry assumptions for the initial data of Einstein metrics.  
  
\item    In Sect.~\ref{sec:lg} we explain how exactly linearized gravity fits into the framework of Sect.~\ref{sec:gauge}. We then use Gaussian normal coordinates near a Cauchy surface $\Sigma$ and a parallel transport argument to reduce $D_2$ to an operator of the form $\p_t^2 + a_2(t)$, where $a_2(t)$ is a family of elliptic operators acting in the spatial variables.

\item In Sect.~\ref{sec01}  we consider the hyperbolic operator $\p_t^2 + a_2(t)$ and in parallel we study boundary value problems for its Wick-rotated elliptic analogue $-\p_s^2 + a_2(\i s)$. Using Shubin's pseudo\-differential calculus on manifolds of bounded geometry, we construct a parametrix for \emph{Hadamard projectors} $c_2^\pm$ (operators that project to Cauchy data of solutions with wavefront set in $\cN^\pm$) and for Calder\'on projectors $\tilde c_2^\pm$  \modif{at $s=0$ for the operator $\tilde D_2$ with Dirichlet boundary conditions at some finite imaginary times $s=\pm T$}. The two parametrices are deduced from similar pseudo\-differential factorizations of the respective operator $\p_t^2 + a_2(t)$ or  $-\p_s^2 + a_2(\i s)$, and are therefore related: we show that in fact, \emph{Hadamard and Calder\'on projectors coincide modulo a smoothing term}. In conclusion,  we can define $\lambda_2^\pm$ using $\tilde c_2^\pm$ as Cauchy data, and then $\lambda_2^\pm$ satisfies the Hadamard condition.  


\item Finally, Sect.~\ref{sec3} \modif{discusses gauge-invariance and positivity on the physical space. The operators $\lambda_2^\pm$ defined from the Calder\'on projectors $\tilde c_2^\pm$ are positive for an auxiliary scalar product, but this scalar product unfortunately differs from the physical Hermitian form in condition (3). However, on the Cauchy surface level,  the two inner products coincide on tensors which have \emph{no mixed components and  are invariant under trace reversal}. The idea is then to find a gauge transformation that maps to tensors of this type, at least on the level of Cauchy data at $t=0$. This is achieved in Lem.~\ref{lem:alt2} in the elliptic setting in the case of Dirichlet boundary conditions.  The problem with that strategy is that $\lambda_2^\pm$ turns out \emph{not} to satisfy gauge invariance (2), and this results in positivity modulo a smoothing term. The  obstruction to gauge invariance is due to the choice of boundary conditions for $\tilde D_2$---this motivates further work , possibly on different boundary conditions in Wick-rotated gravity. }
\end{notations}

These results will be used in subsequent works on quantization of linearized gravity in different settings.

\subsection{Bibliographical remarks} The importance of  Hadamard states for renormalization was realized already in the 1970s,  cf.~Allen--Folacci--Ottewill \cite{AFO} for the case of linearized gravity.  More recently it was re-emphasi\-zed  in works 
by Brunetti, Fredenhagen, Rejzner and other authors \cite{BFR,FR,BFH,rejzner} on the perturbative approach to effective theories of quantum gravity.

The problem of constructing Hadamard states in linear gauge theories was considered in the simplest case of Maxwell equations by Furlani \cite{furlani},  Fewster--Pfenning \cite{FP}, Dappiaggi--Siemssen \cite{DS} and Finster--Strohmaier \cite{FS}. The construction of Hadamard states for Yang--Mills equations linearized around the zero solution  in the BRST framework is due to Hollands \cite{hollands}. Later, Gérard--Wrochna \cite{GW1} considered Yang--Mills equations linearized around non-zero solutions. As pointed out in the introduction, none of these constructions can be  adapted to linearized Einstein equations. The case of linearized Einstein equations on asymptotically flat spacetimes was studied by Benini--Dappiaggi--Murro \cite{BDM} with methods drawing from earlier works of Ashtekar--Magnon-Ashtekar \cite{AA} and Dappiaggi--Moretti--Pinamonti \cite{DMP}, the quantization turns out however to be limited to a subspace of classical degrees of freedom due to divergences at null infinity.

The general structure of quantized linearized gravity and various candidates for \modif{two-point functions of} states on specific spacetimes were studied in many works, too numerous to be listed here exhaustively \modif{(however, the physics literature does not appear to address the problem of positivity)}.

 Building on works among others by Moncrief \cite{moncrief},  the  symplectic structure of linearized gravity
was studied by Fewster--Hunt \cite{FH}; the latter was then built into a more general framework for linear gauge theories (used in the present paper)  by Hack--Schenkel \cite{HS}. The work of Fewster--Hunt also considers the \emph{TT gauge}\footnote{This stands for ``transverse and traceless'', though one should be aware that outside of the special case of Minkowski space one usually  means by this just ``traceless''.} (often used in the physics literature when discussing quantization of linearized gravity) and examines if and under what assumptions it can be implemented on spaces of space-compact solutions.

In particular, many authors  analyzed the question of the existence of a maximally symmetric state on de Sitter space (see e.g.~\cite{dS1,dS2} and references therein), which appears to be extremely subtle due to infrared problems. In our construction, Dirichlet boundary conditions at finite imaginary times  act as a universal infrared regularization (cf.~\cite[Subsect.~4.5]{GW2} for an illustration in the scalar case), but they manifestly break the symmetries on top of the issues with gauge-invariance. In that particular case, however, it is possible to Wick rotate to the sphere and then  boundary conditions are no longer needed, at the cost of having to deal with infrared problems directly: this will be addressed in a  forthcoming paper.


Calder\'on projectors have been studied in various settings, including recently general fibred cusp operators by Fritzsch--Grieser--Schrohe \cite{Fritzsch2023}. One of the main  differences between our situation and the ones typically considered in the literature  is that $\{s=0\}$ plays the role of an interface between two {symmetric} regions, rather than being the boundary of one region which can be extended in an arbitrary way. We also mention  the  recent construction of the Hartle--Hawking--Israel state for scalar fields with closely related techniques  \cite{HHI}.

Finally, we remark that there has been  much recent progress in  the analysis of linearized Einstein equations with microlocal methods  in the context of black hole stability, see e.g.~Hintz--Vasy \cite{Hintz2018a} and Häfner--Hintz--Vasy \cite{HHV}. The techniques we use here are largely different due to the local-in-time character of the problem. However, both kinds of developments  are expected to be relevant to the construction of Hadamard  states with distinguished asymptotic properties,  a problem which remains open for now in the case of linearized gravity.

\subsection{Notation}\label{sec01.1}

Before moving on to the main part of the paper, let us introduce the relevant notation.

\subsubsection{Sections of vector bundles}\label{sec01.1.1}
Let  $V\xrightarrow{\pi}M$ be  a finite rank complex vector bundle  over a smooth manifold $M$.  

\begin{notations}
\item If $\Sigma\subset M$ is a smooth manifold we denote by $V|_{\Sigma}\xrightarrow{\pi}\Sigma$ the restriction of $V$ to $\Sigma$.
\item We denote by $\cinf(M; V)$, resp.~$\coinf(M; V)$ the space of smooth, resp.~ compactly supported smooth sections of $V$.  
\item We denote by $\cD'(M; V)$, resp.~$\cE'(M; V)$ the space of distributional, resp.~ compactly supported distributional  sections of $V$. 
\item If $\Omega\subset M$ is an open set with smooth boundary and $F(M; V)$ is one of the above spaces, we denote by $\bar{F}(\Omega; V)\subset F(\Omega; V)$ the space  of restrictions of $F(M; V)$ to $\Omega$. For example $\bar{\cD'}(\Omega; V)$ is the space of  {\em extendible} distributional sections in $\cD'(\Omega; V)$.
\item If $\Sigma\subset M$ is a submanifold, we denote by $F(\Sigma; V)$ the corresponding space of sections of the restriction $V|_{\Sigma}\xrightarrow{\pi}\Sigma $  of $V$ to $\Sigma$.
\end{notations}

Equipped with their natural seminorms, all the above vector spaces of sections are Fr\'echet spaces.

We use the same notations if $V$ is a finite dimensional vector space, i.e.~we write simply $V$ instead of the trivial vector bundle $M\times V$.

\subsubsection{Globally hyperbolic spacetimes}
We use the convention  $(-,+,\dots,+)$ for the Lo\-rent\-zian signature. Let us recall that a \emph{globally hyperbolic spacetime} is a smooth  Lorentzian manifold equipped with a time orientation and    having a  Cauchy surface $\Sigma$, i.e., a  closed subset of $M$ which is intersected exactly once by each maximally extended time-like curve.
By the Bernal--S\'anchez theorem, this definition implies that  $M$ admits smooth space-like Cauchy surfaces.

\begin{notations}
\item We denote by $J_{\pm}(K)$ the future/past causal shadow of  $K\subset M$.
\item If $M$ is a globally hyperbolic spacetime  we   denote by $\cinf_{\rm sc}(M;V)$ the space of space-compact sections, i.e.~sections in $\cinf(M;V)$ with compactly supported restriction to a Cauchy surface.
\end{notations}

\subsubsection{Distributional kernels and wavefront sets}

 If $u\in \cD'(M; V)$ we denote by $\WF (u)\subset T^{*}M\setminus\zero$ its {\em wavefront set}, which is invariantly defined using local trivializations of $V$. 
 \begin{notations}
  \item If $V_{i}\xrightarrow{\pi}M_{i}$ are two vector bundles as above and $A: \coinf(M_{1}; V_{1})\to \cD'(M_{2}; V_{2})$ is linear continuous, then $A$ admits a distributional kernel, still denoted by $A\in \cD'(M_{2}\times M_{1}; V_{2}\boxtimes V_{1})$.  
  \item We denote by $\WF(A)'\subset (T^{*}M_{2}\times T^{*}M_{1})\setminus \zero$ its {\em primed} wavefront set, defined by 
  \[
  \Gamma'= \{((x_{2}, \xi_{2}), (x_{1}, - \xi_{1})) \st ((x_{2}, \xi_{2}), (x_{1}, \xi_{1}))\in \Gamma\}\hbox{ for }\Gamma\subset T^{*}M_{2}\times T^{*}M_{1}.
  \]
  \end{notations}
  
  \subsubsection{Hermitian bundles}\label{sec0.1.1b}
If $V$ is equipped with a fiberwise non-degenerate Hermitian form $(\cdot| \cdot)_{V}$, we say that $V$ is a {\em Hermitian bundle}. If the Hermitian form is positive definite, we say that $V$ is a {\em Hilbertian bundle}. 

 We will always assume that $M$ is equipped with a pseudo-Riemannian or Riemannian metric $g$ and denote by $d\vol_{g}$ the associated volume form.
 
\begin{notations} 
\item  If $V$ is a Hermitian bundle over $M$ and $U\subset M$ is an open set,  we set
\[
(u| u)_{V(U)}\defeq \int_{U}(u(x)| u(x))_{V}d\vol_{g}, \ u\in \coinf(U; V).
\]
We use $(\cdot|\cdot)_{V(M)}$ to inject $\coinf(M; V)$, resp.~$\cinf(M; V)$, into $\cD'(M; V)$, resp.~$\cE'(M; V)$.
\item  We denote by  $A^*$   the formal adjoint of an operator $A$. 
\item If we need to  consider simultaneously two different Hermitian structures, the two adjoints of $A$ will be denoted by  $A^{*}$ and $A^\star$.
 \end{notations}

\subsubsection{Differential operators}\label{sec01.1.2}
If $V_{1}$, $V_{2}$ are two vector bundles over $M$, we denote by  $\Diff(M;V_{1},V_{2})$ (resp.~$\Diff^m(M;V_{1},V_{2})$) the corresponding space of differential operators (resp.~differential operators of order $m$), equipped with their Fr\'echet space topologies.

\begin{notations}
 \item We abbreviate $\Diff^{(m)}(M; V, V)$ by  $\Diff^{(m)}(M; V)$.
\item  If $A\in \Diff(M;V_{1},V_{2})$ we denote by $A|_{\cinf}$, resp.~$A|_{\coinf}$,  $A|_{\cinfsc}$ its restriction to the space $\coinf(M;V_1)$, resp.~$\cinf(M;V_1)$, $\cinfsc(M; V_{1})$, and we use analogous notation for distributional sections (or vector-valued distributions) $\cD'(M;V_1)$ and compactly supported distributional sections $\cE'(M;V_1)$. 
\end{notations}

\subsubsection{Time dependent objects}\label{sec01.1.3}
If $I\subset \rr$ is an interval  and $\cF$ is a Fr\'echet space whose topology is defined by a family of  seminorms $\norm{ \cdot}_{n}$, $n\in \nn$, we denote by $\cinfb(I, \cF)$ the space of maps $f: I\to \cF$ such that $\sup_{t\in I}\| \p_{t}^{p}f(t)\|_{n}<\infty$ for all $n, p\in \nn$. 
Equipped with the obvious seminorms, $\cinfb(I, \cF)$ is a Fr\'echet space.

We use this notation to define for example $\cinfb(I; \cinfb(\Sigma; V))$, etc.

\subsubsection{Evolution groups}\label{sec01.1.4}
Let $\cH$ be a Hilbert space, $\cD\subset \cH$  a dense subspace, $I\subset \rr$ an interval and $I\ni t\mapsto H(t)$ a map with values in closed linear operators in $L(\cD, \cH)$.   Assume that   for some $z_{0}\in \cc$ one has:
\[
\begin{array}{rl}
i)&z_{0}+ \i \rr\in \rho(H(t)) \,\mbox{ and }  \,\sup_{\rr}\|(z_{0}+ \i \lambda- H(t))^{-1}\|_{B(\cH)}<\infty\ \ \forall t\in I,\\[2mm]
ii)&I\ni t\mapsto H(t)\in L(\cD, \cH)\hbox{ is  strongly }C^{1}.
\end{array}
\]
Then by a result of Kato \cite{K2},  see e.g.~\cite{SG} for a recent exposition, there exists a unique propagator $I\times I\ni (t,s)\mapsto U(t,s)\in B(\cH)$ such that $U(t,s): \cD\to \cD$ and:
\[
 \begin{cases}
\p_{t}U(t,s)u= \i H(t)U(t, s)u,\\[2mm]
U(s,s)u= u
\end{cases} 
\]
for all $u\in \cD$ and $t,s\in I$.
Following the physics literature we will denote $U(t, s)$ by $\Texp(\i \int_{s}^{t}H(\sigma)d\sigma)$.

\section{Preliminaries on gauge theories}\label{sec:gauge}\init

In this section we recall an abstract formalism for the quantization of (linear) gauge theories on curved spacetimes \cite{HS}.  We discuss various equivalent phase spaces that can be used for the algebraic quantization, in particular the phase space obtained by fixing a Cauchy surface and considering initial data.

We also formulate the  definition of Hadamard states for gauge theories and give conditions on the Cauchy surface covariances of a state that imply the Hadamard property.
\subsection{Solution spaces for hyperbolic equations}\label{ss:hyperbolic} Before discussing gauge theories, we recall the setup relevant for hyperbolic equations.

Let $(M,\rg)$ be a globally hyperbolic spacetime, and let $V$ be a Hermitian bundle over $M$. We denote by $A^{\star}$ the formal adjoint of $A\in \Diff(M; V)$.

 One says that $D\in\Diff(M;V)$ is \emph{Green hyperbolic} if $D$  has   retarded/advanced inverses $G_{\rm ret/adv}$ and the same holds true for $D^\star$. This is the case in particular if the principal symbol $\sigma_{\rm pr}(D)$ of $D$ equals $\xi\dual \rg^{-1}(x)\xi\,\one_{V}$, i.e.~is given by the inverse metric, see \cite[Thm.~3.3.1]{BGP}.  

The difference $G\defeq G_{\rm ret}-G_{\rm adv}$ is called interchangeably with its Schwartz kernel the \emph{causal propagator} (or \emph{Pauli-Jordan function}) of $D$.  We recall its fundamental properties below, see \cite[Lem.~3.3 \& Thm.~3.5]{Bar2012} for a proof at this level of generality.

\begin{proposition}\label{prop:basic} Suppose that $V$ is a Hermitian bundle over $M$ and   $D\in\Diff(M;V)$ is a Green hyperbolic operator such that $D^{\star}=D$. Then:
\ben
\item\label{prop:basic2} the induced map 
\[
[G]:\,\frac{\cinf_{\rm c}(M;V)}{\Ran D|_{\cinf_{\rm c}}}\longrightarrow\Ker D|_{\cinf_{\rm sc}}
\]
is well defined and bijective;
\item $(G_{\rm ret/adv})^{\star}=G_{\rm adv/ret}$ and consequently, $G^{\star}=-G$.
\een
\end{proposition}

\subsubsection{Green's formula}\label{greeno}
Let us fix a Cauchy surface $\Sigma$ of $(M,\rg)$. It can be shown that under  the assumptions of Proposition \ref{prop:basic}, $D$ has a well-posed Cauchy problem at $\Sigma$, see e.g.~\cite[Subsect.~4.3]{Kh}. 

More precisely,  there exists a  Hilbertian  bundle $V_{\Sigma}$ over $\Sigma$ and a continuous operator $\varrho :\cinfsc(M;V)\to\coinf(\Sigma;V_{\Sigma})$, mapping a smooth section to its Cauchy data on $\Sigma$,   such that 
\beq\label{eq:rhoinv}
\varrho : \Ker D|_{\cinfsc}\to\coinf(\Sigma;V_{\Sigma})
\eeq
is bijective.  

By Prop.~ \ref{prop:basic}, there exists a unique $q :\coinf(\Sigma;V_{\Sigma})\to \coinf(\Sigma;V_{\Sigma})$ such that $q=q^*$ and
\beq\label{eq:Gq}
(\phi_1| G \phi_2)_{V(M)} = \i^{-1}( \varrho u_1  | q  \varrho u_2)_{V_{\Sigma}},
\eeq
for all $\phi_i\in \coinf (M;V)$ and $u_i=G \phi_i \in  \Ker D|_{\cinfsc}$, $i=1,2$. The l.h.s.~is manifestly  independent on the choice of Cauchy surface $\Sigma$, therefore $q$ is a conserved quantity along the Cauchy evolution,   called the \emph{conserved charge} or simply \emph{charge} of $D$. 

Note that the literature often uses the \emph{conserved symplectic form}, which is given by the r.h.s.~of \eqref{eq:Gq}. It defines the complex symplectic space of Cauchy data (or equivalently, the  symplectic space of solutions of $Du=0$) which is fundamental for field quantization.

  A practical way of computing $q$ is provided by the following elementary lemma.

\begin{lemma}\label{lemmatiti.1b} For $D$ as in Proposition \ref{prop:basic}, the  charge is the unique  $q\in \Diff(\Sigma;V_{\Sigma})$ such that for all $u,v\in\coinf(M;V)$:
\beq\label{eq:green}
(u | D v)_{V(J_\pm(\Sigma))}- ( D u |  v )_{V(J_\pm(\Sigma))}=\pm\i^{-1}( \varrho u | q  \varrho v)_{V_{\Sigma}},
\eeq
where we recall that  $J_{\pm}(K)$  are the future/past causal shadows of  $K\subset M$.
\end{lemma}
\proof Observe that the l.h.s.~of \eqref{eq:green} can be rewritten as $(u| [\one_{J_{\pm}(\Sigma)}, D]v)_{V(M)}$, which shows that it depends only on the traces of $u, v$ on $\Sigma$. 
Let $\tilde{q}\in  \Diff(\Sigma; V_{\Sigma})$ be defined by the r.h.s.~of \eqref{eq:green}.  Let also $u_i=G \phi_i \in  \Ker D|_{\cinfsc}$, $i=1,2$.  Without loss of generality we can assume that $\supp \phi_{i} $ are in  the past of $\Sigma$, so that $G\phi_{i}= G_{\rm ret}\phi_{i}$ near $\Sigma$ and $\supp G_{\rm ret}\phi_{i}\cap J_{-}(\Sigma)$ is compact. Let us fix $\chi\in \coinf(M; \rr)$ such that $\chi\equiv 1$ near $\supp G_{\rm ret}\phi_{i}\cap J_{-}(\Sigma)$.  Then, we have:
\[
\bea
&-\i^{-1}(\varrho G \phi_{1}| \tilde{q} \varrho G \phi_{2})_{V_{\Sigma}}=  -\i^{-1}(\varrho \chi G_{\rm ret} \phi_{1}| \tilde{q} \varrho \chi G_{\rm ret} \phi_{2})_{V_{\Sigma}}\\
&=(\chi G_{\rm ret} \phi_{1}| D \chi G_{\rm ret} \phi_{2})_{V(J_{-}(\Sigma))}- ( D \chi G_{\rm ret} \phi_{1}| \chi G_{\rm ret} \phi_{2} )_{V(J_{-}(\Sigma))}\\
&=( G_{\rm ret} \phi_{1}| \phi_{2})_{V(J_{-}(\Sigma))}- (  \phi_{1}| G_{\rm ret} \phi_{2} )_{V(J_{-}(\Sigma))}\\
&=( G_{\rm ret} \phi_{1}| \phi_{2})_{V(M)}- ( \phi_{1}| G_{\rm ret} \phi_{2} )_{V(M)}= - (\phi_{1}| G \phi_{2})_{V(M)}.
\eea
\]
In the second line we use the definition of $\tilde{q}$, in the third we use that $\chi\equiv 1$ near $\supp G_{\rm ret}\phi_{i}\cap J_{-}(\Sigma)$ and that $DG_{\rm ret}= \one$ and in the fourth line that $\supp {\phi_{i}}\subset J_{-}(\Sigma)$ and $G_{\rm ret}^{\star}= G_{\rm adv}$. This shows that $\tilde{q}= {q}$ and completes the proof of \eqref{eq:green}. \qed


\subsection{Linear gauge theories}\label{ss:HS} When discussing linear gauge theories it is useful to introduce an abstract framework that captures their general structure. Here we use a special case of the framework proposed by Hack--Schenkel  \cite{HS}, which includes examples such as electromagnetism and linearized Yang--Mills equations, and which also turns out to be well adapted to  linearized gravity.  We follow the presentation in \cite{GW1}  and refer to~\cite{WZ} for the relationship with the BRST formalism. The more general BV formalism on Lorentzian manifolds is discussed in \cite{FR}, cf.~\cite{BFR} for the special case of linearized gravity.

\begin{assumption}\label{as:subsidiary}Suppose that we are given:
\ben
\item two Hermitian bundles $V_{1},V_{2}$ over $M$;
\item a formally self-adjoint operator $P\in\Diff(M;V_{2})$;
\item an operator $K\in\Diff(M;V_{1},V_{2})$, such that $K\neq0$ and 
\begin{enumerate}
\item[\upshape{a)}] $PK=0$,
\item[\upshape{b)}] $D_1\defeq K^\star K\in\Diff(M;V_1)$ is Green hyperbolic,
\item[\upshape{c)}] $D_2\defeq P+K K^\star\in\Diff(M;V_2)$ is Green hyperbolic,
\end{enumerate}
where $K^\star$ is the formal adjoint of $K$  defined as in \ref{sec0.1.1b} using the Hermitian structures on $V_1,V_2$.
\een
\end{assumption}

The operator $P$ is the  operator of direct physical interest (in our case, obtained by linearization of Einstein equations, see Sect.~\ref{sec:lg}). The operator $K$ defines gauge transformations $u\mapsto u+K w$, and the identity $PK=0$ states that $P$ is invariant under gauge transformations. 

Thanks to the assumptions on $D_1,D_2$, the non-hyperbolic equation $Pu=0$ can  be reduced by gauge transformations to the subspace $K^\star u=0$ of solutions of  the hyperbolic problem $D_2 u=0$. The equation $K^\star u=0$ is sometimes called \emph{subsidiary condition}. The term $KK^\star$ in the definition of $D_2 =P+K K^\star$ is the so-called \emph{gauge-fixing term}.

Hypothesis  \ref{as:subsidiary} immediately implies the identities
\beq
K^{\star}P=0, \quad K^\star D_2=D_1 K^\star, \quad D_2 K=K D_1. 
\eeq
We apply the notations from \ref{ss:hyperbolic} to the hyperbolic operators $D_1$ and $D_2$, with the addition of a subscript $i=1,2$ to distinguish between the two. In particular $G_2$ is the causal propagator of $D_2$.

\begin{proposition}[{\cite[Prop.~2.7]{GW1}}]\label{prop:Gpasses} Assume Hypothesis \ref{as:subsidiary}. Then  the induced maps
\[
\begin{aligned} 
{\rm a)}\quad &[G_2]:\,\frac{\Ker K^\star|_{\cinf_{\rm c}}}{\Ran P|_{\cinf_{\rm c}}}\longrightarrow \frac{\Ker P|_{\cinf_{\rm sc}}}{\Ran K|_{\cinf_{\rm sc}}},\\
{\rm b)}\quad &[G_2]:\,\frac{\Ker K^\star|_{\cinf_{\rm c}}}{\Ran P|_{\cinf_{\rm c}}}\longrightarrow\frac{\Ker D_2|_{\cinf_{\rm sc}}\cap\Ker K^\star|_{\cinf_{\rm sc}}}{\Ran G_2 K |_{\cinf_{\rm c}}},
\end{aligned}
\]
are  well defined and bijective. 
\end{proposition}

The  isomorphism $\rm a)$ justifies the use of the first quotient space by saying it is  equivalent to the space of solutions of $Pu=0$, modulo gauge transformations. The isomorphism $\rm b)$ is of practical importance as an intermediary step to get an isomorphism with a space of Cauchy data for the Green hyperbolic operator $D_{2}$. 

Note that in both cases $[G_2]$ maps to a quotient space.  The possibility of choosing different representatives of an element of that quotient is called the \emph{residual gauge freedom}.

\begin{definition}
The \emph{physical phase space} is  the Hermitian space $(\cV_{P}, q_{P})$, where:
\[
\cV_{P}= \frac{\Ker K^\star|_{\cinf_{\rm c}}}{\Ran P|_{\cinf_{\rm c}}}, \quad [\bar{u}]\dual q_{P}[v]= \i (u| G_{2}v)_{V_{2}}, \ \ [u], [v]\in  \frac{\Ker K^\star|_{\cinf_{\rm c}}}{\Ran P|_{\cinf_{\rm c}}}.
\]
\end{definition}

By  \cite[Prop.~2.6]{GW1}, $q_{P}$ is a well-defined Hermitian form on $\cV_P$. 

\subsubsection{Phase space of Cauchy data}
Identifying the correct space of Cauchy data requires several more definitions.

For $i=1,2$ we denote by $U_i$ the operator which assigns to a Cauchy datum  the corresponding solution in $\Ker D_i|_{\cinfsc}$, i.e.~$U_i$ is the inverse of $\varrho_i : \Ker D_i|_{\cinfsc}\to\coinf(\Sigma;V_{i,\Sigma})$.  In other terms,
\beq\label{e00.3}
\begin{cases}
D_{i}U_{i}= 0, \\
\varrho_{i}U_{i}= \one.
\end{cases}
\eeq

To the operator $K$ we associate an operator $K_{\Sigma}\in\Diff(\Sigma;V_{1,\Sigma},V_{2,\Sigma})$ acting on Cauchy data by setting
\beq\label{eq:relc}
K_{\Sigma}= \varrho_2 K U_1.
\eeq
We introduce the notation $K^{\dagger}_{\Sig}\in\Diff(\Sigma;V_{2,\Sigma},V_{1,\Sigma})$ for the adjoint w.r.t.~the Hermitian forms $(\cdot|q_{2}\cdot)_{\Sigma}$ and $(\cdot |q_{1} \cdot)_{\Sigma}$, i.e.
\beq\label{eq:defkdag}
q_1 K^{\dagger}_{\Sigma} \defeq K^{\star}_{\Sigma} q_{2}.
\eeq
The notation $^\dagger$ is used to avoid confusion with the formal adjoint $^\star$ w.r.t.~the Hermitian structures on the bundles $V_{1,\Sigma}$, $V_{2,\Sigma}$.  

\begin{lemma}[{\cite[Lem.~2.9]{GW1}}]\label{lem:cauchyrel}Assume Hypothesis \ref{as:subsidiary}. Then:
\ben
\item\label{cauchyrelit1} $K U_1 = U_2 K_\Sigma$ and $K^\star U_2=U_1 K^{\dagger}_{\Sigma}$;
\item\label{cauchyrelit0} $\varrho_2 K =K_{\Sigma}\varrho_1$ on $\Ker D_1|_{\cinf_{\rm sc}}$ and $\varrho_1 K^\star=K^{\dagger}_{\Sigma}\varrho_2$ on $\Ker D_2|_{\cinf_{\rm sc}}$;
\item\label{cauchyrelit2} $\Ker K^{\dagger}_{\Sigma}|_{\cinf_{\rm c}}=\varrho_2 G_2 \Ker K^\star|_{\cinf_{\rm c}}$;
\item\label{cauchyrelit3} $\Ran K_{\Sigma}|_{\cinf_{\rm c}}=\varrho_2 G_2 \Ran K|_{\cinf_{\rm c}}$;
\item\label{cauchyrelit4} $K^{\dagger}_{\Sigma} K_{\Sigma}=0$.
\een
\end{lemma}

\begin{proposition}[{\cite[Prop.~2.10]{GW1}}]\label{prop:rhopasses} The induced map
\[
[\varrho_2]: \ \frac{\Ker D_2|_{\cinf_{\rm sc}}\cap\Ker K^\star|_{\cinf_{\rm sc}}}{\Ran G_2 K |_{\cinf_{\rm c}}}\longrightarrow\frac{\Ker K^{\dagger}_{\Sig}|_{\cinf_{\rm c}}}{\Ran K_\Sigma|_{\cinf_{\rm c}}}
\]
is well defined and bijective.
\end{proposition}

By  \ref{ss:hyperbolic} applied to $D_2$, we obtain   the  Hermitian form $(\cdot| q_2 \cdot)_{V_{2\Sigma}}$ which  is well defined on  the quotient space $\Ker K^{\dagger}_{\Sig}|_{\cinf_{\rm c}}/ \Ran K_\Sig|_{\cinf_{\rm c}}$. This is an easy consequence of  \eqref{eq:defkdag}: indeed, if $u,v\in \Ker K^{\dagger}_{\Sig}|_{\cinf_{\rm c}}$ then for all $w\in \coinf(\Sigma;V_{1,\Sig})$,
\beq\label{eq:whygi}
(u+ K_{\Sigma}w  | q_2  v)_{\Sigma} =(u   |  q_2  v )_{\Sigma} + (w | q_2 K_\Sig^\dag v)_{\Sigma}= (u   |  q_2  v )_{\Sigma},
\eeq
and similarly for gauge transformations in the second argument. Thus, $(u   |  q_2  v )_{\Sigma}$ depends only on the equivalence classes of $u$ and $v$. 

Summarizing we have:
\begin{proposition}\label{prop-added}
 The map 
 \[
 [\varrho_{2}G_{2}]: \bigg(\frac{\Ker K^\star|_{\cinf_{\rm c}}}{\Ran P|_{\cinf_{\rm c}}}, \i (\cdot | G_{2}\cdot )_{V_{2}}\bigg)\longrightarrow\bigg(\frac{\Ker K^{\dagger}_{\Sig}|_{\cinf_{\rm c}}}{\Ran K_\Sigma|_{\cinf_{\rm c}}},(\cdot| q_{2}\cdot)_{V_{2\Sigma}}\bigg)
 \]
 is pseudo-unitary.
\end{proposition}
%

\subsection{Quantization} The algebraic quantization of linear gauge theories is discussed in detail in \cite[Sect.~3]{GW1}. The algebraic framework reduces the quantization problem to showing the existence of physically relevant quantum states on the  CCR $*$-algebra $\CCR(\cV_{P}, q_{P})$ associated to the Hermitian space $(\cV_{P}, q_{P})$ defined in Sect.~\ref{ss:HS}.   The notions of quasi-free states and covariances (or two-point functions) are explained in  \cite[Sect.~3]{GW1} and references therein.

\subsubsection{Two-point functions}
A quasi-free state on $\CCR(\cV_{P}, q_{P})$ is determined by a pair $\Lambda^{\pm}$ of {\em covariances}, i.e.~of Hermitian forms on $\cV_{P}$  such that \[
\begin{cases}
\Lambda^{\pm}\geq 0,\\
\Lambda^{+}- \Lambda^{-}= q_{P}.
\end{cases}
\]
We will consider quasi-free states $ \omega$ on  $\CCR(\cV_{P}, q_{P})$ with covariances obtained  from a  pair of maps $\lambda_{2}^{\pm}:  \coinf(M; V_{2})\to \cinf(M; V_{2})$ (called  the {\em pseudo-covariances} of $\omega$) by:
 \beq\label{e00.1}
 [\bar{u}]\dual \Lambda^{\pm}[v]= (u| \lambda_{2}^{\pm}v)_{V_{2}}, \ \ [u], [v]\in \frac{\Ker K^\star|_{\cinf_{\rm c}}}{\Ran P|_{\cinf_{\rm c}}}.
 \eeq
The following lemma \cite[Lem.~3.16]{GW1} is straightforward.
\begin{lemma}\label{lemma.had}
 Suppose $\lambda_{2}^{\pm}:  \coinf(M; V_{2})\to \cinf(M; V_{2})$ are such that:
  \beq\label{defodefo}
\begin{aligned}
i)&\quad  D_{2}\lambda^\pm_{2}=\lambda^\pm_{2} D_{2} =0  \, \mbox{  and  } \,  \lambda_2^+-\lambda^-_2 = \i^{-1} G_2,\\
ii) & \quad(\lambda^\pm_{2})^{\star}=\lambda^\pm_{2} \hbox{ for }(\cdot| \cdot)_{V_{2}} \, \mbox{  and  } \, \lambda^\pm_{2}: \ \Ran K|_{\coinf}\to\Ran K|_{\cD'},\\
iii) & \quad \lambda^\pm_{2}  \geq 0  \hbox{ for }(\cdot| \cdot)_{V_{2}}\mbox{ on } \Ker K^{\star}|_{\coinf}.
\end{aligned}
\eeq
Then $\lambda_{2}^{\pm}$ are the pseudo-covariances  of a quasi-free state on $\CCR(\cV_{P}, q_{P})$.\end{lemma}
In fact it is easy to check that conditions ${\it i)}$ and ${\it ii)}$ above imply that $\Lambda^{\pm}$ are well defined on the quotient in \eqref{e00.1}. The name `pseudo-covariance' comes from the fact that $\lambda_{2}^{\pm}$ are not required to be positive for $(\cdot| \cdot)_{V_{2}}$ on $\coinf(M; V_{2})$, but only on  the subspace $\Ker K^{\star}|_{\coinf}$.
 
\subsubsection{Hadamard condition} \label{sss:had}
We use  the following definition of Hadamard states \cite{SV},  cf.~\cite[Subsect.~3.4]{GW1} and references therein. The general consensus is that only states satisfying the  \emph{Hadamard condition} (the \emph{Hadamard states}) are  physical. 
We recall that
 $$
 \cN=\{ (x,\xi)\in T^*M\setminus\zero \st \xi\cdot \rg^{-1}(x)\xi = 0 \}
 $$
is  the  characteristic set of the wave operator on $(M,\rg)$, and 
$$
\cN^{\pm}=\cN\cap \{(x, \xi)\in T^*M\setminus\zero  \st \pm v\dual \xi>0 \,\ \forall v\in T_{x}M\hbox{ future-directed time-like}\}
$$
are its  two connected components.

\begin{definition}\label{defhadama}
  A quasi-free state $\omega$ on $\CCR(\cV_{P}, q_{P})$ given by  pseudo-covariances $\lambda_{2}^{\pm}$ as in \Lemma \ref{lemma.had} is {\em Hadamard} if in addition to \eqref{defodefo} it satisfies:
\beq\label{defidefi}
 \WF(\lambda_{2}^{\pm})'\subset \cN^{\pm}\times \cN^{\pm}.
 \eeq
\end{definition}

\subsubsection{Hadamard condition on a Cauchy surface}
One can equivalently consider pseudo-covariances $\lambda^{\pm}_{2\Sigma}$ for $D_{2}$ acting on Cauchy data on a Cauchy surface $\Sigma$, see e.g.~\cite[Subsect.~3.3]{GW1}. Namely  we can set
 \begin{equation}
 \label{e00.2}
 \lambda_{2\Sigma}^{\pm}\defeq \varrho_{2}^{\star}q_{2} \lambda_{2}^{\pm}\varrho_{2}^{\star}q_{2},
 \end{equation}
 and then one also has
 \beq\label{e00.2b}
 \lambda_{2}^{\pm}= (\varrho_{2}G_{2})^{\star}\lambda_{2\Sigma}^{\pm}(\varrho_{2}G_{2}).
 \eeq
 The maps $\lambda^{\pm}_{2\Sigma}$ are correspondingly called {\em Cauchy surface pseudo-covariances}. 
Since $q_{2}$ is non-degenerate, we can set
\beq\label{e00.4}
\lambda^{\pm}_{2\Sigma}\eqdef \pm q_{2} c_{2}^{\pm},
\eeq
and  we now formulate conditions on  $c_2^\pm$ which imply that $\lambda^{\pm}_{2}$ satisfy the conditions in \Lemma \ref{lemma.had} and Def. \ref{defhadama}.
\begin{proposition}\label{prop:states1} Suppose $c_2^\pm: \coinf(\Sigma;V_{2,\Sig})\to \cinf(\Sigma;V_{2,\Sig})$ is a pair of operators such that:
\ben
\item\label{it:con} $c_2^++c_2^-=\one$;
\item \label{it:selfadj}$q_{2}c_{2}^{\pm}= (q_{2}c_{2}^{\pm})^{\star}$ for the scalar product $(\cdot| \cdot)_{V_{2\Sigma}}$;
\item\label{it:pos} $\pm (f|q_2 c^\pm_2 f)_{\Sigma}\geq 0$  for all $f \in \Ker K_\Sig^\dagger|_{\coinf}$;
\item\label{it:gi} $c_2^\pm K_\Sigma= K_{\Sigma}c_1^\pm$ for some $c_1^\pm : \coinf(\Sigma;V_{1,\Sig})\to \cinf(\Sigma;V_{1,\Sig})$.
\een\smallskip
Then $\lambda_{2}^{\pm}$ given by \eqref{e00.2b} and \eqref{e00.4} are the pseudo-covariances of a quasi-free state on $\CCR(\cV_{P}, q_{P})$. Furthermore, suppose that the principal symbol of  $D_2$ is $\xi\dual \rg^{-1}(x)\xi\,\one_{V_2}$. Then, if for some neighborhood $\cU$ of $\Sigma$ in $M$ we have: \smallskip
\ben\setcounter{smallarabics}{4}
\item \label{it:had} $\WF(U_{2}\circ c_{2}^{\pm})'\subset (\cN^{\pm}\cup\cF)\times T^{*}\Sigma$ over $\cU\times \Sigma$,
where $\cF\subset T^{*}M$ is a conic set with $\cF\cap \cN= \emptyset$, 
\een
then the associated state  is Hadamard. 
\end{proposition}

Let us briefly outline the physical significance of conditions \eqref{it:con}--\eqref{it:gi}.   Condition \eqref{it:gi} is interpreted as \emph{gauge-invariance} of the state. It ensures that $\pm(  \cdot | q_2  c_2^\pm \cdot  )_{\Sigma}$ is well-defined on the quotient space $\Ker K^{\dagger}_{\Sig}|_{\cinf_{\rm c}}/ \Ran K_\Sig|_{\cinf_{\rm c}}$ by an argument analogous to \eqref{eq:whygi}. 
Condition \eqref{it:con} expresses the  \emph{canonical commutation relations}, while \eqref{it:selfadj} and \eqref{it:pos} are responsible for  the  \emph{positivity} of the state. We point out that positivity is  only required to hold   on the physical space $\Ker K_\Sig^\dagger|_{\coinf}$.

 Our objective in the next chapters will be to show that conditions \eqref{it:con}--\eqref{it:had} can be satisfied simultaneously in the case of linearized gravity. In view of Prop.~\ref{prop:states1}  this will imply the existence of a Hadamard state.  

\noindent {\bf Proof of Prop.~\ref{prop:states1}}.
The proof that  \eqref{it:con}--\eqref{it:gi} imply \eqref{defodefo} is straightforward. Let us  prove  that \eqref{it:had} implies \eqref{defidefi}, removing the subscripts for ease of notation.  We know that the  operator $U$ defined in \eqref{e00.3} can be expressed as:
$U= \i^{-1}(\varrho G)^{\star}q$,
hence $\lambda^{\pm}= \pm \i^{-1}U c^{\pm}\circ (\varrho G)$. Note that we are allowed to compose the kernels $Uc^{\pm}$ and $\varrho G$ since $\varrho G: \coinf(M; V)\to \coinf(\Sigma; V_{\Sigma})$. It is well known  that  $\WF(G)'\subset \cN\times \cN$. Using also   \eqref{it:had} and \cite[Thm.~8.2.14]{H} we obtain
  \beq\label{e9.00b}
 \WF(\lambda^{\pm})'\subset\big( (\cN^{\pm}\cup \cF)\times \cN\big) \cup\big( (\cN^{\pm}\cup \cF)\times \zero\big) \hbox{ over }U\times M,
 \eeq
 where we recall that $\zero\subset T^{*}M$ is the zero section.  Condition \eqref{it:selfadj} implies that 
 $\lambda^{\pm}= \lambda^{\pm\star}$  for $(\cdot| \cdot)_{V}$  and hence we obtain   that $(X, X')\in \WF(\lambda^{\pm})'$ if and only if $(X', X)\in \WF(\lambda^{\pm})'$. Using that $ \cF\cap \cN= \emptyset$, we then deduce from \eqref{e9.00b} that
\[
  \WF(\lambda^{\pm})'\subset (\cN^{\pm}\times \cN^{\pm} )\cup ( \cN^{\pm}\times \zero) \cup ( \zero\times \cN^{\pm})\hbox{ over }\mathcal{U}\times M.
 \]
 Since $\lambda^{+}- \lambda^{-}= \i G$,   using once more that $\WF(G)'\subset\cN\times \cN$ and the hypothesis $\cF\cap \cN= \emptyset$, this implies that 
 \[
 \WF(\lambda^{\pm})'\cap\big(  (\cN^{\pm}\times \zero) \cup ( \zero\times \cN^{\pm})\big)= \emptyset,
 \]
 which proves   \eqref{defidefi} over $\mathcal{U}\times M$. To extend \eqref{defidefi}  to $M\times M$ we use that  $D\Lambda^{\pm}=0$ and argue by propagation of singularities as in \cite[Lem.~6.5.5]{DH}.  \qed
\section{Analytic spacetimes of bounded geometry}\label{an-bg}\init

 The goal of this section is to formulate the hypotheses on the background spacetime $(M,\rg)$ that we  use to prove the existence of Hadamard states for linearized gravity. 

Roughly speaking the metric $\rg$ should be analytic in a time coordinate, the most natural one being the \emph{Gaussian time} associated to the normal geodesic flow to a Cauchy surface $\Sigma$. If the Cauchy surface $\Sigma$ is not compact, we also need global estimates  on the metric $\rg$ in the `space variables' on $\Sigma$. A convenient way to formulate these estimates is using the language of bounded geometry, recalled in Subsect.~\ref{an-bg.1}.

In practice the analyticity in time is easier to deduce from an analyticity condition in all variables, which leads to the notion of spacetimes of bounded analytic geometry, defined in Subsects.~\ref{an-bg.2}, \ref{an-bg.2.5}. 

Existence of many Einstein manifolds of bounded analytic geometry is proved in Subsect.~\ref{an-bg2.6}. Finally in Subsect.~\ref{an-bg2.7} we give concrete examples, such as the Kerr--Kruskal and maximal Schwarzschild--de Sitter spacetimes.

\subsection{Spacetimes of bounded geometry}\label{an-bg.1}
 \subsubsection{Riemannian manifolds of bounded geometry}\label{an-bg.1.1}

 We recall that an $n$-dimensional  Riemannian manifold $(M, \hat{g})$ is of {\em bounded geometry} if its injectivity radius $r_{\hat{g}}$ is strictly positive and $\nabla^{k}\rR$ are bounded tensors, where $\rR$ is the Riemann curvature tensor and $\nabla$ the covariant derivative associated to $\hat{g}$, see  e.g.~\cite[App.~1]{Sh}.
  
Although we are mainly interested in Lorentzian manifolds $(M, \rg)$,   we will need an  auxiliary  Riemannian metric $\hat{g}$ of bounded geometry to define various function spaces, like spaces of bounded tensors, Sobolev spaces, etc.

  An equivalent characterization (see e.g.~\cite[Thm.~2.2]{GOW})is given in Prop. \ref{prop1.1} below: let us denote by $B_{n}(0, 1)$ the unit ball in $\rr^{n}$, by $\delta$ the flat metric on $\rr^{n}$  and by $\BT^{p}_{q}(B_{n}(0, 1),\delta)$ the space of $(q,p)$-tensors on $B_{n}(0, 1)$ which are bounded on $B_{n}(0, 1)$ together with all their derivatives. 
  
  For $p=q=0$, ie for functions, we use often the notation $\cinfb(B_{n}(0, 1))$ instead of  $\BT^{0}_{0}(B_{n}(0, 1),\delta)$.
  
  Then $\BT^{p}_{q}(B_{n}(0, 1),\delta)$ is a Fr\'echet space and one can hence define  a bounded family of $(q, p)$-tensors in $\BT^{p}_{q}(B_{n}(0, 1),\delta)$. For example, a family $(f_{i})_{i\in I}$ of functions on $B_{n}(0, 1)$ is bounded if $\sup_{i\in I, x\in B_{n}(0, 1)}| \p_{x}^{\alpha}f_{i}(x)|<\infty$ for all $\alpha\in \nn^{n}$.

\begin{proposition}\label{prop1.1}
 $(M, \hat{g})$ is of bounded geometry if and only if for each $x\in M$ there exists $\mathcal{U}_{x}\subset M$ open neighborhood of $x$ and
\[
\psi_{x}: \mathcal{U}_{x} \xrightarrow{\sim} B_{n}(0,1)
\]
a  diffeomorphism with $\psi_{x}(x)=0$ such that if $\hat{g}_{x}\defeq   (\psi_{x}^{-1})^{*}\hat{g}$ then:

\ben
\item  the family $\{\hat{g}_{x}\}_{x\in M}$ is  bounded in $\BT^{0}_{2}(B_{n}(0,1), \delta)$,
\item there exists $c>0$ such that:
\[
c^{-1}\delta\leq \hat{g}_{x}\leq c \delta, \  \ x\in M.
\]
\een
\end{proposition}
It is known that one can find  a sequence $(x_{i})_{i\in \nn}$ of points in $M$ such that setting $(\mathcal{U}_{i}, \psi_{i})\defeq (\mathcal{U}_{x_{i}}, \psi_{x_{i}})$,  $(\mathcal{U}_{i}, \psi_{i})_{i\in \nn}$ is an atlas of $M$ with the additional property that there exists $N\in \nn$ such that $\bigcap_{i\in J}\mathcal{U}_{i}= \emptyset$ if $\sharp J>N$. Such atlases are called {\em bounded atlases} of $M$. 

The standard choice for  $(\mathcal{U}_{i}, \psi_{i})$ is $(B^{\hat{g}}(x_{i}, r), \exp^{\hat{g}}_{x_{i}})$ for $0<r<r_{\hat{g}}$, i.e.~the geodesic ball and exponential map at $x_{i}$. One can associate to a bounded atlas  $(\mathcal{U}_{i}, \psi_{i})_{i\in \nn}$ a partition of unity
\[
1= \sum_{i\in \nn}\chi_{i}^{2}, \ \ \chi_{i}\in \coinf(\mathcal{U}_{i}),
\]
such that $\{(\psi_{i}^{-1})^{*}\chi_{i}\}_{i\in \nn}$ is a bounded sequence in $\cinf_{\rm b}(B_{n}(0,1))$. Such a partition of unity is  called a {\em bounded partition of unity}.

\subsubsection{Bounded tensors, bounded differential operators, Sobolev spaces}\label{an-bg.1.2}
We recall now several notions  due to Shubin \cite{Sh}, see also \cite[Subsect.~2.3]{GOW}.

\begin{notations}
\item  If $(M, \hat{g})$ is a Riemannian manifold of bounded geometry,  we denote by  $\BT^{p}_{q}(M, \hat{g})$ the space of {\em bounded $(q, p)$-tensors} on $M$.  Concretely $T\in \BT^{p}_{q}(M, \hat{g})$ if, given a bounded atlas $\{(\mathcal{U}_{i}, \psi_{i})\}_{i\in \nn}$, the seminorms of the push-forwards of $T$ to $\mathcal{U}_{i}$ by $\psi_{i}$ in $\BT^{p}_{q}(B_{n}(0, 1), \delta)$ are bounded uniformly in $i\in \nn$. The definition is independent on the choice of the bounded atlas.
\item As above we will sometimes use the notation $\cinfb(M)$ for $\BT^{0}_{0}(M, \hat{g})$.
\item If $\mathcal{U}\subset M$ is an open set, we define similarly the space $\BT^{p}_{q}(\mathcal{U}, \hat{g})$ by requiring the above uniform bound only over the $i$ such that $\mathcal{U}\cap \mathcal{U}_{i}\neq \emptyset$.
\item If $I\subset \rr$ is an interval   we use the notation introduced in \ref{sec01.1.3} to define the spaces $\cinfb(I, \BT^{p}_{q}(M, \hat{g}))$.
\item The space  ${\rm Diff}(M, \hat{g})$ is the space of \emph{bounded differential operators} on $M$, i.e.~differential operators which form a bounded family of differential operators on $B_{n}(0, 1)$ when expressed in a bounded atlas of $M$. 
\item Finally,  we denote by  $H^{s}(M, \hat{g})$ the Sobolev space of order $s\in \rr$.
\end{notations}

\subsubsection{Vector bundles of bounded geometry}\label{an-bg.1.3}
Let $(M, g)$ be a Riemannian manifold of bounded geometry.    We recall  the definition of vector bundles of bounded geometry, see \cite{Sh}.  

A vector bundle $E\xrightarrow{\pi}M$ of rank $N$  is of {\em bounded geometry} if there exists a bounded covering $(\mathcal{U}_{i})_{i\in \nn}$ of $M$  which forms a bundle atlas of $E$ such that the transition maps $t_{ij}: \mathcal{U}_{ij}\to M_{N}(\cc)$ are a bounded family of matrices.

Clearly if $(M, \hat{g})$ is of bounded geometry, then all tensor bundles over $M$ (like  the bundles $\otimes_{\rm s}^{k}T^{*}M$ that will be used in Sect.~\ref{sec:lg}) are also of bounded geometry.

The space of bounded sections of $E$ (analogous to the spaces of bounded tensors) is denoted by $\cinfb(M; E)$.  More precisely, if $u$ is a section of $E$, we denote by $u_{i}: \mathcal{U}_{i}\to \cc^{N}$ its local trivializations over $\mathcal{U}_{i}$, and then $u\in \cinfb(M; E)$ if and only if the family $(u_{i})_{i\in \nn}$ is bounded in $\cinfb(\mathcal{U}_{i}; \cc^{N})$.

The notion of bounded differential operators acting on smooth sections of  $E$ can  now defined similarly as before, and so does the notion of bounded Hermitian forms on (the fibers of) $E$.  
\subsubsection{Bounded Hilbert space structures}
Using a partition of unity one can equip a vector bundle $E\xrightarrow{\pi}M$  of bounded geometry with a positive definite bounded Hermitian form $\beta$, (i.e. a Hilbert space structure on the fibers of $E$).  This means that  if $\beta_{i}: U_{i}\to L_{\rm h}(\cc^{N}, \cc^{N*})$ are its local trivializations, then $\beta_{i}>0$ and the families $(\beta_{i})_{i\in \nn}$ and $(\beta_{i}^{-1})_{i\in \nn}$ are bounded in $\cinfb(U_{i}; L_{\rm h}(\cc^{N}, \cc^{N*}))$ and $\cinfb(U_{i}; L_{\rm h}(\cc^{N*}, \cc^{N}))$.

Any two of these bounded Hilbert space structures are equivalent, in the sense of the bounded geometry.

\subsubsection{Sobolev spaces}\label{an-bg.1.4} For $s\in\rr$,
one defines the Sobolev space $H^{s}(M; E)$ in the natural way, for example using the norm
\[
\| u\|_{s}^{2}= \sum_{i\in \nn}\| T_{i}\circ \psi_{i}\chi_{i}u\|^{2}_{s},
\]
where $(\mathcal{U}_{i}, \psi_{i})_{i\in \nn}$ is a bounded atlas, $T_{i}: \pi^{-1}(\mathcal{U}_{i}) \to \mathcal{U}_{i}\times \cc^{N}$ are local trivializations, $1= \sum \chi_{i}^{2}$ is a bounded partition of unity subordinate to $(\mathcal{U}_{i})_{i\in \nn}$ and $\norm{\cdot}_{s}$ is the usual Sobolev norm on $H^{s}(\rr^{n}; \cc^{N})$.

The topology of $H^{s}(M; E)$ is independent of the above choices.

 \subsubsection{Spacetimes of bounded geometry}\label{an-bg.1.5}

 \begin{definition}\label{defp3.1}
 Let $(M, \rg)$ be a globally hyperbolic spacetime and $\Sigma\subset M$ a smooth  spacelike Cauchy surface.  We say that $(M, \rg)$ is of \emph{bounded geometry near $\Sigma$} if the following conditions hold:
 \ben
 
 \item there exists a reference Riemannian metric $\hat{g}$ such that $(M, \hat{g})$ is of bounded geometry;
 \item there exists an open neighborhood $\mathcal{U}$ of $\Sigma$ in $M$ such that $\rg\in {\BT}^{0}_{2}(\mathcal{U}, \hat{g})$ and $\rg^{-1}\in {\BT}^{2}_{0}(\mathcal{U}, \hat{g})$;
 \item  the embedding  $ \Sigma\hookrightarrow \mathcal{U}$ is of bounded geometry for $\hat{g}$, i.e.~there exists a bounded atlas $\{(\mathcal{U}_{i}, \psi_{i})\}_{i\in \nn}$ for $M$ such that if $\Sigma_{i}= \psi_{i}(\Sigma\cap \mathcal{U}_{i})$, we have 
 \[
 \Sigma_{i}= \{(v', v_{n})\in B_{n}(0, 1) \st v_{n}= F_{i}(v')\},
 \]
 where the seminorms of $F_{i}$ in $\BT^{0}_{0}(B_{n}(0, 1), \delta)$ are uniformly bounded for $i\in \nn$ such that $\mathcal{U}_{i}\cap \mathcal{U}\neq \emptyset$;
 \item if $n(y)$ for  $y\in \Sigma$ is the future directed unit normal  for $ \rg$ to $\Sigma$, one has:
 \[
\sup_{y\in \Sigma}\,n(y)\cdot \hat{g}(y)n(y)<\infty.
\]
\een
\end{definition}
 If $(M, \rg)$ is of bounded geometry near $\Sigma$,  then the Gaussian normal coordinates to $\Sigma$ are well adapted to the bounded geometry framework. We recall a  result in this direction, see \cite[Thm.~3.5]{GOW}.
\begin{theorem}\label{th-omar}
 Let $(M,  \rg)$ be a Lorentzian manifold of bounded geometry near a Cauchy surface $\Sigma$. Then the following holds:
\ben
\item there exists $\delta>0$ such that the normal geodesic flow to $\Sigma$:
\[
\chi:\ \begin{array}{l}
\open{-\delta, \delta}\times \Sigma\to M\\[2mm]
(t, \rx)\mapsto \exp_{\rx}^{ \rg}(tn(\rx))
\end{array}
\] 
is well defined and is a smooth diffeomorphism on its range;
\item $\chi^{*} \rg= - dt^{2}+\rh_{t}$, where $\{ \rh_{t}\}_{t\in\open{-\delta, \delta}}$ is a smooth family of Riemannian metrics on $\Sigma$ such that:
\[
\begin{array}{rl}
i)&(\Sigma,  \rh_{0})\hbox{ is of bounded geometry},\\[2mm]
ii)& t\mapsto  \rh_{t}\in \cinf_{\rm b}(\open{-\delta, \delta}, \BT^{0}_{2}(\Sigma,  \rh_{0})), \\[2mm]
iii)& t\mapsto  \rh^{-1}_{t}\in \cinf_{\rm b}(\open{-\delta, \delta}, \BT^{2}_{0}(\Sigma,  \rh_{0})).
\end{array}
\]
 \een
\end{theorem}
We recall that the spaces $\cinfb(I, \BT^{p}_{q}(M, \hat{g}))$ are defined in \ref{an-bg.1.2}.
\subsection{Analyticity in  Gaussian time}\label{an-bg.0}
Let $(M, \rg)$ be a Lorentzian spacetime of bounded geometry near a Cauchy surface $\Sigma$ with respect to a reference Riemannian metric $\hat{g}$ on $M$.  Then by Thm.~\ref{th-omar}, we can assume that $M= I\times \Sigma$ for $I$ an open interval and $\rg= - dt^{2}+\rh_{t}$, the function $t: M\to \rr$ being called the {\em Gaussian time} associated to $\Sigma$.

In later sections, we will perform the {\em Wick rotation} in $t$, corresponding to replacing $t$ by $\i s$, which requires that the metric $\rg$ is real analytic in $t$. We will need the property that $\rg$ extends holomorphically  in $t$ in a  {\em strip} in the complex plane, with estimates in the  space variables $\rx$  adapted to the bounded geometry of $(\Sigma, \rh_{0})$.

To formulate precisely our hypotheses we first introduce the relevant Banach spaces of bounded analytic functions.

\subsubsection{Bounded analytic functions}\label{an-bg.0.1}
Let $D_{1}(0, r)\subset \cc$ be the open disk of radius $r$ and $C_{1}(0, r)= \open{-r, r}$ its intersection with the real line.  We denote by $\mathcal{A}(C_{1}(0, r))$ the Banach space of functions $u\in \cinf(C_{1}(0, r))$ such that
 \[
\|u\|_{r}\defeq\sup_{\alpha\in \nn^{n}} |r^{|\alpha|}(\alpha!)^{-1}\p_{x}^{\alpha}u(0)|<\infty.
\]
The space $\mathcal{A}(C_{1}(0, r))$ coincides with the Banach space of analytic functions on $C_{1}(0, r)$ which extend holomorphically to $D_{1}(0, r)$ with
\[
\|u\|_{r,1}\defeq\sup_{z\in D_{1}(0, r)}|(1- r^{-1}|z|)u(z)|<\infty,\ \ m_{r}(z)= (1- r^{-1}|z|).
\]
The norms $\norm{\cdot}_{r}$ and $\norm{\cdot}_{r, 1}$ are equivalent, as follows easily from the Cauchy integral formula.

As in \ref{an-bg.1.2} it is straightforward to extend this definition to functions with values in a Fréchet space. Namely, if $\cF$ is a Fréchet space whose topology is defined by a family of seminorms $\norm{ \cdot }_{n}$, $n\in \nn$, we denote by $\mathcal{A}(C_{1}(0, r), \cF)$ the space of maps $f: C_{1}(0, r)\to \cF$ such that 
\[
\sup_{\alpha\in \nn^{n}} \|r^{|\alpha|}(\alpha!)^{-1}\p_{x}^{\alpha}f(0)\|_{n}<\infty, 
\]
for all $n\in \nn$.  Equipped with the obvious seminorms it is again  a Fréchet space.

In particular, we use this convention to define the spaces $\mathcal{A}(C_{1}(0, r); \BT^{p}_{q}(M, \hat{g}))$.
 
\subsection{Main hypotheses}\label{hippopotamus}  We can now state precisely our hypotheses on  the Einstein metric $\rg$.

\begin{assumption} We assume that:
\ben
\item[{\upshape I.}]  $(M, \rg)$ is a Lorentzian manifold of bounded geometry near a Cauchy surface $\Sigma$.

\item[{\upshape II.}] The map $\open{-\delta , \delta}\ni t \mapsto  \rh_{t}$ in Thm.~\ref{th-omar} belongs to $\mathcal{A}(C_{1}(0, r); \BT^{0}_{2}(\Sigma,  \rh_{0}))$ for some $0<\delta\ll 1$.
\een
\end{assumption}

\begin{remark}
While it is rather easy to check I  in examples,   it is more difficult to check II. In the next subsection, we will show that it follows from a stronger condition of `bounded analyticity' in all variables.
\end{remark}
\subsection{Manifolds of bounded analytic geometry}\label{an-bg.2}
We now define the analog of the notions in Subsect.~\ref{an-bg.1} in the analytic category.

\subsubsection{Banach spaces of analytic tensors}\label{an-bg.2.1}
We start by defining the analog of the spaces  $BT^{p}_{q}(B_{n}(0, 1), \delta)$ in the analytic case, by extending  the definitions in \ref{an-bg.0.1} to several variables.

For well-known reasons, it is convenient to work with polydisks instead of balls. Accordingly, we denote by $D_{n}(0, r)\subset \cc^{n}$ the polydisk $D_{n}(0, r)= \{z\in \cc^{n} \st |z_{i}|<r\}$ and  by  $C_{n}(0, r)$ the cube $D_{n}(0, r)\cap \rr^{n}$.  We set $m_{r}(z)= \prod_{i=1}^{n}(1- r^{-1}|z_{i}|)$  for $z\in D_{n}(0, r)$.
 
We denote by  $\mathcal{A}(C_{n}(0, r))$  the Banach space   of functions $u\in \cinf(C_{n}(0,r))$ such that
 \[
\|u\|_{r}\defeq\sup_{\alpha\in \nn^{n}} |r^{|\alpha|}(\alpha!)^{-1}\p_{x}^{\alpha}u(0)|<\infty.
\]
 $\mathcal{A}(C_{n}(0, r))$ coincides with the Banach space of analytic functions on $C_{n}(0, r)$ which extend holomorphically to $D_{n}(0, r)$ with
\[
\|u\|_{r,1}\defeq\sup_{z\in D_{n}(0, r)}|m_{r}(z)u(z)|<\infty,
\]
and the norms $\norm{\cdot}_{r}$ and $\norm{\cdot}_{r, 1}$ are again equivalent.

 We denote by $\AT^{p}_{q}(C_{n}(0, r), \delta)$ the Banach space of analytic $(q, p)$-tensors on $C_{n}(0, r)$, equipped with the canonical norm obtained from the metric $\delta$ and the norm $\norm{\cdot}_{r}$ on $\mathcal{A}(C_{n}(0, r))$. 

By Cauchy estimates $\AT^{p}_{q}(C_{n}(0, r), \delta)$ injects continuously into $\BT^{p}_{q}(\mathcal{U}, \delta)$ if $\mathcal{U}\Subset C_{n}(0, r)$.

\subsubsection{Riemannian manifolds of bounded analytic geometry}\label{an-bg.2.2}
 Let $(M, \hat{g})$ be a Riemannian manifold.  Assume that $M$ has been given the structure of a real analytic manifold, compatible with its $C^{\infty}$ structure, and that $\hat{g}$ is an analytic metric.
 
 \begin{definition}\label{def1.1}
 $(M,\hat{g})$ is of {\em bounded analytic geometry} if for each $x\in M$ there exists an open neighborhood  $\mathcal{U}_{x}$  of $x$ and
 \[
\psi_{x}: \mathcal{U}_{x}\bij C_{n}(0,1)
\]
an analytic diffeomorphism with $\psi_{x}(x)=0$ such that if $\hat{g}_{x}\defeq (\psi_{x}^{-1})^{*}\hat{g}$ then:
\ben
\item the family $\{\hat{g}_{x}\}_{x\in M}$ is  bounded in $\AT^{0}_{2}(C_{n}(0, 1), \delta)$,
\item  there exists $c>0$ such that
\[
c^{-1}\delta\leq \hat{g}_{x}\leq c \delta, \ x\in M.
\]

\een
A family $\{(\mathcal{U}_{x}, \psi_{x})\}_{x\in M}$ as above will be called a {\em bounded analytic atlas} of $M$.
\end{definition}

\subsubsection{Bounded analytic tensors}\label{an-bg.2.3}
  If $(M, \hat{g})$ is a Riemannian manifold of bounded analytic geometry, one can naturally define spaces of bounded analytic tensors.
  
\begin{definition}\label{defp0.2}
 Let $(M,\hat{g})$  be of  bounded analytic geometry. We denote by $\AT^{p}_{q}(M,\hat{g})$ the spaces of  smooth $(q,p)$-tensors $T$ on $M$ such that for  a bounded analytic atlas $\{(\mathcal{U}_{x}, \psi_{x})\}_{x\in M}$, the family  $\{(\psi_{x}^{-1})^{*}T\}_{x\in M}$  is  bounded in $\AT^{q}_{p}(C_{n}(0, \epsilon), \delta)$ for some $0<\epsilon\leq 1$. We equip $\AT^{p}_{q}(M, \hat{g})$ with the Banach space topology given by $\|T\|= \sup_{x\in M}\| T_{x}\|$, where $T_{x}= (\psi_{x}^{-1})^{*}T$ and $\|T_{x}\|$ is the norm of $T_{x}$ in $\AT^{p}_{q}(C_{n}(0, \epsilon), \delta)$.  
 \end{definition}
One can show that the spaces $\AT^{p}_{q}(M,\hat{g})$ are independent on the choice of the bounded analytic atlas $\{(\mathcal{U}_{x}, \psi_{x})\}_{x\in M}$, see Subsect.~\ref{app.1}.

If $\mathcal{U}\subset M$ is an open set, the spaces $\AT^{p}_{q}(\mathcal{U}, \hat{g})$ are defined as in \ref{an-bg.1.2}.

Noting that $C_{1}(0, \alpha)= \open{-\alpha, \alpha}$, the spaces $\mathcal{A}(\open{-\alpha, \alpha}; \AT^{q}_{p}(M, \hat{g}))$ are defined as in \ref{an-bg.0.1}.


 \subsection{Spacetimes of bounded analytic geometry}\label{an-bg.2.5}
 \begin{definition}\label{defp3.1anal}
 Let  $(M, \rg)$  be a globally hyperbolic analytic spacetime and $\Sigma\subset M$ an analytic  spacelike Cauchy surface. 
  We say that $(M, \rg)$  is of {\em bounded analytic geometry near} $\Sigma$ if  the following conditions hold:
 \ben
 \item there exists a reference analytic Riemannian metric $\hat{g}$ such that $(M, \hat{g})$ is of bounded geometry;
 \item there exists an open neighborhood $\mathcal{U}$ of $\Sigma$ in $M$ such that  $\rg\in \AT^{0}_{2}(\mathcal{U}, \hat{g})$ and $\rg^{-1}\in \AT^{2}_{0}(\mathcal{U}, \hat{g})$;
  \item  the embedding $\Sigma\hookrightarrow \mathcal{U}$ is analytic of bounded geometry for $\hat{g}$, i.e.~there exists a bounded analytic atlas 
$\{(\mathcal{U}_{i}, \psi_{i})\}_{i\in \nn}$ for $M$ such that if $\Sigma_{i}= \psi_{i}(\Sigma\cap \mathcal{U}_{i})$, we have 
 \[
 \Sigma_{i}= \{(v', v_{n})\in B_{n}(0, 1)  \st v_{n}= F_{i}(v')\},
 \]
 where the seminorms of $F_{i}$ in $\AT^{0}_{0}(B_{n}(0, 1), \delta)$ are uniformly bounded for $i\in \nn$ with $\mathcal{U}_{i}\cap \mathcal{U}\neq \emptyset$;
 \item if $n(y)$ for  $y\in \Sigma$ is the future directed unit normal  for $ \rg$ to $\Sigma$ one has:
 \[
\sup_{y\in \Sigma}\,n(y)\cdot \hat{g}(y)n(y)<\infty.
\]
 \een
\end{definition}
 \begin{theorem}\label{th-omar-analytic}
 Let $(M, \rg)$ be a Lorentzian manifold of bounded analytic geometry near a smooth spacelike  Cauchy surface $\Sigma$. Then the following holds:
\ben
\item there exists $\delta>0$ such that the normal geodesic flow to $\Sigma$:
\[
\chi:\ \begin{array}{l}
\open{-\delta, \delta}\times \Sigma\to M\\
(t, \rx)\mapsto \exp_{\rx}^{\rg}(tn(\rx))
\end{array}
\] 
is well defined and is a smooth diffeomorphism on its image;
\item  there exists $\epsilon_{1}>0$ such that $\chi^{*}\rg= - dt^{2}+\rh_{t}$, where $\{\rh_{t}\}_{t\in\open{-\epsilon_{1}, \epsilon_{1}}}$ is a smooth family of Riemannian metrics on $\Sigma$ such that:
\[
\begin{array}{rl}
i)&(\Sigma, \rh_{0})\hbox{ is of bounded analytic geometry},\\[2mm]
ii)&t\mapsto \rh_{t}\in \mathcal{A}(\open{-\epsilon_{1}, \epsilon_{1}}; \AT^{0}_{2}(\Sigma, \rh_{0})), \\[2mm]
iii)& t\mapsto \rh^{-1}_{t}\in \mathcal{A}(\open{-\epsilon_{1}, \epsilon_{1}}; \AT^{2}_{0}(\Sigma, \rh_{0})).
\end{array}
\]
 \een
\end{theorem}
The proof will be given in Appendix \ref{app.1}.

From Thm.~\ref{th-omar-analytic} we immediately obtain the following corollary.
\begin{corollary}\label{corollaro}
 Assume that $(M, \rg)$ is of bounded analytic geometry near a Cauchy surface $\Sigma$.  Then $(M, \rg)$ satisfies the hypotheses in Subsect.~\ref{hippopotamus}.
\end{corollary}

\subsection{Einstein manifolds of bounded analytic geometry}\label{an-bg2.6}

Let $\Sigma$ be a smooth $3$-dimensional manifold and let $\Lambda\in \rr$.   The initial data on $\Sigma$ for the non-linear Einstein equations 
\[
\Ric_{ab}- \Lambda g_{ab}=0
\]
 are the  Riemannian metric  $\rh$  induced by  $\rg$ on $\Sigma$,   and  the second fundamental form $\rk$,  i.e.~the symmetric $(0, 2)$-tensor on $\Sigma$ defined by
$$
u\dual \rk v= \nabla_{u}n\dual \rg v,  \ \ u, v\in T\Sigma, 
$$
{with} $n$  the  forward unit normal to $\Sigma$. Above, $\Ric$ is the Ricci tensor of $g$ and $\Lambda$ the cosmological constant.

They have to satisfy the constraint equations, see e.g.~\cite[Prop.~13.3]{R}:
\beq\label{e.constr}
\begin{cases}
{\rm Scal}(\rh)-{\rm tr} ((\rk\rh^{-1})^{2})+{\rm tr}(\rk\rh^{-1})^{2}=2\Lambda,\\[2mm]
\div (\rk - {\rm tr}(\rk\rh^{-1})\rh)=0.
\end{cases}
\eeq
We give below conditions on $(\Sigma, \rh, \rk)$
that imply that $\rg$ is of bounded analytic geometry near the initial surface $\Sigma$.

\begin{theorem}\label{cauchykow}
Suppose that $(\Sigma, \hat{h}_{0})$  is a Riemannian manifold of analytic bounded geometry, and that
\[
\rh, \rk\in \AT^{0}_{2}(\Sigma, \hat{h}_{0}),  \ \ \rh^{-1}\in \AT^{2}_{0}(\Sigma, \hat{h}_{0}),
\]
satisfy the constraint equations \eqref{e.constr}. Then there exists $\delta>0$   and a solution $\rg$ of the Einstein equations
\beq\label{eq:einsteen}
\Ric - \Lambda \rg=0
\eeq
on  $M=  \open{-\delta, \delta}\times \Sigma$ with Cauchy data $(\rh, \rk)$,  such that $(M, \rg)$ is of bounded analytic geometry near $\Sigma$.
\end{theorem}
\proof  The result follows easily from  the local existence for the Einstein equations  given in \cite[Chap.~14]{R} and the Cauchy--Kowalevski theorem. Let us first recall the  arguments in \cite[Chap.~14]{R}.  Let us recall that ${\bG}_{ab}= \Ric_{ab}- \12 \rR g_{ab}$ is the Einstein tensor and that \eqref{eq:einsteen} is equivalent to $\bG + \Lambda g=0$.

 One equips $\rr\times \Sigma$ with the Lorentzian metric $\tilde{g}= - dt^{2}+ \rh$ and considers the  equation
\begin{equation}
\label{e.er0}
\Ric_{ab}+ \nabla_{(a}D_{b)}- \Lambda\rg=0,
\end{equation}
where the auxiliary $(0, 1)$-tensor $D$ is defined by
\[
D_{\nu}= - \rg_{\mu\nu}\rg^{\alpha\beta}(\Gamma^{\mu}_{\alpha\beta}- \tilde{\Gamma}^{\mu}_{\alpha\beta})
\]
and $\Gamma^{\mu}_{\alpha\beta}, \tilde{\Gamma}^{\mu}_{\alpha\beta}$ are the Christoffel symbols for $\rg$ and $\tilde{g}$.   The equation \eqref{e.er0} with unknown $\rg$ is a quasilinear hyperbolic system.  If  $\rg$ solves \eqref{e.er0} in some neighborhood $\Omega$ of $\Sigma$ in $\rr\times \Sigma$ then
\begin{equation}
\label{e.er-1}
\bG_{ab}= - \nabla_{(a}D_{b)}+ \12 \nabla^{c}D_{c}\rg_{ab} \hbox{ in }\Omega,
\end{equation}
hence since  $\nabla^{a}\bG_{ab}=0$
\begin{equation}
\label{e.er1}
- \square D_{a}+ \Ric\indices{_{a}^{b}}D_{b}=0 \hbox{ in }\Omega,
\end{equation}
 see \cite[(14.8)]{R}.  
 
 One then  fixes initial data on $t=0$ for \eqref{e.er0} by specifying $\rg_{|t=0}$ and $\p_{t}\rg_{|t=0}$. Concretely, one can take  $\rg_{|t=0}= - dt^{2}+ \rh$ and  $ (\p_{t}\rg_{\Sig\Sig})_{|t=0}= 2\rk$, where we use the notation introduced in Subsect.~\ref{sec:lg.3.2}. The remaining components $(\p_{t}\rg_{tt})_{|t=0}$ and $(\p_{t}\rg_{t\Sig})_{|t=0}$ are then completely fixed by requiring that $(D_{a})_{|t=0}=0$.

 The constraints \eqref{e.constr} are equivalent to $(\bG+ \Lambda \rg)_{t\Sigma}= 0$ on $t=0$, i.e.~to $\bG_{t\Sigma}=0$ on $t=0$. Using \eqref{e.er-1} and the fact that  $D_{a}=0$ on $t=0$ we obtain that $\p_{t}D_{a}=0$ on $t=0$, hence $D_{a}=0$ in $\Omega$ if $\rg$ solves \eqref{e.er0} in $\Omega$. It follows that a solution of \eqref{e.er0} in $\Omega$ with initial data fixed as above (assuming that $\rh, \rk$ satisfy the constraints \eqref{e.constr}) is a solution of the Einstein equations \eqref{eq:einsteen}	 in $\Omega$.												

Let us now fix  a bounded analytic atlas $\{(\mathcal{U}_{i}, \psi_{i})\}_{i\in \nn}$ of $\Sigma$ and transport $\rh, \rk$ to $C_{d}(0, 1)$ by $\psi_{i}$, obtaining tensors $\rh_{i}, \rk_{i}$. The tensors $\rh_{i}, \rk_{i}$, resp.~$\rh_{i}^{-1}$ belong to $\AT^{0}_{2}(C_{d}(0, 1), \delta)$, resp.~$\AT^{2}_{0}(C_{d}(0, 1), \delta)$, with norms uniformly bounded in $i$.

Let us drop  the index $i$ for the moment.  We work on $I\times C_{d}(0, 1)$ for some time interval $I$.
 By the Cauchy--Kowalevski theorem, the solution $\rg_{\mu\nu}$  of \eqref{e.er0} with initial data specified above is analytic in $\open{-\epsilon, \epsilon}\times C_{d}(0, \12)$ for some $\epsilon>0$.  The value of $\epsilon$ and the norm of $\rg_{\mu\nu}$  in  $\AT(\open{-\epsilon, \epsilon}\times C_{d}(0, \12))$ depend only on the norm of $\rh_{ij}, \rk_{ij}$ in $\AT(C_{d-1}(0, 1))$. Therefore putting back the index $i$ which labels the atlas  $\{(\mathcal{U}_{i}, \psi_{i})\}_{i\in \nn}$, we obtain that  the solution $\rg$ of \eqref{e.er0} constructed above belongs to $\AT^{0}_{2}(I\times \Sigma; dt^{2}+ \hat{h}_{0})$. \qed

\subsection{Examples}\label{an-bg2.7}
We now give several examples of spacetimes satisfying the hypotheses in Subsect.~\ref{hippopotamus}.
\subsubsection{Warped products}
Let $(\Sigma, \rh_{0})$ be a Riemannian manifold of bounded analytic geometry, $0\in I\subset \rr$ an open interval  and $f: I\to \rr$ an analytic function with $f(t)>0$ for $t\in I$. Take $M= I\times \Sigma$ and $\rg= -dt^{2}+ f(t)\rh$. Then $(M, \rg)$ is of bounded analytic geometry near $\{0\}\times \Sigma$.

This applies for example to  the \emph{de Sitter spacetime} 
\[
M= \rr_{t}\times \bS^{3}, \quad \rg=- dt^{2}+ 3\Lambda^{-1}\cosh((\Lambda/3)^{\12}t)d^{2}\omega.
\]
\subsubsection{The  Kerr--Kruskal spacetime}
Let us prove that the 
 maximal globally hyperbolic extension of the slowly rotating  exterior Kerr spacetime (i.e.~with parameters $a, M$ such that $0<|a|<M$), which we call for the sake of brevity the {\em Kerr--Kruskal spacetime}, satisfies the hypotheses in  Subsect.~\ref{hippopotamus}.   We will do this by applying Thm.~\ref{cauchykow}.
 
 The Kerr metric is given  on $\rr_{t}\times \rr_{r}\times\bS^{2}_{\theta, \varphi}$ in Boyer--Lindquist coordinates  $(t, r, \theta,\varphi)$   by 
 \begin{equation}
\label{kerr-expression}
\rg=-\left(1-\frac{2mr}{\rho^2}\right)dt^2-\frac{4aMr\sin^2\theta}{\rho^2}d td\varphi +\frac{\rho^2}{\Delta}d r^2+\rho^2d\theta^2+\frac{\sigma^2}{\rho^2}\sin^2\theta d\varphi^2,
\end{equation}
for 
\[
 \begin{array}{l}
 \Delta = r^{2}- 2Mr+a^2,\ 
\rho^{2} = r^{2}+ a^{2}\cos^{2}\theta,\\[2mm]
 \sigma^{2} = (r^{2}+a^{2})\rho^{2}+ 2 a^{2}Mr\sin^{2}\theta.
 \end{array}
\]

If $r_{-}<r_{+}$ are the two roots of $\Delta$, the regions ${\rm I}\defeq \{r_{+}<r\}$, ${\rm II}\defeq \{r_{-}<r<r_{+}\}$ are the first two Boyer--Lindquist blocks. The apparent singularity of $\rg$ at $r= r_{+}$ can be removed using Kruskal--Boyer--Lindquist coordinates $(U, V, \theta, \varphi^{\sharp})$, see \cite[Sect.~3.5]{N2}.
One obtains in this way the {\em Kerr--Kruskal spacetime} ${\rm KK}$, equal to $\rr_{U}\times \rr_{V}\times \bS^{2}_{\theta, \varphi^{\sharp}}$. 
The expression of the metric $\rg$ on ${\rm KK}$ is given in \cite[Prop.~3.5.3]{N2} and is manifestly analytic.  ${\rm KK}$ contains isometrically  the four Boyer-Lindquist blocks: ${\rm I}$, the blackhole exterior, ${\rm II}$, the black hole interior, and ${\rm I}'$, ${\rm II}'$, where $M'$ equals $M$ with the reversed time orientation. 

It is shown in \cite[Prop.~C.12]{GHW} that $({\rm KK}, \rg)$ is globally hyperbolic, with $\Sigma= \{U=V\}$ as an analytic  spacelike Cauchy surface. 

  \begin{figure}[ht]
  \def\svgwidth{12cm}
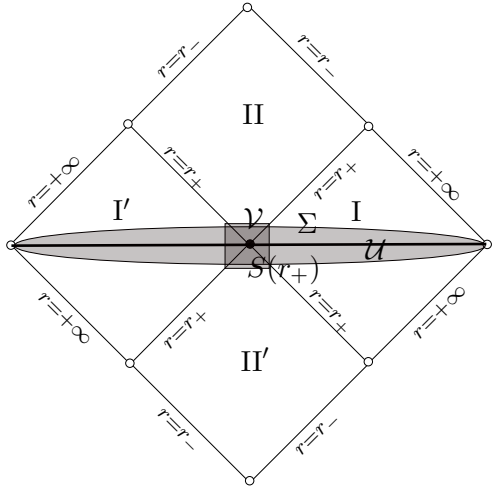
\caption{The Kerr--Kruskal spacetime. The shaded regions represent the neighborhoods $\cU$ and $\cV$.}
\end{figure}

 To check the hypotheses of Thm.~\ref{cauchykow} we  can remove  an arbitrary compact neighborhood $\mathcal {V}$ of the bifurcation sphere $S(r_{+})= \{U=V=0\}$.  
  
 If $\mathcal{U}$ is a neighborhood of $\Sigma$, then $\mathcal{U}\setminus\mathcal{V}$ is included in ${\rm I}\cup {\rm I}'$ (see Fig.~1),  so it suffices to consider the  Kerr metric in block ${\rm I}$, where the
 Cauchy surface $\Sigma$ equals $\{t= 0\}$, and we can use the original Boyer--Lindquist coordinates.

 As a reference Riemannian metric on $\Sigma$ we choose $\hat{h}_{0}= dr^{2}+ r^{2}d^{2}\omega$, where $d^{2}\omega$ is the round metric on $\bS^{2}$, which is clearly analytic of bounded geometry if we identify $\open{r_{0}, +\infty}_{r}\times \bS^{2}_{\theta, \varphi}$ with $\rr^{3}\setminus B(0, r_{0})$ by spherical coordinates. 
 
 It follows that to check that a $(0,2)$-tensor $f$ on $\Sigma$ belongs to $\AT^{0}_{2}(\Sigma, \hat{h}_{0})$, it suffices to check that the functions
 \[
 f_{rr},\  r^{-1}f_{r\theta},\  r^{-1}f_{r\varphi}, \ r^{-2}f_{\varphi\varphi},\  r^{-2}f_{\theta\varphi},\  r^{-2}f_{\theta\theta}
 \]
 extend as bounded holomorphic functions in $(r, \theta, \varphi)$ in some  strip $\{\module{\Im r}+ \module{ \Im \theta} + \module{ \Im \varphi}<\delta\}$.
 
 From this observation we immediately obtain that  the first fundamental form $\rh$ belongs to  $\AT^{0}_{2}(\Sigma, \hat{h}_{0})$.
 Let us now consider the second fundamental form $\rk$.
 
We note that if $X, Y\in T\Sigma$ then $X\dual \rk Y= (-dt\dual \rg^{-1}dt)^{-\12} \nabla_{X}dt\dual Y$, since  the future directed unit normal $n$ equals $(-dt\dual \rg^{-1}dt)^{-\12} \rg^{-1}dt$. Using the coordinates $r, \theta, \varphi$ on $\Sigma$, we obtain that
 \[
 \begin{array}{l}
 \rk_{r\varphi}= \12 (\rg^{tt})^{-\12}(\rg^{tt} \p_{r}\rg_{t\varphi}+ \rg^{t\varphi}\p_{r}\rg_{\varphi\varphi}),\\[2mm]
 \rk_{\theta\varphi}=\12 (\rg^{tt})^{-\12}(\rg^{tt} \p_{\theta}\rg_{t\varphi}+ \rg^{t\varphi}\p_{\theta}\rg_{\varphi\varphi}),
 \end{array}
 \]
all other components being $0$. We use that $\rg^{tt}= - \frac{\sigma^{2}}{\Delta\rho^{2}}$ and obtain similarly that $\rk\in \AT^{0}_{2}(\Sigma, \hat{h}_{0})$. \qed

 \subsubsection{The Schwarzschild--de Sitter spacetime}
The Schwarzschild--de Sitter metric is given on $\rr_{t}\times \rr_{\open{0, +\infty}}\times \bS^{2}_{\theta, \varphi}$ by
\[
\rg= - F(r)dt^{2}+ F^{-1}(r)dr^{2}+ r^{2}d^{2}\omega \, \hbox{ on } \, \rr_{t}\times \open{0, +\infty}_{r}\!\times \bS^{2},
\]
where $d^{2}\omega$ is the round metric on $\bS^{2}$ and $F(r)= 1-\frac{2m}{r}- \frac{\Lambda}{3}r^{2}$. For $9\Lambda^{2}m^{2}<1$, $F$ has two positive   $0<r_{+}<r_{++}$, with $F'(r_{+})>0$, $F'(r_{++})<0$ and $F(r)>0$ on $\open{r_{+}, r_{++}}$.  

We now have three Boyer--Lindquist blocks ${\rm I}\defeq \{r_{+}<r<r_{++}\}$, ${\rm II}_{H}\defeq\{0<r<r_{+}\}$ and ${\rm II}_{C}\defeq \{r_{+}<r<+\infty\}$. Block ${\rm I}$ is called the {\em static region} of S-dS,  $\mathcal{H}_{H/C}= \{r= r_{+/++}\}$ being the black hole, resp.~cosmological horizon.

The apparent singularities at $t= r_{+}, r_{++}$ are removed as before by introducing Kruskal-type coordinates near $\cH_{H}$, resp.~$\cH_{C}$ see e.g. ~\cite[Appendix]{GN}.  One first defines the Regge--Wheeler coordinate $r^{*}$ by $\frac{dr^{*}}{dr}= F^{-1}(r)$ (the integration constant being irrelevant). Near $\cH_{H}$ one sets $
U_{H}= - \e^{- \alpha(t-r^{*})}, \ V_{H}= \e^{\alpha(t+ r^{*})}$, so that $r^{*}= (2\alpha)^{-1} \ln(- U_{H}V_{H})$, $t= - (2\alpha)^{-1}\ln ( - U_{H}V_{H}^{-1})$. If $\alpha= \12 F'(r_{+})$, one obtains that $U_{H}V_{H}= (r- r_{+})f_{1}(r)$, where $f$ is analytic near $r= r_{+}$ and $f_{1}(r_{+})\neq 0$. 

Expressed in the coordinates $(U_{H}, V_{H}, \omega)$, the metric becomes
\[
\rg= - f_{2}(r) dU_{H}dV_{H}+ r^{2}d^{2}\omega,
\]
where again $f_{2}$ is analytic near $r= r_{+}$ and $f_{2}(r_{+})> 0$.  The coordinates $(U_{H}, V_{H})$ allow to glue the blocks ${\rm I}$, ${\rm II}_{H}$, ${\rm I}'$ and ${\rm II}_{H}'$ along $r= r_{+}$.

One can perform a similar change of coordinates near $r= r_{C}$ setting now
$U_{C}= \e^{\alpha(t- r^{*})}, V_{C}= - \e^{- \alpha(t+ r^{*})}, \ \alpha= - \12 F'(r_{++})$,
with similar conclusions with $r_{+}$ replaced by $r_{++}$. The coordinates $(U_{C}, V_{C})$ allow to glue the blocks ${\rm I}$, ${\rm II}_{C}$, ${\rm I}'$ and ${\rm II}_{C}'$ along $r= r_{++}$.
\def\rcup{\begin{turn}{45}$r= r_{++}$\end{turn}}
 \def\rcdown{\begin{turn}{-45}$r= r_{++}$\end{turn}}
 \def\rhup{\begin{turn}{45}$r= r_{+}$\end{turn}}
 \def\rhdown{\begin{turn}{-45}$r= r_{+}$\end{turn}}
  \begin{figure}[ht]
  \centering
    \def\svgwidth{14cm}
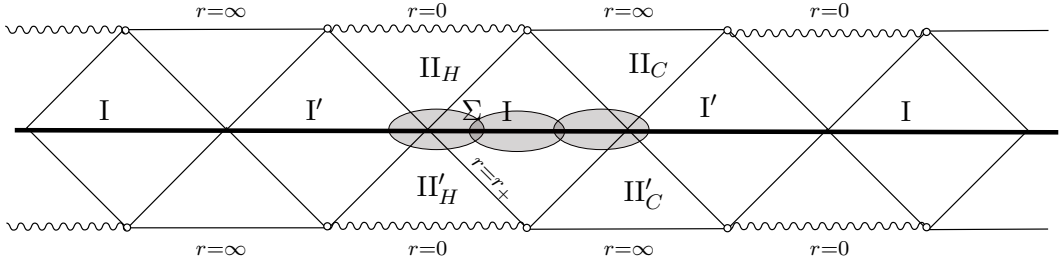
\caption{The extended Schwarzschild--de Sitter spacetime.}
\label{fig2}
\end{figure}

The {\em extended} {\rm S--dS} {\em  spacetime} is obtained by a bi-infinite sequence of these two gluing procedures, see Fig.~\ref{fig2}.

Let us now consider as initial surface $\Sigma= \{t=0\}$ in block ${\rm I}$. In the Kruskal coordinates (near $r= r_{+}$ or $r= r_{++}$), $\Sigma= \{U_{H/C}+ V_{H/C}=0\}$ and we can extend $\Sigma$ to the extended S--dS spacetime.

The metric $\rg$ is clearly analytic and bounded in each of the three shaded open sets  in Fig.~\ref{fig2}.  From this we obtain that 
 if $\rh, \rk$ are the Cauchy data of $\rg$ on $\Sigma$,  $(\Sigma, \rh)$ is of bounded analytic geometry and $\rk\in \AT^{0}_{2}(\Sigma, \rh)$. 
\begin{remark}
We expect that the Kerr--de Sitter metric can be handled similarly. We refer the reader to \cite{borthwick} for a recent investigation on maximal extensions of the Kerr--de Sitter spacetime.
\end{remark}

\section{Linearized gravity}\label{sec:lg}\init
  In this section we  formulate linearized gravity as a classical gauge theory, following \cite{HS}. We also rewrite the various operators occurring in linearized gravity after using  Gaussian normal coordinates to a Cauchy surface and parallel transport. This reduced setting will be used in later sections.
\subsection{Notation} \label{sec:lg.1}
 We start by fixing notation. Let $(M,\rg)$ be a $4$-dimensional  Lorentzian manifold.   
\subsubsection{Convention for the Riemann tensor}\label{sec:lg.1.1} 
We use the same convention as in e.g.~\cite{R,FH,BDM} for the sign of the Riemann tensor $\rR\indices{_{abcd}}=\rR\indices{_{abc}^{e}}\rg_{ed}$, i.e.
\[
(\nab_{a}\nab_{b}- \nab_{b}\nab_{a}) u_{c}= \rR\indices{_{abc}^{d}}u_{d}
\]
on $(0,1)$-tensors in terms of the Levi-Civita connection $\nabla_a$ on $(M,\rg)$. Let us recall that the Ricci tensor  is the symmetric tensor
\[
\Ric_{ab}=\rR\indices{_{acb}^{c}}= \rR\indices{^{c}_{acb}},
\]
and the scalar curvature is $\scal= \rg^{ab}\Ric_{ab}$. If $\dim M=4$ then the non-linear Einstein equations with cosmological constant $\Lambda\geq 0$, i.e.
$
\Ric -\12\rg\scal +\Lambda \rg =0
$,
are equivalent to 
\beq\label{eq:einstein}
\Ric =\Lambda \rg,
\eeq 
and as a consequence of \eqref{eq:einstein} one gets $\scal=4\Lambda$.

\subsubsection{Hermitian forms on tensors}\label{sec:lg.1.2}
 We denote by 
\[
V_{k}\defeq\cc\otimes^{k}_{\rm s}T^{*}M
\]
 the complex bundle of symmetric $(0,k)$-tensors.  We will  only need  the cases $k=1,2$.
$V_{k}$ is  equipped with the non-degenerate Hermitian form
\begin{equation}
\label{eq:uvk}
(u| u)_{V_{k}}\defeq\bar{u}\dual  k! (\rg^{\otimes k})^{-1}u.
\end{equation}
 In abstract index notation,
\[
 (u|u)_{V_{k}}= k! \,\rg^{a_{1}b_{1}}\cdots\rg^{a_{k}b_{k}}\bar{u}_{a_{1}\dots a_{k}} u_{b_{1}\dots b_{k}}.
\]
 For example for $k=2$ we have
\beq\label{e0.00}
(u|u)_{V_{2}}= 2\trace(u^{*}\rg^{-1}u\rg^{-1}).
\eeq
The $k!$ normalization differs from the most common convention, it has however the advantage that various expressions involving adjoints look more symmetric.

Sometimes it will be convenient to indicate explicitly  the metric used to define the Hermitian form on $V_{k}$, for $k=1,2$. Accordingly, we will set
\beq\label{e0.00b}
\begin{array}{l}
(u| u)_{\rg}= \bar{u}\dual \rg^{-1}u= (u| u)_{V_{1}}, \\[2mm]
(u| u)_{g^{\otimes 2}}=  \bar{u}\dual (g^{\otimes 2})^{-1}u= \12 (u| u)_{V_{2}}.
\end{array}
\eeq

As in Sect.~\ref{sec:gauge}, for $U\subset M$ open,  the Hermitian form \eqref{eq:uvk} on fibers induces a Hermitian form
\begin{equation}
\label{uvkU}
(u|v)_{V_{k}(U)}=  \int_{U}(u(x)| v(x))_{V_{k}}\dvol_{\rg}, \quad u, v\in \coinf(U;V_{k}).
\end{equation} 
The adjoint of  $A: \cinf(M; V_{k})\to \cinf(M; V_{l})$ for those Hermitian forms will be denoted by $A^{*}$.
\subsubsection{The differential and its adjoint}\label{sec:lg.1.3}
Let  
\[
d: \begin{array}{l}
\cinf(M; V_{k})\to \cinf(M; V_{k+1})\\[1mm]
(d u)_{a_{1} \dots, a_{k+1}}= \nab_{(a_{1}}u_{a_{2}\dots, a_{k+1})},
\end{array}
\]
where $u_{(a_{1} \dots a_{k})}$ is the symmetrization of $u_{a_{1}\dots a_{k}}$,
and
\[
\delta: \begin{array}{l}
\cinf(M; V_{k})\to\cinf(M; V_{k-1})\\[1mm]
(\delta u)_{a_{1}, \dots, a_{k-1}}= -k\nab^{a}u_{aa_{1}\dots a_{k-1}}.
\end{array}
\]
With these conventions, we have $d^*=\delta$ w.r.t.~the Hermitian form \eqref{uvkU}.

\subsubsection{Operators on tensors}\label{sec:lg.1.4}
The operator of \emph{trace reversal} $I$
 is given by
 \[
I\defeq \one - \frac{1}{4}|\rg)(\rg|,
\]
i.e.~$I$ is the orthogonal symmetry w.r.t.~the line $\cc \rg$. Equivalently
\[
(Iu)_{ab}= u_{ab}- \frac{1}{2} \trace_{\rg}(u) \rg_{ab}, \quad \trace_{\rg}(u)\defeq \rg^{ab}u_{ab}= \12 (\rg| u)_{V_{2}}. 
\]
It satisfies 
\beq\label{idiot}
I^{2}= \one, \quad I= I^{*}\hbox{ on }\cinf(M; V_{2}).
\eeq
The \emph{d'Alembertian}  is 
\[
\Box= \nabla^{c}\nabla_{c}, \hbox{ acting on }(p, q)\hbox{-tensors}.
\]
and the \emph{Ricci operator} is
\[
\Riem(u)_{ab}\defeq \rR\indices{_{a}^{cd}_{b}}u_{cd}= \rR\indices{^{c}_{ab}^{d}}u_{cd}, \ \ u\in\cinf(M; V_{2}).
\]
The fact that $\Riem$ preserves symmetric $(0,2)$-tensors follows from the symmetries of the Riemann tensor.

The proofs of the next two lemmas are easy and left to the reader.
\begin{lemma}\label{lemma2.11} The d'Alembertian satisfies:
 \beq\label{e2.4}
\begin{array}{rl}
i)& \square\circ I= I \circ \square, \\[2mm]
 ii)& \square= \square^{*}.
 \end{array}
\eeq
\end{lemma}

\begin{lemma}\label{lemma2.12} The Ricci operator satisfies:
 \beq\label{e2.5}
\begin{array}{rl}
i)&\Riem \rg=-\Ric,\\[2mm]
ii)&\Riem \circ I=  I\circ \Riem,\hbox{ if }g\hbox{ is Einstein},\\[2mm]
iii)&\Riem= \Riem^{*}.
\end{array}
\eeq
\end{lemma}

\begin{lemma}\label{lemma2.1}
If $(M, \rg)$ is Einstein then:
 \begin{equation}
\label{e2.8}
\begin{array}{rl}
i)&(\square - 2\Riem)\circ d= d\circ (\square+\Lambda)\hbox{ on }\cinf(M; V_{1}),\\[2mm]
ii)& \delta \circ (\square - 2 \Riem)= (\square+ \Lambda) \circ\delta\hbox{ on }\cinf(M; V_{2}),\\[2mm]
iii)&\delta\circ I\circ d=  \square+ \Lambda\hbox{ on }\cinf(M; V_{1}),\\[2mm]
iv)&(\square- 2 \Riem)\circ I= I\circ (\square- 2 \Riem)\hbox{ on }\cinf(M; V_{2}),\\[2mm]
v)&\trace_{\rg}\circ(\square- 2\Riem)= (\square- 2\Lambda) \circ\trace_{\rg}\hbox{ on }\cinf(M; V_{2}).
\end{array}
\end{equation}
\end{lemma}

\proof {\it ii)}  is \cite[Lem.~2.6]{FH} and it implies {\it i)}  by taking adjoints w.r.t.~$(\cdot | \cdot)_{V_{2}(M)}$.
 Next,
\[
\bea
(I\circ dw)_{ab}&= \12\left(\nab_{a}w_{b}+ \nab_{b}w_{a}- \12\rg^{cd}(\nab_{c}w_{d}+ \nab_{d}w_{c})\rg_{ab}\right)\\
&=\12\left(\nab_{a}w_{b}+ \nab_{b}w_{a}-(\nab^{d}w_{d})\rg_{ab}\right),\\[2mm]
(\delta \circ I\circ  dw)_{b}
&=\-\left(\nab^{a}\nab_{a}w_{b}+ \nab^{a}\nab_{b}w_{a}- \nab_{b}\nab^{a}w_{a}\right)\\
&= -(\square w)_{b}- \rR\indices{^{a}_{ba}^{d}}w_{d}= -(\square w)_{b}- \Lambda w_{b}.
\eea
\]
Furthermore, {\it iv)}  follows from \eqref{e2.4} {\it ii)} and \eqref{e2.5} {\it iii)}. Finally,
\[
\trace_{\rg}(\Riem(u))= \12(\rg| \Riem(u))= \12(\Riem(\rg)| u)_{V_{2}}=\frac{\Lambda}{2}(\rg| u)_{V_{2}}=\Lambda\,\trace_{\rg}u,
\]
since $\Riem= \Riem^{*}$ and $\Riem(\rg)=\Lambda\rg$, which implies {\it v)}.\qed

\subsection{Linearized gravity as a gauge theory} \label{sec:lg.2}
Let us now explain how linearized gravity fits in the framework introduced in Sect.~\ref{ss:HS}.

Let $(M,\rg)$ be a globally hyperbolic spacetime of dimension $4$. Let us introduce the differential operators
\beq\label{e2.4bb}
\bea
P&\defeq - \square-  I \circ d \circ \delta + 2\,\Riem \in \Diff^2(M;V_2), \\
K&\defeq  I\circ d\in\Diff^1(M;V_1,V_2).
\eea
\eeq
From now on we assume that $(M,\rg)$ is Einstein. Then,  $Pu=0$ is  the \emph{linearized Einstein equation}.  The condition $K^\star u=0$, where $K^{\star}$ is defined below,  is the linearized \emph{de Donder} or \emph{harmonic gauge}.
\begin{remark}
  We use the same formalism as in \cite[Example 3.8]{HS}. In \cite{FH},  $P$ is replaced by  $L=\12 P\circ I$ and $K$ by $d= I\circ K$, which corresponds to replacing  $u$ by $Iu$.
\end{remark}

We consider $V_{1}$ resp.~$V_{2}$ as Hermitian bundles, where the Hermitian forms on fibers is now
\beq\label{def-physical-scalar-product}
(u| u)_{I,V_{1}} \defeq (u|u)_{V_{1}}, \mbox{ for } u\in  V_1, \mbox{ resp.~} (u| u)_{I,V_{2}} \defeq (u| I u)_{V_{2}}, \mbox{ for } u\in  V_2. 
\eeq
The corresponding Hermitian form on smooth sections of $V_i$, $i=1,2$, is 
\begin{equation}\label{eq:hform}
(u|u)_{I,V_{k}(U)}=  \int_{U}(u(x)| u(x))_{I,V_{k}}\dvol_\rg, \quad u, v\in \coinf(U;V_k).
\end{equation} 
We denote by $A^{\star}$ the corresponding formal adjoint of $A$ for $(\cdot| \cdot)_{I, V_{k}(M)}$ to distinguish it from the formal adjoint $A^*$ for   $(\cdot|\cdot)_{V_{k}(M)}$.  The two are related as follows:
\beq\label{e2.7}
\begin{array}{rl}
A^{\star}= I A^{*}I \hbox{ if }A: \cinf(M;V_{2})\to \cinf(M;V_{2}),\\[2mm]
A^{\star}= A^{*}I\hbox{ if }A: \cinf(M;V_{1})\to \cinf(M;V_{2}),\\[2mm]
A^{\star}= IA^{*}\hbox{ if }A: \cinf(M;V_{2})\to \cinf(M;V_{1}),\\[2mm]
A^{\star}= A^{*}\hbox{ if }A: \cinf(M;V_{1})\to \cinf(M;V_{1}).
\end{array}
\eeq 
In particular,
\beq\label{eq:Kstar}
K^{\star}= K^{*}I=\delta \circ I \circ I=\delta.
\eeq 
\begin{proposition}\label{lemma2.2} Let 
\[
\begin{array}{l}
D_2=P+KK^\star: \cinf(M; V_{2})\to \cinf(M; V_{2}),\\[2mm]
D_1=K^\star K: \cinf(M; V_{1})\to \cinf(M; V_{1}).
\end{array}
\] Then
\beq\label{eq:D1D2}
\begin{array}{l}
D_2=-\square +2\Riem,\\[2mm]
D_1=-\square-\Lambda.
\end{array}
\eeq
Furthermore, Hypothesis  \ref{as:subsidiary} is satisfied, i.e.~$P^\star=P$, $PK=0$ and $D_1,D_2$ are Green hyperbolic.
\end{proposition}
\proof  The first identity in \eqref{eq:D1D2} follows from $KK^{\star}=I \circ d \circ \delta$. Using  \eqref{eq:Kstar} and \Lemma \ref{lemma2.1} we get $K^{\star}K= \delta \circ  I  \circ d= -\square-\Lambda$, hence the second identity.

Next, by  \Lemmas  \ref{lemma2.11} and  \ref{lemma2.12} we have:
\[
P^{*}= -\square-d\circ \delta\circ  I+ 2\Riem,\quad
P^{\star}= I P^{*}I= P.
\]
The identity $PK= 0$ follows from
\[
\begin{array}{l}
(-\square -I\circ  d\circ  \delta + 2 \Riem)I\circ  d\\[2mm]
= I\circ (-\square + 2\Riem)\circ d- I\circ  d \circ (\delta\circ  I \circ d)\\[2mm]
= I\circ  d( - \square - \Lambda)+ I\circ  d( \square + \Lambda)=0,
\end{array}
\]
by  \Lemma \ref{lemma2.1}. Finally, Green hyperbolicity of $D_i$ for $i=1,2$ follows from \cite[Thm.~3.3.1]{BGP}, since their  principal symbol is $(\xi\cdot \rg^{-1}(x)\xi) \one_{V_i}$.   \qed

\begin{remark}
 The following alternative expression for $D_{i}$ can be found in \cite{Wuensch}:
 \[
 \begin{array}{l}
D_{1}= \delta\circ  d- d\circ \delta,\\[2mm]
D_{2}= \delta\circ  d- d\circ \delta+  4 \Riem + 2 \Lambda.
\end{array}
\]
\end{remark}

\subsection{Gaussian normal coordinates}\label{sec:lg.3}
We assume that  $M= I_{t}\times \Sigma_{\rx}$, where $ I\subset \rr$ is an interval with $0\in \mathring{I}$ and $\Sigma$ a smooth $d$-dimensional manifold.  We denote by $\xi= (\tau, k)$ the dual variables to $x= (t, \rx)$.

We set $\Sigma_{t}= \{t\}\times \Sigma$ and identify $\Sigma_{0}$ with $\Sigma$. We assume that 
\[
\rg= - dt^{2}+ \rh(t, \rx)d\rx^{2},
\]
where 
 \[
 \rh\in \cinf(M,  \otimes^{2}_{s}T^{*}\Sigma),
 \]
 is a $t$-dependent Riemannian metric on $\Sigma$. If the normal geodesic flow to $\Sigma\subset M$ is defined for a uniform time interval, as is the case for spacetimes of bounded geometry, we can reduce ourselves to this model situation.
 
 Let
 \beq\label{e4.01}
 \rer\defeq\12 \p_{t}\rh \rh^{-1}\in \cinf(M; L(T^{*}\Sigma)).
\eeq 
By computing the Christoffel symbols for $\rg$, see \eqref{etiti.1}, one can check that the second fundamental form of $\rg$ is ${\rk}=  \12 \p_{t}\rh$, hence $\rer= {\rk}\rh^{-1}$. 

 We will use the notation introduced in \eqref{e0.00b} for the Hermitian forms on $V_{1}, V_{2}$.
 \subsubsection{ Decomposition of $(0,1)$-tensors}\label{sec:lg.3.1}
 We identify  \beq\label{etiti.0}
 \begin{array}{l}
 \cinf(M; T^{*}M)\tosim\cinf(M)\oplus \cinf(M; T^{*}\Sigma)\hbox{ by}\\[2mm]
  w\mapsto (w_{t}, w_{\Sig}), \\[2mm]
   w \eqdef w_{t}dt+ w_{\Sig}.
  \end{array}
 \eeq
  The scalar product $(\cdot| \cdot)_{\rg}$ reads then
 \[
 (w| w)_{\rg}= - | w_{t}|^{2}+ (w_{\Sig}| w_{\Sig})_{\rh},
 \]
 i.e.,
\[
\rg^{-1}= \mat{-1}{0}{0}{\rh^{-1}}.
\]
\begin{lemma}\label{lemmatiti.1}  If $w_{1}, w_{2}$ are $(0, 1)$-tensors on $\Sigma$ we have:
 \beq\label{e2.04b}
(w_{ 1}| \rer w_{ 2})_{\rh}= (\rer w_ {1}| w_{2})_{\rh},
\eeq
i.e.~$\rer$ is selfadjoint for $(\cdot| \cdot)_{\rh}$.
\end{lemma}

\subsubsection{Decomposition of $(0,2)$-tensors}\label{sec:lg.3.2}
Similarly we identify  \beq\label{etiti.-1}
 \begin{array}{l}
 \cinf(M; \otimes^{2}_{\rm s}T^{*}M)\tosim\cinf(M)\oplus \cinf(M; T^{*}\Sigma)\oplus\cinf(M; \otimes^{2}_{\rm s}T^{*}\Sigma)\hbox{ by}\\[2mm]
  u\mapsto (u_{tt}, u_{t\Sig},u_{\Sig\Sig}), \\[2mm]
u\eqdef u_{tt}dt\otimes dt+ w_{t\Sig}\otimes dt+ dt\otimes w_{t\Sig}+ w_{\Sig\Sig}.
  \end{array}
 \eeq

The scalar product  $(\cdot| \cdot)_{g^{\otimes 2}}$  reads
\beq\label{e2.9b}
(u|u)_{g^{\otimes 2}}= |u_{tt}|^{2}-  2(u_{t\Sig}| u_{t\Sig})_{\rh} +  (u_{\Sig\Sig}| u_{\Sig\Sig})_{h^{\otimes 2}}.
\eeq
 We have the following analog of \Lemma \ref{lemmatiti.1}.  In \eqref{e2.04c} below, $r^{\scriptscriptstyle\#}\in L(T\Sigma)$ is the transpose of $\rer\in L(T^{*}\Sigma)$.
\begin{lemma}\label{lemmatiti.2}
 If $u_{1}, u_{2}$ are $(0,2)$-tensors on $\Sigma$ we have:
 \begin{equation}
 \label{e2.04c}
 \begin{array}{l}
  (u_{1}| \rer \circ u_{2})_{h^{\otimes 2}}= (\rer\circ u_{1}| u_{2})_{h^{\otimes 2}}, \\[2mm]
  (u_{1}| u_{2}\circ r^{\scriptscriptstyle\#} )_{h^{\otimes 2}}= (u_{1}\circ r^{\scriptscriptstyle\#}| u_{2})_{_{h^{\otimes 2}}},
\end{array}
 \end{equation}
 i.e.~the operators $\rer \circ \cdot$ and $\cdot \circ r^{\scriptscriptstyle\#}$ are self-adjoint for $(\cdot| \cdot)_{_{h^{\otimes 2}}}$.
\end{lemma}
\proof
 Since $(u_{1}| u_{2})_{_{h^{\otimes 2}}}= \trace(u_{1}^{*}\rh^{-1}u_{2}\rh^{-1})$, we have
 \[
 \begin{array}{l}
 (u_{1}| \rer \circ u_{2})_{_{h^{\otimes 2}}}= \trace(u_{1}^{*}\rh^{-1}\p_{t}\rh\rh^{-1}u_{2}\rh^{-1})\\[2mm]
 = \trace((\p_{t}\rh \rh^{-1}u_{1})^{*}\rh^{-1}u_{2}\rh^{-1})= (\rer\circ u_{1}| u_{2})_{_{h^{\otimes 2}}},\\[2mm]
 (u_{1}| u_{2}\circ r^{\scriptscriptstyle\#})_{_{h^{\otimes 2}}}=  \trace(u_{1}^{*}\rh^{-1}u_{2}\rh^{-1}\p_{t}\rh\rh^{-1})\\[2mm]
 = \trace((u_{1}\p_{t}\rh\rh^{-1})^{*}\rh^{-1}u_{2}\rh^{-1})= (u_{1}\circ r^{\scriptscriptstyle\#}| u_{2})_{_{h^{\otimes 2}}}. \ \Box
 \end{array}
 \]
\subsubsection{Covariant derivatives}\label{sec:lg3.3}
We recall that  $\nabla$, $d$, $\delta$  are the covariant derivative, differential and codifferential for $\rg$. Similarly we denote by $\nabla_{\Sigma}$, $d_{\Sigma}$, $\delta_{\Sigma}$ the corresponding operators for $\rh$, acting on tensors on $\Sigma$. 
We set
\[
\nabla_{0}= \nabla_{\p_{t}}.
\]

Let us denote by $\Gamma^{c}_{ab}(\rg)$ resp.~ $\Gamma^{i}_{ij}(\rh)$ the Christoffel symbols of $\rg$ resp.~$\rh$. A routine computation gives
\begin{equation}
\label{etiti.1}
\begin{array}{l}
\Gamma^{0}_{ij}(\rg)= \12\p_{t}\rh_{ij},\\[2mm]
\Gamma^{i}_{0j}(\rg)=\Gamma^{i}_{j0}(\rg)=  \rer\indices{_{j}^{i}},\\[2mm]
\Gamma^{i}_{jk}(\rg)= \Gamma^{i}_{jk}(\rh),
\end{array}
\end{equation}
all other entries being equal to $0$.  
From \eqref{etiti.1} we obtain that
\begin{equation}
\label{etiti.2}
\begin{array}{l}
(\nabla_{0}w)_{t}= \p_{t}w_{t},\\[2mm]
(\nabla_{0}w)_{\Sigma}= \p_{t}w_{\Sigma}- \rer w_{\Sigma},\\[2mm]
(\nabla_{0}u)_{t\Sigma}= \p_{t}u_{t\Sigma}- \rer u_{t\Sigma},\\[2mm]
(\nabla_{0}u)_{\Sigma\Sigma}= \p_{t}u_{\Sigma\Sigma}- \rer \circ u_{\Sigma\Sigma}- u_{\Sigma\Sigma}\circ r^{\scriptscriptstyle\#}.   
\end{array}
\end{equation}
Furthermore,
 \begin{equation}
 \label{etiti.4}
 \begin{array}{l}
  (dw)_{tt}= \p_{t}w_{t},\\[2mm]
  (dw)_{t\Sigma}= \12 \p_{t}w_{\Sigma}- \rer w_{\Sigma}+ \12 d_{\Sigma}w_{t},\\[2mm]
  (dw)_{\Sigma\Sigma}=  d_{\Sigma}w_{\Sigma}- \12 \p_{t}\rh w_{t}.
 \end{array}
 \end{equation}

\subsection{Reduced setting} \label{ss:reduction}
 We now perform a reduction  which corresponds to identify tensors  on $\Sigma_{t}$ with tensors on $\Sigma$ by parallel transport along $\p_{t}$. This allows us to simplify the expressions of  the operators $D_{i}, \nabla_{0}$, $d, \delta$.
Let
\[
 \ru(t)= \Texp(\textstyle\int_{t}^{0} \rer(s)ds)\in\cinf(I; \cinf(\Sigma; L(T^{*}\Sigma))),
\]
i.e.~$\ru(t)$ is the unique solution of
\beq\label{e2.a1}
\begin{cases}
\p_{t}\ru(t)= - \ru(t)\rer(t),\\
\ru(0)= \one.
\end{cases}
\eeq
for $t\in I$. We  define
\beq\label{eq:red1}
\rU_{1}(w_{t},w_{\Sigma}) \defeq (w_{t}, \ru w_{\Sigma}),
\eeq
where we use the identification in \eqref{etiti.0} and
 \beq\label{eq:red2}
\rU_{2}(u_{tt}, u_{t\Sigma}, u_{\Sigma\Sigma})\defeq (u_{tt}, \ru w_{t\Sigma}, \ru \circ  u_{\Sigma\Sigma}\circ \ru^{\scriptscriptstyle\#}),
\eeq
using the identification in \eqref{etiti.-1}. Let us set
\beq\label{eq:red3}
\ha\rg_{0}\defeq - dt^{2}+ \ha\rh_{0}(\rx)d\rx^{2},
\eeq
where $\ha\rh_{0}(\rx)=\rh(0,\rx)$. 

\begin{proposition}\label{prop4.1}  We have:
 \ben
 \item\label{it:U1}  $\rU_{1}^{*}\ha\rg_{0}^{-1}\rU_{1}= \rg^{-1}$, 
 \item\label{it:U2} $\rU_{2}^{*}\ha (g_{0}^{\otimes 2})^{-1} \rU_{2}=(g^{\otimes 2})^{-1}$,
  \item\label{it:U3}  $\rU_{2}\rg= \ha\rg_{0}$, 
 \item\label{it:U4}  $\ha I_{0}\defeq\rU_{2} I \rU_{2}^{-1}$ is the trace reversal w.r.t.~$\ha\rg_{0}$, 
 \item\label{it:U5} $\rU_{i}\circ \nab_{0}= \p_{t}\circ \rU_{i}$,  $i= 1, 2$. 
 \een
\end{proposition}
\proof We compute using \eqref{e2.a1}:
 \[
\p_{t}(\ru(t)\rh({t})\ru^{\scriptscriptstyle\#}(t))= \ru(t)\left(- \rer(t) \rh(t)+ \p_{t}\rh(t)- \rh(t)r^{\scriptscriptstyle\#}(t)\right)\ru^{\scriptscriptstyle\#}(t)=0.
\]
This implies that $\ru(t)\rh({t})\ru^{\scriptscriptstyle\#}(t)= \ha\rh_{0}$,
hence  $ \rU_{1}^{\scriptscriptstyle\#}\rg_{0}^{-1}\rU_{1}= \rg^{-1}$ or equivalently, $\rU_{2}\rg= \rg_{0}$, which proves \eqref{it:U1} and \eqref{it:U3}. Similarly
\[
\bea
(\rU_{2}u|\rU_{2} u)_{\rg_{0}^{\otimes 2}}&= \trace(\rg_{0}^{-1}\rU_{2}\bar{u} \rg_{0}^{-1}\rU_{2}u)
=\trace(\rg_{0}^{-1}\rU_{1} \bar{u} \rU_{1}^{\scriptscriptstyle\#}\rg_{0}^{-1}\rU_{1}u \rU_{1}^{\scriptscriptstyle\#})\\[2mm]
&=\trace(\rU_{1}^{\scriptscriptstyle\#}\rg_{0}^{-1}\rU_{1} \bar{u} \rU_{1}^{\scriptscriptstyle\#}\rg_{0}^{-1}\rU_{1}u )
=\trace(\rg^{-1}\bar{u}\rg^{-1}u)= (u| u)_{\rg^{\otimes 2}},
\eea
\]
which proves  \eqref{it:U2}. From \eqref{it:U2} and \eqref{it:U3}  we obtain
\[
(\rg | u)_{\rg}= (\rU_{2}\rg| \rU_{2}u)_{\rg_{0}}= (\rg_{0}| \rU_{2}u)_{\rg_{0}},
\]
which implies \eqref{it:U4}.  Finally, by \eqref{etiti.2} we obtain  \eqref{it:U5}. \qed\medskip
\begin{lemma}
 \label{lemmatiti.3}
We have:
\[
\begin{array}{l}
(\rU_{2} d\rU_{1}^{-1}w)_{tt}= \p_{t}w_{t},\\[2mm]
(\rU_{2} d\rU_{1}^{-1}w)_{t\Sigma}= \12 (\p_{t}w_{\Sigma}- \ru\rer \ru^{-1}w_{\Sigma}+ \ru d_{\Sigma}w_{t}),\\[2mm]
(\rU_{2} d\rU_{1}^{-1}w)_{\Sigma\Sigma}= \ru (d_{\Sigma}\ru^{-1}w_{\Sigma})\ru^{\scriptscriptstyle\#}- \12 \ru \p_{t}\rh \ru^{\scriptscriptstyle\#}w_{t}.
\end{array}
\]
\end{lemma}
\proof This follows from \eqref{etiti.4}, using that $\p_{t}\ru^{-1}= \rer \ru^{-1}$. \qed

\subsubsection{Reduced operators}\label{sssect.reduct}
Let us set for $u\in\coinf(M; V_{i})$, $i= 1,2$:
\[
Tu= \bs_{t}u\hbox{ for }\bs_{t}= |\rh_{t}|^{\frac{1}{4}}|\rh_{0}|^{-\frac{1}{4}}. 
\]
Then 
\beq\label{etiti.5}
(u| u)_{V_{i}(M)}=\int_{M}(u| u)_{V_{i}}d\vol_{\rg}=  \int_{M} (Tu| Tu)_{V_{i}} d\vol_{\rg_{0}},
\eeq
and hence
\[
\bea
(u| u)_{V_{1}(M)}&= \int_{M} (T\rU_{1}u| T\rU_{1}u)_{\rg_{0}} d\vol_{\rg_{0}},\\
(u| u)_{V_{2}(M)}&= 2\int_{M} (T\rU_{2}u| T\rU_{2}u)_{g^{\otimes 2}_{0}} d\vol_{\rg_{0}}.
\eea
\]
We define the analogues of $D_{i}$, $d$ and $I$ in the reduced setting:
\beq\label{etiti.7}
\bea
\hat{D}_i&\defeq (T\rU_i ) \circ D_i  \circ (T\rU_i)^{-1},  \\
 \hat{d}&\defeq (T\rU_2) \circ  d \circ  (T\rU_1)^{-1},\\
\hat{I}&\defeq (T\rU_2)\circ I \circ (T\rU_2)^{-1}.
\eea
\eeq
\begin{proposition}\label{prop4.2} The operators $\hat{D}_{i}$, $\hat{d}$ and $\hat{I}$ have the following properties.
 \ben
 \item\label{iuiu1} $\hat{D}_{i}= \p_{t}^{2}+ \hat{a}_{i}(t)$, where $\hat{a}_{i}\in \cinfb(I; \Psi^{2}(\Sigma; V_{i}))$ has principal symbol ${\rm k}\dual \rh_{t}^{-1}{\rm k}\,\one_{V_{i}}$ and is self-adjoint for the Hermitian form $(\cdot| \cdot)_{\hat{V}_{i}(\Sigma)}$ defined by:
 \[
 (u| u)_{\hat{V}_{1}(\Sigma)}\defeq  \int_{\Sigma}(u| u)_{\rg_{0}}d\vol_{\rh_{0}},\quad 
 (u| u)_{\hat{V}_{2}(\Sigma)}\defeq  2\int_{\Sigma}(u| u)_{g^{\otimes 2}_{0}}d\vol_{\rh_{0}}.
\]
 \item\label{iuiu2}We have
 \[
\begin{array}{l}
 (\hat{d}w)_{tt}= \p_{t}w_{t}- \12 \trace(r) w_{t},\\[2mm]
 (\hat{d}w)_{t\Sigma}= \12 ( \p_{t}w_{\Sigma}- \12 \trace (\rer) w_{\Sigma}- \ru \rer \ru^{-1}w_{\Sigma} + (\bs\ru)\circ d_{\Sigma}\circ (\bs\ru)^{-1}w_{t}),\\[2mm]
 (\hat{d}w)_{\Sigma\Sigma}=  (\bs\ru)\circ d_{\Sigma}\circ(\bs\ru)^{-1}w_{\Sigma}\ru^{\scriptscriptstyle\#}- \12 \ru \p_{t}\rh\ru^{\scriptscriptstyle\#}w_{t}.
\end{array}
 \]
  \item  \label{iuiu3} $
  \hat{I}= I_{0}= \one - \frac{1}{4}| \rg_{0}) (\rg_{0}|$.
  \item \label{iuiu4} We have $\hat{D}_{2}\hat{d}= \hat{d}\hat{D}_{1}$ and  $\hat{I}\hat{D}_{2}= \hat{D}_{2}\hat{I}$.  \een
\end{proposition}
\proof 
\eqref{iuiu4} is straightforward. 

Observe  that $\hat{D}_{i}$ has the same principal symbol as $D_{i}$,  i.e.~$\xi\cdot \rg^{-1}\xi\, \one_{V_i}$. Therefore,
\[
\hat{D}_{i}= \p_{t}^{2}+ b_{i}(t, \rx)\p_{t}+a_{i}(t, \rx, \p_{x}), 
\]
where $\sigma_{\rm pr}(a_{i})(t, \rx, k)= 
k_{i}\rh^{ij}(t, \rx)k_{j}\one$  and $\hat{a}_{i}\in \cinf(I; \Psi^{2}(\Sigma; V_{i}))$, $b_{i}\in \cinf(I; \Psi^{0}(\Sigma; V_{i}))$.

Since $D_{i}$ is self-adjoint for  $(\cdot| \cdot)_{V_{i}(M)}$, we deduce from 
 Prop.~\ref{prop4.1} and \eqref{etiti.5} that $\hat{D}_{i}$ is self-adjoint for the Hermitian form:
\[
(u| u)_{\hat{V}_{i}(M)}= \int_{M}(u| u)_{\rg_{0}}d\vol_{\rg_{0}}.
\]
This implies that $b_{i}(t, \rx)=0$ and that $a_{i}(t, \rx, \p_{\rx})$ is self-adjoint for the scalar product $(\cdot| \cdot)_{\hat{V}_{i}(\Sigma)}$, which proves \eqref{iuiu1}. 

 We have $\p_{t}|\rh_{t}|= 2|\rh_{t}|\trace(r)$ hence $\p_{t}\bs=  \12 \trace(r) \bs$.  Using \Lemma \ref{lemmatiti.3} we obtain \eqref{iuiu2}. Finally \eqref{iuiu3} follows from Prop.~\ref{prop4.1} (4). \qed
\subsubsection{Gauge invariance}\label{sssect.gaugeinv}
We can write $\hat{d}$ in the form 
\[
\hat{d}= \hat{d}_{0}(t)\p_{t}+ \hat{d}_{1}(t),
\]
where $\hat{d}_{i}\in \cinfb(I; \Diff^{i}(\Sigma; V_{1}, V_{2}))$.  
An easy computation shows that the gauge identity $\hat{D}_{2}\hat{d}= \hat{d}\hat{D}_{1}$ is equivalent to
\begin{equation}
\label{e3.1}
\begin{array}{rl}
i)& \p_{t}\hat{d}_{0}=0,\\[2mm]
ii)& 2 \p_{t}\hat{d}_{1}+ \hat{a}_{2}\hat{d}_{0}- \hat{d}_{0}\hat{a}_{1}=0,\\[2mm]
iii)& \p_{t}^{2}\hat{d}_{1}+ \hat{a}_{2}\hat{d}_{1}- \hat{d}_{1}\hat{a}_{1}- \hat{d}_{0}\p_{t}\hat{a}_{1}=0.
\end{array}
\end{equation}
\section{Hadamard and \calde projectors}\label{sec01}\init

In this section we first revisit the construction of Hadamard projectors for second order hyperbolic equations acting on sections of  Hermitian vector bundles. 

Hadamard projectors  act on distributional Cauchy data and project on Cauchy data whose solutions have wavefront sets in one of the two energy shells $\cN^{\pm}$. If the Hermitian bundle is Hilbertian, i.e.~if the fiber scalar product is positive definite, they produce Hadamard states for the associated quantum fields. In general they produce only Hadamard pseudo-states.

We then consider second order elliptic equations, typically obtained by Wick rotation  of the hyperbolic equations in the time variable, and we construct the associated \calde projectors.  We also study Dirichlet-to-Neumann maps, which will be important in later sections. Finally we show that Hadamard and \calde projectors coincide modulo smoothing operators.

\subsection{$\Psi$DO calculus on manifolds of bounded geometry}\label{sec01.-1}
The constructions in this section rely on a global pseudodifferential calculus on a Cauchy surface $\Sigma$. Namely, we  use  Shubin's calculus which we now quickly recall.

Let $(M, \hat{g})$ be a Riemannian manifold of bounded geometry and $V\xrightarrow{\pi}M$ a finite rank complex vector bundle of bounded geometry. 

One can then define for $m\in \rr$ the symbol classes $S^{m}(T^{*}M; L(V))$ of poly-homogeneous symbols, see \cite{Sh} or \cite[Sect.~5]{GS}. Using a bounded atlas $\{(U_{i}, \psi_{i})\}_{i\in \nn}$ and associated local trivializations of $V$ and partition of unity $1= \sum_{i}\chi_{i}^{2}$ one can define a quantization map
\[
\Op: S^{m}(T^{*}M; L(V))\to L(\coinf(M; V)),
\]
$\Op(a)$ being a (classical) pseudodifferential operator of order $m$.  Choosing a different atlas, trivializations or partition of unity produces of course in general a different quantization map $\Op'$. However,  
\[
\Op(a)- \Op'(a)\in \cW^{-\infty}(M; V),
\]
where
\[
\cW^{-\infty}(M; V)\defeq \medcap_{m\in \nn}B(H^{-m}(M; V), H^{m}(M; V)),
\]
is an ideal of smoothing operators.  
Similarly 
if $\Omega\subset M$ is an open set we set
 \[
 \cW^{-\infty}(\Omega; V)\defeq\medcap_{m\in \nn}B(H^{-m}(\Omega; V),H^{m}(\Omega; V)),
 \]
and   if $(M_{i}, \hat{g}_{i})$ and $V_{i}\xrightarrow{\pi}M_{i}$ are of bounded geometry: \[
\cW^{-\infty}(M_{1}, M_{2}; V_{1}, V_{2})\defeq \medcap_{m\in \nn}B(H^{-m}(M_{1}; V_{1}), H^{m}(M_{2}; V_{2})).
 \]
One sets then
\[
\Psi^{m}(M; V)= \Op(S^{m}(T^{*}M; L(V)))+ \cW^{-\infty}(M; V),
\]
and $\Psi^{\infty}(M; V)= \bigcup_{m\in \rr}\Psi^{m}(M; V)$.  

We refer the reader to \cite[App.~1]{Sh} and \cite[Sect.~5]{GS} for more details.

 \subsection{Lorentzian case}\label{sec01.2}
We set $M= I_{t}\times \Sigma_{\rx}$, where $ I\subset \rr$ is an interval with $0\in \mathring{I}$ and $(\Sigma, \rh_{0})$ a $d$-dimensional Riemannian manifold of bounded geometry. We set $\Sigma_{t}= \{t\}\times \Sigma$ and identify $\Sigma_{0}$ with $\Sigma$.
The dual variables to $(t, \rx)$ are denoted by $(\tau, {\rm k})$.

 We fix a $t$-dependent Riemannian metric on $\Sigma$,
 \[
 \rh:I \ni t\mapsto \rh(t)\in \cinfb(I; \BT_{2}^{0}(\Sigma, \rh_{0})).
 \]
  We assume that $\rh(0)= \rh_{0}$ and for ease of notation we often denote  $\rh(t)$  by $\rh_{t}$.
 
 We equip $M$ with the Lorentzian metric
 \[
 \rg\defeq -dt^{2}+ \rh_{t}d\rx^{2}.
 \]

\subsubsection{Hermitian bundle}\label{sec01.2.1}
We fix a finite rank complex vector bundle $V\xrightarrow{\pi}\Sigma$ of bounded geometry  over $(\Sigma, \rh_{0})$.  We still denote by $V$  the vector bundle over $M$: $I\times V\xrightarrow{\pi} M$  which is a vector bundle  with the same fibers as $V$. 
We have for example
\[
\cinfb(M; V)\sim \cinfb(I; \cinfb(\Sigma; V)).
\]
We assume that  $V\xrightarrow{\pi}M$ is equipped with a non-degenerate fiberwise {\em Hermitian} structure $(\cdot| \cdot)_{V}$, which is assumed to be  independent of $t$.
 
 We  fix a reference fiberwise {\em Hilbertian} structure $(\cdot| \cdot)_{\tV}$ on the fibers of $V$ which is also independent of $t$.

 We denote by $(\cdot| \cdot)_{\tV}$, $(\cdot| \cdot)_{V}$ the same Hermitian structures acting on  the fibers of $V\xrightarrow{\pi}\Sigma$. If $x\in M$ and $u, v\in V_{x}$ we have 
 \[
 (u| v)_{V}= (u| \tau_{x}v)_{\tV},  \ \ \tau_{x}\in L(V_{x}),
 \]
 and we denote by $\tau\in \cinf(M; L(V))$ the corresponding section, which is independent on $t$.
  
 Note that $\tau= \tau^{*}$. By polar decomposition, after possibly changing $(\cdot| \cdot)_{\tV}$, we can assume that
 \[
 \tau^{*}\tau= \one,\hbox{ i.e. } \tau\hbox{ is unitary for }(\cdot| \cdot)_{\tV}.
 \]
 We assume that the Hermitian structures $(\cdot| \cdot)_{V}$ and  $(\cdot| \cdot)_{\tV}$, and hence $\tau$, are of bounded geometry.
 
  If $a\in L(V_{x})$ for  $x\in M$ we denote by $a^{*}$, resp.~ $a^{\star}$, the adjoints of $a$ for $(\cdot| \cdot)_{\tV}$, resp.~ $(\cdot| \cdot)_{V}$. Then,
   \begin{equation}
 \label{e01.1}
 a^{\star}= \tau^{-1}_x a^{*}\tau_x \end{equation}
for some $\tau_{x}\in L(V_{x})$.

For $u, v\in \cinf_{\rm sc}(M; V)$ we set
 \[
 \begin{array}{l}
  (u|v)_{\tV(\Sigma_{t})}\defeq\int_{\Sigma_{t}}(u|v)_{\tV}|\rh_{0}|^{\12}d\rx,\\[2mm]
  (u|v)_{V(\Sigma_{t})}\defeq\int_{\Sigma_{t}}(u|v)_{V}|\rh_{0}|^{\12}d\rx= (u| \tau v)_{\tV(\Sigma_{t})},\\[2mm]
  (u| v)_{\tV(M)}\defeq\int_{M}(u(t)|v(t))_{\tV}|\rh_{0}|^{\12}dtd\rx,\\[2mm]
 (u| v)_{V(M)}\defeq\int_{M}(u(t)|v(t))_{V}|\rh_{0}|^{\12}dtd\rx= (u| \tau v)_{\tV(M)}.
\end{array}
 \]
If $\Omega\subset M$ is some open set, we  also denote 
\[
\begin{array}{l}
 (u| v)_{\tV(\Omega)}\defeq\int_{\Omega}(u(t)|v(t))_{\tV}|\rh_{0}|^{\12}dtd\rx,\\[2mm]
 (u| v)_{V(\Omega)}\defeq\int_{\Omega}(u(t)|v(t))_{V}|\rh_{0}|^{\12}dtd\rx= (u| \tau v)_{\tV(\Omega)}.
\end{array}
\]
We denote by $L^{2}(\Sigma; \tV)$ the $L^{2}$ space obtained from the Hilbertian scalar product $(\cdot| \cdot)_{\tV(\Sigma)}$.
 \subsubsection{Adjoints}\label{sec01.2.2}
 If  $a\in \cinfb(I; \Diff(\Sigma; V))$, resp.~ $A\in \Diff(M; V)$, we denote by  $a^{*}$ resp.~ $A^{*}$ its formal adjoint for $(\cdot| \cdot)_{\tV(\Sigma_{t})}$ resp.~$(\cdot| \cdot)_{\tV(M)}$.  We set $\Re a= \12 (a+ a^{*})$. 
 
 We denote by $a^{\star}$ resp.~ $A^{\star}$ its formal adjoint for $(\cdot| \cdot)_{V(\Sigma_{t})}$ resp.~$(\cdot| \cdot)_{V(M)}$. As above we have:
 \[
  a^{\star}= \tau^{-1} a^{*}\tau, \quad A^{\star}= \tau^{-1}A^{*}\tau.
 \]

 \subsubsection{Hyperbolic operator}\label{sec01.2.3}
 We fix a $t$-dependent differential operator $a= a(t, \rx, D_{\rx})$ belonging to $\cinfb(I; \Diff^{2}(\Sigma; V))$ and  denote by $\sigma_{\rm pr}(a)\in \cinf(T^{*}\Sigma; L(V))$ its principal symbol.
 
 We assume the following properties:
\[
 \begin{array}{rl}
{\rm (H1)}& \ a(t) = a^{\star}(t),\  t\in I, \\[2mm]
{\rm (H2)}&\sigma_{\rm pr}(a)(t)(\rx, {\rm k})= {\rm k}\dual \rh_{t}^{-1}(\rx){\rm k}\,\one_{V}, \ t\in I.
\end{array}
\]
We set
\beq\label{e01.1b}
D\defeq \p_{t}^{2}+ a(t)\hbox{ acting on }\coinf(M; V),
\eeq
which is a hyperbolic operator with scalar principal part. Note that 
 $D= D^{\star}$, but of course $D\neq D^{*}$ in general.
 \subsubsection{Green's formula}\label{sec01.2.2b}
  We set
 \beq\label{e01.1bb}
 \varrho u= \begin{pmatrix}{u(0)} \\{\i^{-1}\p_{t}u(0)}\end{pmatrix}, \ \ u\in \cinfsc(M; V),
 \eeq
 and  equip  $\coinf(\Sigma; V\otimes \cc^{2})$ with the Hilbertian scalar product $(\cdot| \cdot)_{\tV(\Sigma)\otimes \cc^{2}}$ defined by
 \[
 (f|f)_{\tV(\Sigma)\otimes \cc^{2}}\defeq (f_{0}| f_{0})_{\tV(\Sigma)}+ (f_{1}| f_{1})_{\tV(\Sigma)}.
 \]
The Green identity in \Lemma \ref{lemmatiti.1b} takes the form:
\begin{equation}
  \label{e01.3z}
  (u| Dv)_{V(J_{\pm}(\Sigma))}- (Du| v)_{V(J_{\pm}(\Sigma))}= \pm \i^{-1}(\varrho u| q \varrho v)_{\tV(\Sigma)\otimes\cc^{2}},
  \end{equation}
  where 
 \begin{equation}
 \label{e01.3d}
 q= \mat{0}{\tau}{\tau}{0}.
 \end{equation}
The Hilbertian vector bundle $V_{\Sigma}$ in \ref{greeno} equals $\tV(\Sigma)\otimes \cc^{2}$, since the operator $D$ is of second order.
\subsubsection{Symplectic adjoint}\label{sss.sympl-adj}
 If $A$ is an operator acting on $\coinf(\Sigma; \tV\otimes \cc^{2})$ we denote by $A^{\dag}$ its adjoint for $q$, i.e.:
 \[
 A^{\dag}= q^{-1}A^{*}q.
 \]

  \subsubsection{Square root of $a(t)$}\label{sec01.2.3b} We first construct an approximate square root $a(t)$ adapted to our future needs.
 
\begin{lemma}\label{lemma01.1}
There exist $r_{-\infty}\in \cinfb(I; \Psi^{-\infty}(\Sigma; V))$ and $\epsilon\in \cinfb(I; \Psi^{1}(\Sigma; V))$ such that:
\ben
\item $\epsilon= \epsilon^{\star}, \ \epsilon^{2}= a+ r_{-\infty}$,
\item $\sigma_{\rm pr}(\epsilon)=   (k\dual \rh_{t}^{-1}(\rx)k)^{\12}\one_{V}$,
\item $\epsilon$ with domain $H^{1}(\Sigma; V)$ is $m$-accretive with $\Re \epsilon\geq 1$.
\een
\end{lemma}
\proof 
  Note first that $a$ is uniformly elliptic in $\cinfb(I; \Psi^{2}(\Sigma; V))$ hence has a parametrix $a^{(-1)}\in \cinfb(I; \Psi^{-2}(\Sigma; V))$. Therefore $a(t)$ is closed with domain $H^{2}(\Sigma; \tV)$ and $\Dom a^{*}(t)= H^{2}(\Sigma; \tV)$. We set 
\[
2a_{\rm ref}\defeq \12(a+ a^{*})+ \tau^{-1}\12(a+ a^{*})\tau.
\]
It satisfies:
\begin{equation}
\label{e01.2}
\begin{array}{rl}
i)&a- a_{\rm ref}\in \cinfb(I; \Diff^{1}(\Sigma; V)),\\[2mm]
ii)&a_{\rm ref}= a_{\rm ref}^{*}= a_{\rm ref}^{\star}.\\[2mm]
\end{array}
\end{equation}
In fact, \eqref{e01.2} {\it i)} follows from (H2), and \eqref{e01.2} {\it ii)} from the fact that $\tau= \tau^{*}= \tau^{-1}$.

Clearly $a_{\rm ref}$ with domain $H^{2}(\Sigma; \tV )$ is self-adjoint for $(\cdot| \cdot)_{\tV(\Sigma)}$. We fix  $\chi\in \coinf(\rr)$ with $\chi(0)=1$ and  set $\chi_{R}(\lambda)= \chi(R^{-1}\lambda)$ for $R\geq 1$.    Since $a_{\rm ref}$ is  elliptic we know that $\chi_{R}(a_{\rm ref})\in \cinfb(I; \Psi^{-\infty}(\Sigma; V))$. 

We set now
\begin{equation}
\label{e01.2b}
r_{-\infty}=  R\chi_{R}(a_{\rm ref}),
\end{equation}
where $R\gg 1$ will be chosen below. From \eqref{e01.2}  we deduce that
\begin{equation}
\label{e01.2c}
r_{-\infty}= r_{-\infty}^{\star}= r_{-\infty}^{*}\end{equation}
and
\begin{equation}
\label{e01.2d}
\Re (a+ r_{-\infty}) = a_{\rm ref}+ R\chi_{R}(a_{\rm ref})+ a_{1}
\end{equation}
for some $a_{1}\in \cinfb(I; \Diff^{1}(\Sigma; V))$.

By  the self-adjoint functional calculus we can find $R\gg 1$ such that:
\beq\label{e01.3e}
a_{\rm ref} +  R\chi(R^{-1}a_{\rm ref})\geq \12 a_{\rm ref}+ 1,
\eeq
and hence  by \eqref{e01.2d}
\begin{equation}
\label{e01.2e}
 \Re (a+ r_{-\infty})\geq 1.
\end{equation}
The same inequality is valid for $(a+ r_{-\infty})^{*}$. If follows that For $\Re z\geq 0 $, $a+ r_{-\infty}+ z: H^{2}(\Sigma; \tV)\to L^{2}(\Sigma; \tV)$ is injective with a dense range.   \eqref{e01.2e} also implies that $\Ran( a+ r_{-\infty}+z)$ is closed, using that $a$ is closed. Therefore $a+ r_{-\infty}$ is $m$-accretive. By \cite[Thm.~V.3.35]{K1} $a+ r_{-\infty}$ has a unique $m$-accretive square root:
\beq\label{e01.2f}
\bea
\epsilon &\defeq  \pi^{-1}\int_{0}^{+\infty} \lambda^{-\12}(a+ r_{-\infty}+ \lambda)^{-1}(a+ r_{-\infty})d\lambda\\[2mm]
&= 2\pi^{-1}\int_{0}^{+\infty}(a+ r_{-\infty}+ s^{2})^{-1}(a+ r_{-\infty})d\lambda,
\eea
\eeq
where the integrals are strongly convergent on $\Dom a = H^{2}(\Sigma; \tV)$.
 By \cite[Pb.~V.3.39]{K1} we have
 \beq\label{e01.10}
 \Re\epsilon\geq 1.
 \eeq
Arguing  as in  \cite[Subsect.~5.3]{GS}, using the representation of $\epsilon$ in the second line of \eqref{e01.2f},  we obtain that  $\epsilon\in \cinfb(I; \Psi^{1}(\Sigma; V))$ with 
 \[
 \sigma_{\rm pr}(\epsilon)= (k\dual \rh_{t}^{-1}(\rx)k)^{\12}\one_{V}.
  \]
 The operator $\epsilon(t)$ with domain $H^{1}(\Sigma; \tV)$ is closed, elliptic, $m$-accretive   and invertible by \eqref{e01.10}, hence $\epsilon^{-1}\in \cinfb(I; \Psi^{-1}(\Sigma; V))$.
From \eqref{e01.2f} we obtain that
$\epsilon= \epsilon^{\star}$. \qed

 \subsubsection{Factorization of $D$}\label{sec01.2.4}
\begin{proposition}\label{prop01.1}
 There exists $b\in \cinfb(I; \Psi^{1}(\Sigma; V))$ unique modulo $\cinfb(I; \Psi^{-\infty}(\Sigma; V))$ such that
 \[
 \begin{array}{rl}
 i)& b= \epsilon+ \cinfb(I; \Psi^{0}(\Sigma; V)),\\[2mm]
 ii)&\i\p_{t}b- b^{2}+ a= r_{-\infty}\in \cinfb(I; \Psi^{-\infty}(\Sigma; V)),\\[2mm]
 iii)&\begin{array}{l}
 b+ b^{\star}= (2\epsilon)^{\12}(\one + r_{-1})^{2}(2\epsilon)^{\12}, \\[2mm]
  r_{-1}\in \cinfb(I; \Psi^{-1}(\Sigma; V)),\ r_{-1}= r_{-1}^{\star},  \ \| r_{-1}\|\leq \frac{2}{3}.
\end{array}
  \end{array}
 \]
 
\end{proposition}
\begin{proof} We first solve {\it i)} and {\it ii)}.  Let  $b_{0}\in \cinfb(I, \Psi^{0}(\Sigma; V))$. A routine computation shows that
\[
\i\p_{t}(\epsilon+ b_{0})- (\epsilon+ b_{0})^{2}+ a= 0
\]
iff 
\beq
\label{e01.3c} b_{0}= (2\epsilon)^{-1}\i \p_{t}\epsilon+ F(b_{0}),
\eeq
for
\[ F(b_{0})= (2\epsilon)^{-1}(\i \p_{t}b_{0}+ [\epsilon, b_{0}]- b_{0}^{2}). 
\]
By  Prop.~\ref{prop:fixed-point} we find  $b_{0}\in \cinfb(I, \Psi^{0}(\Sigma; V))$, unique modulo $\cinfb(I, \Psi^{-\infty}(\Sigma; V))$ solving  \eqref{e01.3c}  modulo $\cinfb(I, \Psi^{-\infty}(\Sigma; V))$.

Next we take $\chi$ as in \eqref{e01.3e} and set
\[
b=  \epsilon+ b_{0}(1- \chi_{R}(a_{\rm ref})),
\]
where $R\gg 1$ will be fixed below.
Since $\chi_{R}(a_{\rm ref})\in \cinfb(I; \Psi^{-\infty}(\Sigma; V))$, we see that $b$ satisfies {\it i)} and {\it ii)}.  

It remains to check {\it iii)}.
 By the same argument as in \Lemma \ref{lemma01.1} we can construct the square root $(2\epsilon)^{\12}$, which belongs to $\cinfb(I; \Psi^{\12}(\Sigma; V))$ and is invertible with $(2\epsilon)^{\12\star}= (2\epsilon)^{\12}$. We have
 \[
 b+ b^{\star}= (2\epsilon)^{\12}(1- s_{-1})(2\epsilon)^{\12},
 \]
 where
 \[
 s_{-1}= -(2\epsilon)^{-\12}(b_{0}(1-\chi_{R}(a_{\rm ref}))+ (1-\chi_{R}(a_{\rm ref}))b_{0}^{\star})(2\epsilon)^{-\12}.
 \]
 We see that $s_{-1}= s_{-1}^{\star}$, $s_{-1}\in \cinfb(I; \Psi^{-1}(\Sigma; V))$ and since $(1-\chi_{R}(a_{\rm ref}))(2\epsilon)^{-\12}$ tends to $0$ in norm when $R\to \infty$ we can fix $R\gg 1$ such that $\|s_{-1}\|\leq \frac{1}{4}$.  We set now
 \[
 1+r_{-1}= (1- s_{-1})^{\12}= \sum_{n=0}^{\infty}c_{n}(s_{-1})^{n},
 \]
where  $c_{n}= \frac{f^{(n)}(0)}{n!}$ with  $f(x)= (1-x)^{\12}$ satisfies $|c_{n}|\leq 2$ for $n\in\nn$. It follows that
 \[
 \|r_{-1}\|\leq 2 \sum_{n=1}^{\infty}\| s_{-1}\|^{n}= 2 \|s_{-1}\|(1- \| s_{-1}\|)^{-1}\leq \frac{2}{3}.
 \]
 Moreover 
 \[
 r_{-1}\in \cinfb(I; \Psi^{-1}(\Sigma; V)), \quad r_{-1}= r_{-1}^{\star}, \quad (1+ r_{-1})^{2}= (1-s_{-1}).
 \]
This proves {\it iii)}.  \end{proof}

\medskip

We now set
\beq\label{e01.21c}
b^{+}\defeq b, \quad b^{-}= - b^{\star}, 
\eeq
and obtain that
\begin{equation}
\label{e01.4}
\begin{array}{l}
b^{\pm}= \pm \epsilon+ \cinfb(I; \Psi^{0}(\Sigma; V)), \\[2mm]
\i \p_{t}b^{\pm}- (b^{\pm })^{2}+ a= r_{-\infty}^{\pm}, 
\end{array}
\end{equation}
for $r^{+}_{-\infty}= r_{-\infty}$, $r^{-}_{-\infty}= r_{-\infty}^{\star}$.
This is equivalent to  the two factorizations of $D$ modulo smoothing error terms:
\begin{equation}
\label{e01.5}
(\p_{t}+ \i b^{\pm})(\p_{t}- \i b^{\pm})= D - r_{-\infty}^{\pm}.
\end{equation}

 \subsubsection{Cauchy evolution}\label{sec01.2.5}
 For $s\in I$ the Cauchy problem
 \begin{equation}
 \label{e01.3b}
\begin{cases}
 Du= 0\hbox{ in }M\\
 \varrho_{s}u= f\in  \coinf(\Sigma; V\otimes \cc^{2})
\end{cases}
 \end{equation}
 is well-posed, where  
 \[
 \varrho_{s}u= \col{u(s)}{\i^{-1}\p_{t}u(s)}.
 \]
 We denote by $u= U_{s}f$ the unique solution of \eqref{e01.3b}, so that $U_{s}: \coinf(\Sigma; V\otimes \cc^{2})\to \cinf_{\rm sc}(M; V)$. For $t\in I$ we denote by 
 \[
 U(t, s)\defeq \varrho_{t}\circ U_{s}: \coinf(\Sigma; V\otimes \cc^{2})\to \coinf(\Sigma; V\otimes \cc^{2})
 \]
 the Cauchy evolution of $D$.  The Green identity \eqref{e01.3z} implies that $U(t,s)$ is pseudo-unitary for $(\cdot|q \cdot)_{\tV(\Sigma)\otimes \cc^{2}}$:
  \[
  q= U(t,s)^{*}q U(t,s), \hbox{ or equivalently }U(t,s)^{\dag}= U(s, t)\ t, s\in I,
  \]
  where as before $A^{*}$ denotes the adjoint of $A$ for $(\cdot| \cdot)_{\tV(\Sigma)\otimes \cc^{2}}$.

    \subsubsection{Factorization of the Cauchy evolution}\label{sec01.2.6}
  For a solution  $u\in\cinf_{\rm sc}(M; V)$ of $Du=0$ we set $\psi(t)= \col{u(t)}{\i^{-1}\p_{t}u(t)}$ so that
  \[
  \p_{t}\psi(t)= \i A(t)\psi(t), \quad A(t)= \mat{0}{\one}{a(t)}{0}, 
  \]
and $\psi(t)= U(t,s)\psi(s)$. Next, we define $S(t)$ by  
  \[
 S^{-1}(t)\psi(t)= \col{(\p_{t}-\i b^{-}(t)) u(t)}{(\p_{t}-\i b^{+}(t))u(t)},
  \]
which yields
\[
S= \i^{-1}\mat{\one}{-\one}{b^{+}}{-b^{-}}(b^{+}- b^{-})^{-1}, \quad S^{-1}= \i \mat{- b^{-}}{\one}{-b^{+}}{\one}.
\]
From \eqref{e01.5} we obtain that
\[
\p_{t}S^{-1}(t)\psi(t)= \i B(t) S^{-1}(t)\psi(t), 
\]
for
\[ B= \mat{- b^{-}}{0}{0}{- b^{+}}+B_{-\infty},
\]
where
\[B_{-\infty}=  \mat{r_{-\infty}^{-}}{- r_{-\infty}^{-}}{r_{-\infty}^{+}}{-r_{-\infty}^{+}}(b^{+}- b^{-})^{-1}\in \cinfb(I; \Psi^{-\infty}(\Sigma; V\otimes \cc^{2})).
\]
We find
\[
S^{*}q S= \tau(b^{+}- b^{-})^{-1}\mat{\one}{0}{0}{-\one}.
\]
Let $c= (2\epsilon)^{\12}(1+ r_{-1})$, where $r_{-1}$ is as in Prop.~\ref{prop01.1} {\it iii)}. Then $c^{\star}(b^{+}- b^{-})^{-1}c=1$,  hence 
$c^{*}\tau(b^{+}-b^{-})^{-1}c= \tau$. Setting 
\[
T\defeq S\circ\mat{c}{0}{0}{c},
\]
we have
\begin{equation}
\label{e01.7}
T^{*}q T= \mat{\tau}{0}{0}{-\tau}.
\end{equation}
Moreover  we have
\[
\p_{t}T^{-1}(t)\psi(t)= \i C(t)T^{-1}(t)\psi(t), 
\]
for
\[
C= \mat{\epsilon^{+}}{0}{0}{\epsilon^{-}}+ C_{-\infty}(t)
\]
where
\[
\epsilon^{\pm}= \pm c^{-1}b^{\pm}c+ \i c^{-1}\p_{t}c= \pm \epsilon+ \cinfb(I; \Psi^{0}(\Sigma; V)),
\]
and
\[
C_{-\infty}= c^{-1}B_{-\infty}c\in \cinfb(I; \Psi^{-\infty}(\Sigma; V\otimes \cc^{2})).
\]
In the next proposition we use the notation recalled in \ref{sec01.1.4}. The hypotheses of Kato's theorem are easy to check using $\Psi$DO calculus.
\begin{proposition}\label{prop01.2}
 For all $t,s\in I$ we have
 \[
\bea
 U(t,s)&= T(t)\Texp(\i \textstyle\int_{s}^{t}C(\sigma)d\sigma)T^{-1}(s)\\[2mm]
 &= T(t)\mat{\Texp(\i\int_{s}^{t}\epsilon^{+}(\sigma)d\sigma)}{0}{0}{\Texp(\i\int_{s}^{t}\epsilon^{-}(\sigma)d\sigma)}T^{-1}(s)+ R_{-\infty}(t,s).
\eea
  \]
 
\end{proposition}
\subsubsection{Hadamard projectors}\label{sec01.2.7}
We set $\pi^{+}= \mat{\one}{0}{0}{0}$, $\pi^{-}= \mat{0}{0}{0}{\one}$, and
\beq\label{e01.21b}
c^{\pm}\defeq T(0)\pi^{\pm}T^{-1}(0)=  \mat{\mp(b^{+}- b^{-})^{-1}b^{\mp}}{\pm(b^{+}- b^{-})^{-1}}{\mp b^{+}(b^{+}- b^{-})^{-1}b^{-}}{\pm b^{\pm}(b^{+}- b^{-})^{-1}}(0).
\eeq
We recall that $A^{\dag}$ is the symplectic adjoint of $A$, see \ref{sss.sympl-adj}.
\begin{proposition}\label{prop01.2b1}
The operators $c^\pm$ defined in \eqref{e01.21b} satisfy:
 \ben
  \item $c^{+}+ c^{-}= \one$,
  \item $(c^{\pm})^{\dag}= c^{\pm}$,
  \item $\WF(U(\cdot, 0)c^{\pm})'\subset (\cN^{\pm}\cup \cF)\times T^{*}\Sigma$ for $\cF= \{k=0\}\subset T^{*}M.$
 \een
 \end{proposition}
 
 \proof (1) is straightforward; (2) follows from \eqref{e01.7}. We set  $Q_{\pm}= (\p_{t}- \i \epsilon^{\pm}(t, \rx, \p_{\rx}))$, considered as an operator acting on $M\times \Sigma$ on the first group of variables and let $A(t, \rx, \rx')$ the distributional kernel of $U(\cdot, 0)c^{\pm}$. 

Prop.~\ref{prop01.2} it follows that $Q_{\pm}A\in \cinf(M\times \Sigma, L(V\otimes\cc^{2}, V))$.  If $Q_{\pm}$ were classical $\Psi$DOs on $M\times \Sigma$, this would imply that $\WF(U(\cdot, 0)c^{\pm})'\subset \cN^{\pm}\times T^{*}\Sigma$ by elliptic regularity. We reduce ourselves to this situation by an argument from \cite[Lem.~6.5.5]{DH}, see for example \cite[Prop.~6.8]{GS} for details. \qed\smallskip
 
 We call the maps $c^{\pm}$ {\em Hadamard projectors}.
 
\begin{remark}
 From \eqref{e01.21b} we obtain immediately that $c^{\pm}$ are projections  indeed. The terminology  is   justified by the fact that  Prop.~\ref{prop01.2b1} implies that $\lambda^{\pm}= \pm q\circ c^{\pm}$  are a pair of Cauchy surface \emph{Hadamard} pseudo-covariances for $D$. Note however that the positivity condition $\lambda^{\pm}\geq 0$ for $(\cdot| \cdot)_{\tV(\Sigma)\otimes \cc^{2}}$ is in general not satisfied. Moreover, different choices of $b$ in Prop.~\ref{prop01.1} lead to different projections  $c^{\pm}$, differing by a term in $\Psi^{-\infty}(\Sigma; V\otimes\cc^{2})$.
\end{remark}

\subsection{Euclidean case}\label{sec01.3} We now consider a Euclidean analogue of the setting considered so far. 
We set  $\tilde{M}= I_{s}\times \Sigma_{\rx}$, where $I\subset \rr$ is an interval with $0\in \mathring{I}$ and $\Sigma$  a $d$-dimensional manifold. As before we identify $\{0\}\times \Sigma$ with $\Sigma$. We fix an $s$-dependent sesquilinear form:
\[
\tilde{\rh}: I\ni s\mapsto \cinfb(I; L(T\Sigma, T\Sigma^{*})),
\]
such that $\trh(0)$ is a Riemannian metric on $\Sigma$, i.e.~$\trh(0)= \trh(0)^{*}$, $\trh(0)>0$.  For ease of notation $\trh(s)$ is often denoted by $\trh_{s}$.

We assume that $\trh$ is {\em uniformly coercive}, i.e. there exists $C>0$ such that:
\begin{equation}
\label{e01.8}
\begin{array}{l}
C^{-1}\trh(0)\leq\Re \trh(s)\leq C \trh(0) \\[2mm]
\big|\! \Im \trh(s)\big|\leq C \Re\trh(s), \ \  s\in I. 
\end{array}
\end{equation}
\subsubsection{Hilbertian bundle}\label{sec01.3.1}
We equip $\tilde{M}$ with the Hilbertian bundle $\tV$ as in \ref{sec01.2.1}.
For $u,v\in \coinf(\Sigma; \tV)$ resp.~$\coinf(\tM; \tV)$ we set
\[
\begin{array}{l}
(u| v)_{\tV(\Sigma)}\defeq \int_{\Sigma}(u| v)_{\tV}|\trh_{0}|^{\12}d\rx, \\[2mm]
(u| v)_{\tV(\tM)}\defeq\int_{\tM}(u| v)_{\tV}|\trh_{0}|^{\12}dtd\rx.
\end{array}
\]
\subsubsection{Adjoints}\label{sec01.3.2}
As in \ref{sec01.2.2} if $\ta(s)\in \cinfb(I; \Diff(\Sigma; \tV)$, resp.~$\tA\in \Diff(\tM; \tV)$ we denote by $\ta^{*}(s)$ resp.~$\tA^{*}$ its formal adjoint for $(\cdot| \cdot)_{\tV(\Sigma)}$ resp.~$(\cdot| \cdot)_{\tV(\tM)}$.

\subsubsection{Elliptic operator}\label{sec01.3.3}
We fix an $s$-dependent differential operator $\ta(s)= \ta(s, \rx, D_{\rx})$ belonging to $\cinfb(I; \Diff^{2}(\Sigma; \tV))$ and denote by $\sigma_{\rm pr}(\ta)(s)$ its principal symbol. We assume the following property:
\[
\begin{array}{rl}
{\rm (\tilde{H}1)} \quad \sigma_{\rm pr}(\ta)(s)(\rx, {\rm k})= {\rm k}\dual \trh_{s}^{-1}(\rx){\rm k}\one_{\tV}.
\end{array}
\]
We set
\beq\label{e01.8b}
\tD\defeq -\p_{s}^{2}+ \ta(s)\hbox{ acting on }\coinf(\tM; \tV),
\eeq 
which is an elliptic differential operator. 

\subsubsection{Factorization of $\tD$}\label{sec01.3.4}
As in \Lemma \ref{lemma01.1}, we see that $\ta(s)$ is closed with domain $H^{2}(\Sigma; \tV)$ and $\Dom \ta^{*}(s)= H^{2}(\Sigma; \tV)$.

We add to $\ta(s)$ a self-adjoint term $\tilde{r}_{-\infty}\in \cinfb(I; \Psi^{-\infty}(\Sigma; \tV))$ such that 
\beq\label{e01.9}
\Re \ta(s)+ r_{-\infty}(s)\geq \delta \one, \ \  \delta>0, 
\eeq
and  $\ta(s)+ r_{-\infty}(s)$ is $m$-accretive. We denote by 
\[
\teps\defeq (\ta+ r_{-\infty})^{\12}\in \cinfb(I; \Psi^{1}(\Sigma; \tV))
\]
  its unique $m$-accretive square root  given by \eqref{e01.2f}, which satisfies:
 \beq\label{e01.10b}
 \Re\teps(s)\geq \delta^{\12}\one.
 \eeq
As in \ref{sec01.2.3b} we have:
 \[
 \sigma_{\rm pr}(\teps(s))= (\sigma_{\rm pr}(\ta(s)))^{\12}.
 \]
 The operator $\teps(s)$ with domain $H^{1}(\Sigma; \tV)$ is closed, elliptic   and invertible by \eqref{e01.10b}, hence $\teps^{-1}(s)\in \cinfb(I; \Psi^{-1}(\Sigma; \tV))$.

\begin{proposition}\label{prop01.2bbb}
 There exists $\tb^{\pm}(s)\in\cinfb(I; \Psi^{1}(\Sigma; \tV))$, unique modulo a term in $\cinfb(I; \Psi^{-\infty}(\Sigma; \tV))$,  such that:
 \[
 \begin{array}{rl}
 i)& \tb^{\pm}(s)= \pm\teps(s)+ \cinfb(I; \Psi^{0}(\Sigma; \tV)),\\[2mm]
 ii)&\p_{s}\tb^{\pm}(s)- (\tb^{\pm})^{2}(s)+ \ta(s)= \tr^{\pm}_{-\infty}(s)\in \cinfb(I; \Psi^{-\infty}(\Sigma; \tV)),\\[2mm]
 iii)&\pm\Re \tb^{\pm}\geq \one,\\[2mm]
 iv)&\tb^{+}- \tb^{-}: H^{1}(\Sigma; \tV)\to L^{2}(\Sigma; \tV)\hbox{ is invertible,  }(\tb^{+}- \tb^{-})^{-1}\in \cinfb(I; \Psi^{-1}(\Sigma; \tV)).
 \end{array}
 \]
 \end{proposition}
 \proof We  look for $\tb^{\pm}$ under the form $\tb^{\pm}= \pm\teps+ \tb_{0}$, $\tb_{0}\in \cinfb(I; \Psi^{0}(\Sigma; \tV))$ and obtain the equation
 \[
 \tb_{0}= (2 \teps)^{-1}\p_{s}\teps+ \tilde{F}^{\pm}(\tb_{0}), 
 \]
 for 
 \[
 \tilde{F}^{\pm}(\tb_{0})= (2\teps)^{-1}(\pm\p_{s}\tb_{0}+ [\teps, \tb_{0}]\mp\tb_{0}^{2}).
 \]
 We use the same fixed point argument as in Prop.~\ref{prop01.1} and obtain $\tb^{\pm}$ satisfying {\it i)} and {\it ii)}. To obtain {\it iii)} we use that $\tb^{\pm}= \pm\teps+ \cinfb(I; \Psi^{0}(\Sigma; \tV))$  and add to $\tb^{\pm}$ elements  $r_{1, -\infty}^{\pm}\in \cinfb(I, \Psi^{-\infty}(\Sigma; \tV))$ so that 
 $\pm \Re \tb^{\pm}\geq 1$. Then $\tb^{+}- \tb^{-}$ is $m$-accretive with $0\not\in \sigma(\tb^{+}- \tb^{-})$, which implies {\it iv)}.
  \qed
 
 We obtain the following factorization of $\tD$, analogous to \eqref{e01.5}:
 \begin{equation}
 \label{e01.11}
  (-\p_{s}+ \tb^{\pm})(\p_{s}+ \tb^{\pm})= \tD - \tr^{\pm}_{-\infty}.
 \end{equation}

\begin{remark}\label{remark01.1}
 Suppose that the interval $I$ is symmetric with respect to $0$ and that
 \[
 \ta^{*}(s)= \tau\ta(-s)\tau^{-1}, \ \ s\in I.
 \]
 Then from \eqref{e01.2f} we have $\teps^{*}(s)= \tau\teps(-s)\tau^{-1}$.  We set
 \[
 \tilde{F}(\tb_{0})= (2\teps)^{-1}(\p_{s}\tb_{0}+ [\teps, \tb_{0}]- \tb_{0}^{2}),
 \]
 we solve the fixed point equation
 \[
 \tb_{0}= (2\teps)^{-1}\p_{s}\teps + \tilde{F}(\tb_{0}),
 \]
  and construct  $\tb= \teps+ \tb_{0}$ such that 
  \[
  \begin{array}{l}
  \p_{s}\tb- \tb^{2}+ \ta= \tr_{-\infty}\in \cinfb(I; \cW^{-\infty}(\Sigma; \tV)),\\[2mm]
  \Re \tb\geq \one. 
  \end{array}
  \]
  Then we can take:
  \beq\label{e01.21d}
  \begin{array}{l}
  \tb^{+}(s)= \tb(s), \quad \tr^{+}_{-\infty}(s)= \tr_{-\infty}(s),\\[2mm]
\tb^{-}(s)= - \tau^{-1}\tb^{*}(-s)\tau, \quad \tr_{-\infty}^{-}(s)= \tau^{-1}\tr_{-\infty}^{*}(-s)\tau.
\end{array}
\eeq
 \end{remark}
\subsubsection{Parametrix for $\tD$}\label{sec01.3.5}
By Prop. \ref{prop01.2bbb} we know that   $\pm\tb^{\pm}(s)$ is $m$-accretive with $0\in \rs(\tb^{\pm}(s))$ (the resolvent set)  and $\Dom \tb^{\pm}(s)= H^{1}(\Sigma; \tV)$ for all $s\in I$. Since   $\tb^{\pm}\in \cinfb(I; \Psi^{1}(\Sigma; \tV))$, we can check the hypotheses of \cite{K2} (in the version presented in \cite{SG}) and conclude that
\[
V^{\pm}(s, s')\defeq \Texp(\textstyle\int_{s'}^{s}\tb^{\pm}(\sigma)d\sigma)\hbox{ exists for }\mp (s-s')\geq 0, \ s,s'\in I.
\]
\begin{lemma}\label{lemma01.2}
\ben
\item  $V^{\pm}(s, s'): H^{m}(\Sigma; \tV)\to H^{m}(\Sigma; \tV)$ is uniformly bounded for $s, s'\in I$ and  $\pm(s-s')\geq 0$;
 \item $V^{\pm}(s, s')\in \cW^{-\infty}(\Sigma; \tV)$ for $\mp(s-s')\geq \delta>0$.
\een
\end{lemma}
\proof  To prove (1) it suffices to apply Kato's theorem in \cite{K2} to the Hilbert space $H^{m}(\Sigma; \tV)$.  The hypotheses follow from $\Psi$DO calculus. If we denote by $V^{\pm}_{m}(s, s')$ the resulting semi-group on $H^{m}(\Sigma; \tV)$, then  $V_{m}^{\pm}$ is an extension resp.~restriction of $V^{\pm}$ if $m<0$ resp.~$m>0$.  Let us prove (2) in the $+$ case. We fix $s'\in I$, $\chi\in \coinf(I)$ with $\supp \chi\subset \{s\leq s'-\delta/2\}$, $\chi=1$ in $\{s\leq s'- \delta\}$. Then $(\p_{s}- \tilde{b}^{+}(s))\chi(s)V^{+}(s, s')u= \chi'(s)V^{+}(s, s')u$. The operator $\p_{s}- \tilde{b}^{+}(s)$ has principal symbol $(\i \sigma- \sigma_{\rm pr}(\ta(s))({\rm k})^{\12})\one_{\tV}$ hence is elliptic in $\Psi^{1}(I\times\Sigma; \tV)$. If $u\in H^{m_{0}}(\Sigma; \tV)$, then $V^{+}(\cdot,  s')u\in H^{n_0}(I\times \Sigma; \tV)$ for some $n_{0}$  and  by elliptic regularity  $\chi(\cdot )V^{+}(\cdot, s')u\in H^{n_{0}+1}(I\times \Sigma; \tV)$. By iterating this argument we obtain that $\chi(\cdot )V^{+}(\cdot, s')u\in H^{n}(I\times \Sigma; \tV)$ for any $n$ so $V^{+}(s, s')u\in H^{m}(\Sigma; \tV)$ for any $m\in \nn$.
\qed

For $v\in \cinfb(I; \coinf(\Sigma; \tV))$  we set 
\[
T^{\pm}v(s)\defeq \pm \int_{\rr}H(\mp(s-s'))V^{\pm}(s, s')v(s')ds',
\]
where $H(t)= \one_{\rr^{+}}(t)$ is the Heaviside function, 
so that
\begin{equation}
\label{e01.12}
(-\p_{s}+ \tb^{\pm})\circ T^{\pm}=T^{\pm}\circ (-\p_{s}+ \tb^{\pm})= \one.
\end{equation}
\begin{proposition}\label{prop01.2c}
 Let 
 \[
 \tD^{(-1)}= \left((\tb^{+}- \tb^{-})^{-1}(T^{+}- T^{-})\right).
 \]
 Then 
 \[
 \tD \circ \tD^{(-1)}= \one + R_{-\infty},
 \]
 for some $R_{-\infty}\in \cW^{-\infty}(\tM; \tV)$.
 \end{proposition}
 \proof
We obtain using \eqref{e01.12} and Prop.~\ref{prop01.2bbb} {\it ii)}:
\[
\bea
&\p_{s}\left((\tb^{+}- \tb^{-})^{-1}(T^{+}- T^{-})\right)\\
&=(\tb^{+}- \tb^{-})^{-1}\left(\tb^{+}T^{+}- \tb^{-}T^{-}- (\tb^{+2}- \tb^{-2})(\tb^{+}- \tb^{-})^{-1}(T^{+}- T^{-})\right)+ r_{1, -\infty},
\eea
\]
for
\[
r_{1,-\infty}= - (\tb^{+}- \tb^{-})^{-1}(\tr^{+}_{-\infty}- \tr^{-}_{-\infty})(\tb^{+}- \tb^{-})^{-1}(T^{+}- T^{-}).
\]
Next,
\[
\begin{array}{l}
\tb^{+}- (\tb^{+2}- \tb^{-2})(\tb^{+}-\tb^{-})^{-1}= -(\tb^{+}- \tb^{-})\tb^{-}(\tb^{+}- \tb^{-})^{-1}\\[2mm]
-\tb^{-}+ (\tb^{+2}- \tb^{-2})(\tb^{+}- \tb^{-1})^{-1}= (\tb^{+}- \tb^{-}\tb^{+})(\tb^{+}- \tb^{-})^{-1},
\end{array}
\]
hence
\beq\label{e01.12b}
\bea
&\p_{s}\left((\tb^{+}- \tb^{-})^{-1}(T^{+}- T^{-})\right)\\[2mm]
&= - \tb^{-}(\tb^{+}- \tb^{-})^{-1}T^{+}+ \tb^{+}(\tb^{+}- \tb^{-})^{-1}T^{-}+ r_{1,-\infty}\\[2mm]
&=T^{\pm}- \tb^{\pm}\left((\tb^{+}- \tb^{-})^{-1}(T^{+}- T^{-})\right)+ r_{1,-\infty}.
\eea
\eeq
Using again \eqref{e01.12} we obtain
\[
(- \p_{s}+ \tb^{+})(\p_{s}+ \tb^{+})\left((\tb^{+}- \tb^{-})^{-1}(T^{+}- T^{-})\right)= \one- (\p_{s}- \tb^{+})r_{1,-\infty}.
\]
Hence, setting
\begin{equation}
\label{e01.13}
\tD^{(-1)}v(s)= (\tb^{+}(s)- \tb^{-}(s))^{-1}(T^{+}- T^{-})v(s),
\end{equation}
 by \eqref{e01.11} we get
\begin{equation}
\label{e01.14}
\tD \tD^{(-1)}= \one+ R_{-\infty}, \ \ R_{-\infty}\in \cinfb(I; \Psi^{-\infty}(\Sigma; \tV)),
\end{equation}
for
\[
R_{-\infty}= (-\p_{s}+ \tb^{+}r_{1, -\infty})+ \tr^{+}_{-\infty}\tD^{(-1)}.
\]
Using \eqref{e01.12} and \Lemma \ref{lemma01.2}  (1) we obtain that if $r_{-\infty}\in \cinfb(I; \cW^{-\infty}(\Sigma; \tV))$, then $r_{-\infty}T^{\pm}\in \cW^{-\infty}(\tM; \tV)$.  This implies that $R_{-\infty}\in \cW^{-\infty}(\tM; \tV)$ as claimed. \qed

\subsubsection{Dirichlet realization of $\tD$}\label{sec01.3.5b}
Let us fix $T$ such that $\open{-T, T}\subset I$ and set $\Omega= \open{-T, T}\times \Sigma$. We denote by $H^{1}_{0}(\Omega; \tV)$ the closure of $\coinf(\Omega, \tV)$ for the norm 
\beq\label{e01.13b}
\|u\|^{2}_{H^{1}(\Omega; \tV)}= \int_{\Omega}\left((\p_{s}u| \p_{s}u)_{\tV}+ (u| -\Delta_{\trh_{0}}u)_{\tV}+ (u| u)_{\tV}\right)|\trh_{0}|^{\12}dtd\rx.
\eeq
We denote by $L^{2}(\Omega, \tV)$ the $L^{2}$ space defined using  the scalar product $(\cdot| \cdot)_{\tV(\Omega)}$.

We consider the sesquilinear form
\[
Q_{\Omega}(v, u)\defeq  (v| \tD u)_{\tV(\Omega)}, \hbox{ with domain }\Dom Q_{\Omega}= \coinf(\Omega; \tV).
\]
\begin{proposition}\label{prop01.2bb}
There exist $T_{0}>0$  such that  for $0<T\leq T_{0}$ one has:
\ben
\item $Q_{\Omega}$ and $Q^{*}_{\Omega}$ are closeable on $L^{2}(\Omega; \tV)$;
\item their closures $\bar{Q_{\Omega}}$, $\bar{Q^{*}_{\Omega}}$ are sectorial with domain $H^{1}_{0}(\Omega; \tV)$;
\item  the closed operators $\tD_{\Omega}$, $ \tD^{*}_{\Omega}$ associated to $\bar{Q_{\Omega}}$, $\bar{Q^{*}_{\Omega}}$ satisfy $0\in {\rm rs}(\tD_{\Omega})$, $0\in {\rm rs}(\tD^{*}_{\Omega})$;
\item $\tD^{*}_{\Omega}$ is the adjoint of $\tD_{\Omega}$.
\een
\end{proposition}
\proof
Let $Q_{\rm ref}$ be the sesquilinear form associated to \eqref{e01.13b} with domain $\coinf(\Omega; \tV)$. By \eqref{e01.9} and the Poincaré inequality, we can find $T_{0}>0$ such that for $0<T\leq T_{0}$ one has
\beq\label{eq:coercive}
C^{-1}Q_{\rm ref }\leq \Re Q_{\Omega}\leq C Q_{\rm ref}, \quad  \module{\Im Q_{\Omega}}\leq C Q_{\rm ref}, \quad C>0.
\eeq
This implies (2). Then,  (3) follows  from \cite[Sect.~VI.2.1]{K1} and (4) from \cite[Thm.~VI.2.5]{K1}. \qed

 The operator $\tD_{\Omega}$ is the {\em Dirichlet realization} of $\tD$.  We denote by $\tD_{\Omega}^{-1}: L^{2}(\Omega; \tV)\to \Dom \tD_{\Omega}$ its  inverse.
\subsubsection{Parametrix for the Dirichlet problem}\label{sec01.3.6}
Let  us fix $T>0$ such that $[-T, T]\subset \mathring{I}$. We want to find a parametrix for the Dirichlet problem
\begin{equation}
\label{e01.15}
\begin{cases}
\tD u= f\hbox{ in }\Omega,\\
u_{|\p\Omega}=0.
\end{cases}
\end{equation}
Let
\[
W^{\pm}(s, s')= \Texp(-\textstyle\int_{s'}^{s}\tb^{\pm}(\sigma)d\sigma), \hbox{ for }\pm(s-s')\geq 0,
\]
and
\[
R_{1, -\infty}=  \mat{0}{W^{-}(-T, T)}{W^{+}(T, -T)}{0}.
\]
Since $\pm\Re\tb^{\pm}\geq 1$ by Prop.~\ref{prop01.2bbb}, we have  $\| R_{1, -\infty}\|_{B(L^{2}(\Sigma; \tV\otimes \cc^{2}))}\leq \e^{-2T}$ and by \Lemma \ref{lemma01.2} (2), $R_{1, -\infty}\in \cW^{-\infty}(\Sigma; \tV\otimes \cc^{2})$. Therefore  $\one + R_{1, -\infty}$ is invertible in $B(L^{2}(\Sigma; \tV\otimes \cc^{2}))$  and by spectral invariance, $(\one+ R_{1, -\infty})^{-1}\in \one + \cW^{-\infty}(\Sigma, \tV\otimes \cc^{2})$.

 Let us set for $g\in \cinfb(I; \cinfb(\Sigma; \tV))$:
 \[
 \varrho_{\p \Omega} g\defeq \col{g(-T)}{g(T)},
 \]
and for  $v^{\pm}\in L^{2}(\Sigma; \tV)$:
 \beq\label{e01.01}
S\col{v^{+}}{v^{-}}(s)\defeq W^{+}(s, -T)v^{+}+ W^{-}(s, T)v^{-},
 \eeq
 so that 
 \[
 \varrho_{\p \Omega}\circ S= \one + R_{1, -\infty}.
 \]
The following proposition gives a construction of a parametrix for the Dirichlet problem \eqref{e01.15}.
\begin{proposition}\label{prop01.3}
For $f\in \cinfb(I; \cinfb(\Sigma; \tV))$ let 
\[
\tD_{\Omega}^{(-1)}= \tD^{(-1)}- S\circ (1+ R_{1, -\infty})^{-1}\circ \varrho\circ \tD^{(-1)}.
\]
Then
\[
\begin{cases}
 \tD\circ \tD_{\Omega}^{(-1)}= \one + R_{-\infty}, \ \  R_{-\infty}\in \cW^{-\infty}(\tM; \tV),\\
 \varrho_{\p \Omega}\circ \tD_{\Omega}^{(-1)}=0.
\end{cases}
\]
\end{proposition}
\proof 
By \eqref{e01.11} and the analog of \Lemma \ref{lemma01.2} for $W^{\pm}$ we obtain that:
\[
\tD S\col{v^{+}}{v^{-}}= \tr^{+}_{-\infty}W^{+}(\cdot, -T)v^{+}+ \tr^{-}_{-\infty}W^{-}(\cdot, T)v^{-}= R_{2, -\infty}\col{v^{+}}{v^{-}}
\]
where $R_{2, -\infty}\in \cinfb(I; \cW^{-\infty}(\Sigma; \tV))$. On the other hand by \eqref{e01.13} we have
\[
\tD^{(-1)}v(\pm T)= (\tb^{+}- \tb^{-})^{-1}(\pm T)\int_{-T}^{T}V^{\mp}(\pm T, s')v(s')ds'.
\]
Therefore if we set 
\[
u= \tD^{(-1)}v- S(1+R_{1, -\infty})^{-1}\varrho\tD^{(-1)}v,
\]
we get
\[
\begin{cases}
\tD u= v+ R_{3, -\infty}v,\\
\varrho_{\p \Omega}u=0,
\end{cases}
\]
for some $R_{3, -\infty}\in  \cW^{-\infty}(\tM; \tV)$. \qed

\begin{proposition}\label{prop01.3b} Let  $\tD_{\Omega}^{(-1)}$ be as Prop.~\ref{prop01.3}. Then,
 \[
 \tD_{\Omega}^{-1}- \tD_{\Omega}^{(-1)}\in \cW^{-\infty}(\Omega; \tV).
 \]
 
\end{proposition}
\proof 
From the sesquilinear form associated to $\tD_{\Omega}$ we know that $\tD_{\Omega}^{-1}: H^{-1}(\Omega;\tV)\to H^{1}(\Omega;\tV)$. By the usual argument  of commuting tangential derivatives with $\tD$ and using the equation $\tD u =v$ to control $\p_{s}$ derivatives, we obtain that $\tD_{\Omega}^{-1}= H^{s}(\Omega; \tV)\to H^{s+2}(\Omega; \tV)$ for all $s\in\rr$.  By Prop.~\ref{prop01.3} we obtain  $\tD^{(-1)}_{\Omega}= \tD_{\Omega}^{-1} + R_{-\infty}\tD_{\Omega}^{-1}$, which proves the proposition. \qed

\newmodif{
Let us fix an extension map $e: \cinf(\p\Omega; \tV)\to \bar{C^{\infty}}(\Omega; \tV)$ such that $(ef)\traa{\Omega}=f$. We can assume moreover that $\trho^{\pm}ef= 0$ by choosing $e$ such that $ef= 0$ near $s= 0$. 
\begin{lemma}\label{minito}
We have
\[
\trho^{\pm}\tD_{\Omega}^{-1}\tD u= \trho^{\pm}u+ \trho^{\pm} \tD_{\Omega}^{-1}\tD e(u\traa{\Omega}), \ u\in \bar{H^{1}}(\Omega; \tV).
\]
\end{lemma}
\proof Let $v= \tD_{\Omega}^{-1}\tD u$. We set $v= u - eu\traa{\p \Omega}+ w$ where
\[
\begin{cases}
\tD w= \tD eu\traa{\Omega}\hbox{ in }\Omega,\\
w\traa{\p \Omega}= 0,
\end{cases}
\]
i.e.~$w= \tD_{\Omega}^{-1}\tD e u\traa{\Omega}$. Applying $\trho^{\pm}$ to this identity we obtain the lemma. \qed
}

\subsubsection{Green's formula}\label{sec01.3.3b}
Let $\Omega^{\pm}= \Omega\cap \{\pm s>0\}$. For $u\in\cinf(\Omega; \tV)$ we set 
\beq\label{e01.14b}
 \trho u\defeq \col{u\traa{\Sigma}}{-\p_{s}u\traa{\Sigma}}=\col{u(0)}{-\p_{s}u(0)}.
\eeq
We denote by $\trho^{\pm}$ the analogous trace operators defined on $\bar{\cinf}(\Omega^{\pm}; \tV)$:
\[
\trho^{\pm}u\defeq  \col{u(0^{\pm})}{-\p_{s}u(0^{\pm})},
\]
\modif{and by $\trho_{\pm T}$ the trace operator at $\pm T$, i.e.
$$
\trho_{\pm T} u\defeq \col{u\traa{\p\Omega^{\pm}\setminus \Sigma}}{-\p_{s}u\traa{\p\Omega^{\pm}\setminus \Sigma}}=\col{u(\pm T)}{-\p_{s}u(\pm T)}.
$$}
\modif{The definitions extend of course to spaces with sufficient Sobolev regularity}. 

\modif{Let us denote
\[
\tsig= \mat{0}{-\one}{\one}{0}, \quad \tilde{q}= \mat{0}{\one}{\one}{0}.
\]}

\begin{proposition}\label{prop01.3c}
Let $u, v\in \modif{\bar{H^{1}}(\Omega^{\pm}; \tV)}$ with $\tD^{*} u, \tD v\in L^{2}(\Omega^{\pm}; \tV)$. Then
\begin{equation}
\label{e01.green1}
\bea
(u| \tD v)_{\tV(\Omega^{\pm})}- (\tD^{*} u| v)_{\tV(\Omega^{\pm})}&= \pm (\trho^{\pm} u| \tsig\trho^{\pm}v)_{\tV(\Sigma)\otimes \cc^{2}} \fantom\,  \modif{\mp (\trho_{\pm T} u| \tsig\trho_{\pm T}v)_{\tV(\Sigma)\otimes \cc^{2}}}.
\eea
\end{equation}
Let $u, v\in \modif{ \bar{H^{1}}(\Omega^{\pm}; \tV)}$ with $\tD u, \tD v\in L^{2}(\Omega^{\pm}; \tV)$. Then
\begin{equation}
\label{e01.green2}
\bea
(u| \tD v)_{\tV(\Omega^{\pm})}+ (\tD u | v)_{\tV(\Omega^{\pm})}&= 2\,\eta_{\Omega^{\pm}}(u, v)\mp (\trho^{\pm} u| \tilde{q}\trho^{\pm}v)_{\tV(\Sigma)\otimes \cc^{2}}
\fantom\,  \modif{\pm (\trho_{\pm T} u| \tilde{q}\trho_{\pm T}v)_{\tV(\Sigma)\otimes \cc^{2}}},
\eea
\end{equation}
where
\[
\eta_{\Omega^{\pm}}(u, v)=  (\p_{s}u| \p_{s}v)_{\tV(\Omega^{\pm})}+ (u | (\Re \ta) v)_{\tV(\Omega^{\pm})}.
\]
\end{proposition}
\proof  By elliptic regularity we know that $u, v\in H^{2}(\Omega^{\pm}; \cc^{2})$, hence $\trho^{\pm}u, \trho^{\pm}v$ belong to $H^{1}(\Sigma; \tV)\oplus L^{2}(\Sigma; \tV)$ and in consequence the r.h.s.~in \eqref{e01.green1}, \eqref{e01.green2} are well-defined. The identities follow then by integration by parts in $s$. \qed

In agreement with the notation introduced in \ref{sec01.1.1} we denote by $\bar{H^{1}_{0}}(\Omega^{\pm}; \tV)$ the space of restrictions to $\Omega^{\pm}$ of elements of $H^{1}_{0}(\Omega; \tV)$, i.e. 
\beq
\label{defdebarsobolo}
u\in \bar{H^{1}_{0}}(\Omega^{\pm}; \tV) \hbox{ iff }u\in H^{1}(\Omega^{\pm}; \tV)\hbox{ and }u\traa{\p\Omega^{\pm}\setminus \Sigma}=0.
\eeq
\modif{Note that in the special case when $u, v\in \modif{\bar{H^{1}_0}(\Omega^{\pm}; \tV)}$, we have  
\beq\label{greendir}
(\trho_{\pm T} u| \tsig\trho_{\pm T}v)_{\tV(\Sigma)\otimes \cc^{2}}=0= (\trho_{\pm T} u| \tilde q\trho_{\pm T}v)_{\tV(\Sigma)\otimes \cc^{2}}
\eeq
and the Green's formulas in Prop.~\ref{prop01.3c} become simpler. 
 }

\subsubsection{\calde projectors}\label{sec01.3.7}
 We denote by $\trho^{*}: \cE'(\Sigma; \tV)\to \cD'(\Omega; \tV)$ the formal adjoint of $\trho$ in \eqref{e01.14b}, where $\cinf(\Omega; \tV)$, resp.~$\coinf(\Sigma; \tV)\otimes\cc^{2}$, are equipped with the scalar products $(\cdot| \cdot)_{\tV(\tM)}$, resp.~$(\cdot| \cdot)_{\tV(\Sigma)\otimes \cc^{2}}$.  Explicitly we have
\begin{equation}
\label{e01.16}
\trho^{*}f= \delta(s)\otimes f_{0}+ \delta'(s)\otimes f_{1}, \  \ f= \col{f_{0}}{f_{1}}\in \coinf(\Sigma; \tV)\otimes \cc^{2}.
\end{equation}
\begin{definition}\label{def01.1}
 The {\em \calde projectors}  \modif{for the Dirichlet realization $\tD_\Omega$ of $\tD$} are the maps 
 \beq\label{e01.16b}
 \tilde{c}^{\pm}\defeq \mp \trho^{\pm} \tD_{\Omega}^{-1} \trho^{*}\tilde{\sigma}.
 \eeq
 \end{definition}
We will see in Prop.~\ref{prop01.4b} that $\tilde{c}^{\pm}$ are indeed projections if we consider them as operators acting on the spaces 
\beq\label{e01.16cc}
\cH^{s}(\Sigma; \tV\otimes \cc^{2})\defeq H^{s}(\Sigma; \tV)\oplus H^{s-1}(\Sigma; \tV), \ \  s\in \rr.
\eeq
\begin{proposition}\label{prop01.4} The \calde projectors satisfy
 \beq\label{e01.16c}
 \tilde{c}^{\pm}= \mat{\mp(\tb^{+}- \tb^{-})^{-1}\tb^{\mp}}{\pm(\tb^{+}- \tb^{-})^{-1}}{\mp \tb^{+}(\tb^{+}- \tb^{-})^{-1}\tb^{-}}{\pm \tb^{\pm}(\tb^{+}- \tb^{-})^{-1}}(0)+ R^{\pm}_{-\infty}.
\eeq
for some  $R^{\pm}_{-\infty}\in\cW^{-\infty}(\Sigma; \tV \otimes \cc^2)$.
 \end{proposition}
 \proof 
  We first claim that we can replace $\tD_{\Omega}^{-1}$ by $\tD^{(-1)}$ in \eqref{e01.16b} modulo an error term in $\cW^{-\infty}(\Sigma; \tV\otimes \cc^{2})$. By 
   Prop.~\ref{prop01.3b} we can replace $\tD_{\Omega}^{-1}$ by the parametrix $\tD_{\Omega}^{(-1)}$ in \eqref{e01.16b}, modulo an error term in $\cW^{-\infty}(\Sigma; \tV\otimes \cc^{2})$.  
   Next, by  \Lemma \ref{lemma01.2} (2) with $V^{\pm}$ replaced by $W^{\pm}$, we  obtain that $\varrho^{\pm}S$, where $S$ is defined in \eqref{e01.01},  belongs to $\cW^{-\infty}(\Sigma; \tV\otimes \cc^{2})$. Using the expression of $\tD_{\Omega}^{(-1)}- \tD^{(-1)}$ given in Prop.~\ref{prop01.3} this proves our claim.

Furthermore,    if $v= \delta(s)\otimes f_{1}$ then:
 \[
 \bea
  T^{\pm}v(s)&= \pm\int_{\rr}H(\mp(s-s'))V^{\pm}(s, s')\delta(s')\otimes f_{1}ds'\\
  &= \pm H(\mp s)V^{\pm}(s, 0)f_{1}.
 \eea
 \]
Therefore 
 \[
 T^{+}v(0^{+})=0, \quad T^{-}v(0^{+})= - f_{1}
 \]
and
 \beq\label{e01.17}
 (\tD^{(-1)}v)(0^{+})= (\tb^{+}- \tb^{-})^{-1}(0)f_{1}.
 \eeq
 By \eqref{e01.12b} we have
 \begin{equation}
 \label{e01.18}
 \p_{s}(\tD^{(-1)}v)(0^{+})= T^{+}v(0^{+})- \tb^{+}(\tD^{(-1)}v)(0^{+})+ r_{1, -\infty}f_{1},  \ \ r_{1, -\infty}\in\cW^{-\infty}(\Sigma; \tV).
 \end{equation}
 If $v= - \delta'(s)\otimes f_{0}$, we obtain similarly
 \[
\bea
  T^{\pm}v(s)&= \mp\int_{\rr}H(\mp(s-s'))V^{\pm}(s, s') \delta'(s')\otimes f_{0}ds'\\[2mm]
  &=\pm\int_{\rr}\left(\delta(\mp(s-s'))V^{\pm}(s, s')\pm H(\mp(s-s'))V^{\pm}(s, s')\tb^{\pm}(s')\right) \delta(s')\otimes f_{0}ds'\\[2mm]
  &= \delta(s)f_{0}\mp H(\mp s)V^{\pm}(s, 0)\tb^{\pm}(0)f_{0},
\eea
 \]
using that $\p_{s'}V^{\pm}(s, s')= - V^{\pm}(s, s')\tb^{\pm}(s')$.  Therefore,
\[
T^{+}v(0^{+})=0, \quad T^{-}v(0^{+})= \tb^{-}(0)f_{0}
\]
and
\begin{equation}
\label{e01.19}
(\tD^{(-1)}v)(0^{+})= - (\tb^{+}- \tb^{-})^{-1}(0^{+})\tb^{-}(0)f_{0}.
\end{equation}
Using again \eqref{e01.18} we obtain:
\[
(\p_{s}\tD^{(-1)}v)(0^{+})= \tb^{+}(0)(\tb^{+}- \tb^{-})^{-1}(0)\tb^{-}(0)f_{0}+r_{1, -\infty}f_{0}.
\]
Therefore we obtain
\[
\tilde{c}^{+}= \mat{-(\tb^{+}- \tb^{-})^{-1}\tb^{-}}{(\tb^{+}- \tb^{-})^{-1}}{- \tb^{+}(\tb^{+}- \tb^{-})^{-1}\tb^{-}}{\tb^{+}(\tb^{+}- \tb^{-})^{-1}}(0)+ R_{-\infty}.
\]
The proof for $\tilde{c}^{-}$ is analogous. \qed

\begin{proposition}\label{prop01.4b} The \calde projectors 
 $\tilde{c}^{\pm}$ are bounded on $\cH^{s}(\Sigma; \tV\otimes \cc^{2})$ for $s\in \rr$ and satisfy
 \[
 \tilde{c}^{+}+ \tilde{c}^{-}=\one, \quad \tilde{c}^{\pm}= (\tilde{c}^{\pm})^{2}\hbox{ on }\cH^{s}(\Sigma; \tV\otimes \cc^{2}).
 \]
 \end{proposition}
 \proof The boundedness property follows immediately from Prop.~\ref{prop01.4}.  We prove that $\tilde{c}^{+}+ \tilde{c}^{-}=\one$ as in \cite[Thm.~4.5]{GW2}. To prove the last  statement we apply \cite[Prop.~4.8]{GW2} which generalizes easily to the vector bundle setting. 
 Let us explain the argument in  \cite[Prop.~4.8]{GW2}.
 If $\Sigma$ is compact the proof is elementary.

 If $\Sigma$ is not compact, we first choose a convenient sequence $(\psi_{n})_{n\in \nn}$ of cutoff functions on $\Sigma$.  We
 fix some reference point $x_{0}\in \Sigma$ and denote by $d(x_{0}, x)$ the geodesic distance for $\trh_{0}$. Since $(\Sigma, \trh_{0})$ is of bounded geometry, there exists a function $r\in \cinf(\Sigma)$  such that 
\[
C^{-1} d(x_{0}, x)\leq r(x)\leq C d(x_{0}, x), \ \nabla r\in \cinfb(\Sigma; T\Sigma).
\]
Next, we set  $\psi_{n}(x)= F(n^{-1}r(x))$ for $F\in \coinf(\rr)$ equal to $1$ near $0$.  Then:
\beq\label{prop-de-cutoff}
\begin{array}{rl}
i)& \psi_{n}= 1 \hbox{ on }K\Subset \Sigma \hbox{ for }n\hbox{ large enough},\\[2mm]
ii)& \nabla \psi_{n}\in O(n^{-1}) \hbox{ in }\cinfb(\Sigma; T\Sigma).
\end{array}
\eeq
Clearly,  $\slim_{n\to \infty}\psi_{n}= \one$
on $\coinf(\Sigma; \tV)\otimes \cc^{2}$. Arguing as in  \cite[Prop.~4.8]{GW2} we obtain that   $\psi_{n}\tilde{c}^{\pm}f- \tilde{c}^{\pm}\psi_{n}\tilde{c}^{\pm}f$ tends to $0$ in $\cD'(\Sigma; \tV)\otimes \cc^{2}$ as $n\to \infty$ for all $f\in \coinf(\Sigma; \tV)\otimes \cc^{2}$.

By \eqref{prop-de-cutoff} the sequence $\psi_{n}$ tends also strongly to $\one$ on $\cH^{s}(\Sigma; \tV\otimes \cc^{2})$. Therefore $\tilde{c}^{\pm}= (\tilde{c}^{\pm})^{2}$ on $\coinf(\Sigma; \tV)\otimes \cc^{2}$ and thus on $\cH^{s}(\Sigma; \tV\otimes \cc^{2})$ by density. \qed

\subsubsection{Dirichlet-to-Neumann maps}\label{sec01.3.8}
Let us now consider the following inhomogeneous boundary value problem:
\begin{equation}
\label{e01.20}
\begin{cases}
\tD u= 0\hbox{ in }\Omega^{\pm},\\
u\traa{\p \Omega^{+}\setminus \Sigma}=0,\\
u\traa{\Sigma}=v.
\end{cases}
\end{equation}
We claim that \eqref{e01.20} has a unique solution $u\eqdef P_{\Omega^{\pm}}v$. In fact, let $\chi\in \bar{\coinf}(\open{0, T})$ with $\chi= 1$ on $\open{0, T/2}$ and $\chi^{\pm}(s)= \chi(\pm s)$. Then for
\[
u_{1}= \chi^{\pm}(s)W^{\pm}(s, 0)v,
\]
we have $u_{1}(\pm T)=0$, $u_{1}(0)= v$ and by \eqref{e01.11}
\[
\tD u_{1}= (- \p_{s}+ \tb^{\pm})(\p_{s}\chi)^{\pm}W^{\pm}(\cdot, 0)v+ \tr^{\pm}_{-\infty}\chi^{\pm}W^{\pm}(\cdot, 0)v\defeq m^{\pm}_{-\infty}v,
\]
where $m^{\pm}_{-\infty}\in \cW^{-\infty}(\Sigma, \Omega^{\pm}; \tV)$. 

By the same arguments as in \ref{sec01.3.5b}, we can consider  the Dirichlet realization $\tD_{\Omega^{\pm}}$  of $\tD$ in $\Omega^{\pm}$, which for $0<T\ll 1$ is invertible. Therefore we can solve \eqref{e01.20} by
\beq\label{e01.21}
P_{\Omega^{\pm}}v= \chi^{\pm}W^{\pm}(\cdot, 0)v- \tD_{\Omega^{\pm}}^{-1}m^{\pm}_{-\infty}v.
\eeq
\begin{definition}\label{def01.2}
 The {\em Dirichlet-to-Neumann maps} are
 \[
 N_{\Omega^{\pm}}v\defeq- \p_{s}P_{\Omega^{\pm}}v \traa{\Sigma}.
 \]
 \end{definition}
\begin{proposition}\label{prop01.5} The Dirichlet-to-Neumann maps satisfy
 \[
 N_{\Omega^{\pm}}= \tb^{\pm}(0)+ r^{\pm}_{-\infty}, \ \ r^{\pm}_{-\infty}\in \cW^{-\infty}(\Sigma; \tV).
 \]
 
\end{proposition}
\proof By  \eqref{e01.21} we have 
\[
\bea
N_{\Omega^{\pm}}v&= - \p_{s}\chi^{\pm}W^{\pm}(\cdot, 0)v\traa{\Sigma}+ \p_{s}\tD_{\Omega^{\pm}}^{-1}m_{-\infty}^{\pm}v\traa{\Sigma}\\
&\eqdef\tb^{\pm}(0)v+ r^{\pm}_{-\infty}v.
\eea
\]
To prove that $r^{\pm}_{-\infty}\in  \cW^{-\infty}(\Sigma; \tV)$, we  use Prop.~\ref{prop01.3b}  and the explicit form of the parametrix $\tD_{\Omega}^{(-1)}$  with $\Omega$ replaced by $\Omega^{\pm}$. \qed

 \subsubsection{Positivity of the Dirichlet to Neumann maps}\label{sec01.3.9}
 \begin{proposition}\label{prop01.5b}
 The Dirichlet-to-Neumann maps satisfy
 \[
 \pm \Re N_{\Omega^{\pm}}\sim (-\Delta_{\trh_{0}}+1)^{\12} \hbox{ on }\coinf(\Sigma; \tV),
 \]
 for the scalar product $(\cdot| \cdot)_{\tV(\Sigma)}$.
\end{proposition}
\proof 
 By Prop.~\ref{prop01.5} we have $\pm \Re N_{\Omega^{\pm}}\leq c_{0}(-\Delta_{\trh_{0}}+1)^{\12}$, so it suffices to prove the other inequality.
Let $u= P_{\Omega^{\pm}}v$ for $v\in \coinf(\Sigma; \tV)$.  We have $u\in \bar{H^{1}_{0}}(\Omega^{\pm}; \tV)$ and $\tD u= 0$ in $\Omega^{\pm}$ so we can apply \eqref{e01.green2}. We obtain:
\[
\pm(v|\Re N_{\Omega^{\pm}}v)_{\tV(\Sigma)}= \eta_{\Omega^{\pm}}(u, u)\sim \|u\|^{2}_{H^{1}(\Omega^{\pm})}.
\]
By the continuity properties of $u\mapsto u\traa{\Sigma}$ between Sobolev spaces  the rhs is equivalent to $\|v\|^{2}_{H^{\12}(\Sigma)}$.  \qed

\subsection{Relation between Lorentzian and Euclidean projectors}\label{sec01.4} 

We are ready to prove that if $\ta(s)$ formally coincides with the Wick rotation of $a(t)$ in the sense of Taylor coefficients at $0$, then the  Hadamard projectors $c^\pm$  for $D=\p_t^2+a(t)$ differs by a smoothing term from the \calde projectors $\tilde c^\pm$ for $\tD=-\p_s^2+\ta(s)$.

\begin{proposition}\label{prop01.6}
 Suppose that
\beq\label{e01.23b}
  (\i \p_{t})^{n}a(0)= \p_{s}^{n}\ta(0) \  \forall n\in \nn.
 \eeq
 Then
 \beq\label{e01.22b}
 c^{\pm}- \tilde{c}^{\pm}\in \cW^{-\infty}(\Sigma; \tV\otimes \cc^{2}).
 \eeq
 \end{proposition}
 \proof
By comparing formulas \eqref{e01.21b}  and \eqref{e01.16c}, we see that it suffices to prove that
 \beq\label{e01.30}
 b^{\pm}(0)- \tb^{\pm}(0)\in \cW^{-\infty}(\Sigma; \tV),
 \eeq
 where $b^{\pm}$ are defined in \eqref{e01.21c} and $\tb^{\pm}$ in Prop.~\ref{prop01.2bbb}. 
 The  proof of \eqref{e01.30} is divided in two steps.

 \step{1}  In the first step we prove that
 \begin{equation}
 \label{e01.28}
 \tb^{-}(0)= - \tau^{-1}\tb^{+*}(0)\tau \, \hbox{ mod } \, \cinfb(I; \cW^{-\infty}(\Sigma; \tV)).
 \end{equation}
 
 Let us prove \eqref{e01.28}.
  Let  us abbreviate  $\tb^{+}$ by $\tb$.  We know that $\tb=\teps+ \tb_{0}$,  where $\tb_{0}$ solves the fixed point equation
\begin{equation}
\label{e01.25}
\tb_{0}= \tilde{c}_{0}+ \tilde{F}(\tb_{0}), \ \  \tilde{c}_{0}= (2\teps)^{-1}(\p_{s}\teps)
\end{equation}
 for
 \beq\label{e01.25b}
\tilde{F}(\tilde{d})= (2\teps)^{-1}(\p_{s}\tilde{d}+ [\teps, \tilde{d}]-\tilde{d}^{2}).
\eeq
Let us recall that \eqref{e01.25} is solved symbolically by
\begin{equation}
\label{e01.26}
\tb_{0}= \tilde{c}_{0}+ \sum_{k\geq 1}\tilde{d}_{k}- \tilde{d}_{k-1}, \quad \tilde{d}_{0}= \tilde{c}_{0}, \ \tilde{d}_{k}= \tilde{c}_{0}+ \tilde{F}(\tilde{d}_{k-1}).
\end{equation}

Let us also set 
  \[
  \hat{b}(s)= - \tau^{-1}\tb^{*}(-s)\tau.
  \]
  Since 
  \[
  \p_{s}\tb(s)- \tb^{2}(s)+ a(s)=0\hbox{ mod }\cinfb(I; \cW^{-\infty}(\Sigma; \tV)),
  \]
   we obtain that 
  \[
  \p_{s}\hat{b}(s)- \hat{b}^{2}(s)+ \hat{a}(s)=0\hbox{ mod }\cinfb(I; \cW^{-\infty}(\Sigma; \tV))
  \]
   for $\hat{a}(s)= \tau^{-1}\ta^{*}(-s)\tau$.  Note that $\hat{a}(s)$ has the same properties as $\ta(s)$. Moreover,
 \[
 \hat{b}(s)= -\hat{\epsilon}(s)\hbox{ mod }\cinfb(I; \Psi^{0}(\Sigma; \tV)),
 \]
 where $\hat{\epsilon}(s)= \hat{a}^{\12}(s)$ is the $m$-accretive square root of $\hat{a}(s)$. By the uniqueness statement in Prop.~\ref{prop01.2bbb}, it follows that 
 \[
 \hat{b}(s)= \hat{b}^{-}(s)\hbox{ mod }\cinfb(I; \cW^{-\infty}(\Sigma; \tV)),
 \]
 where $\hat{b}^{-}(s)$ is the solution in Prop.~\ref{prop01.2bbb} with $a(s)$ replaced by $\hat{a}(s)$. Therefore we have
 \[
 \hat{b}= - \hat{\epsilon}+ \hat{b}_{0},
 \]
 where $\hat{b}_{0}$ solves the fixed point equation:
\begin{equation}
\label{e01.24}
\hat{b}_{0}= \hat{c}_{0}+ \hat{F}^{-}(\hat{b}_{0}),  \ \ \hat{c}_{0}= (2\hat{\epsilon})^{-1}\p_{s}\hat{\epsilon},
\end{equation}
where
\beq\label{e01.24b}
\hat{F}^{-}(\hat{d})= (2\hat{\epsilon})^{-1}(-\p_{s}\hat{d}+ [\hat{\epsilon}, \hat{d}]- \hat{d}^{2}).
\eeq
Next,  \eqref{e01.24} is solved symbolically by \eqref{e01.26}, with $\tilde{c}_{0}, \tilde{d}_{k}, \tilde{F}$ replaced by $\hat{c}_{0}$, $\hat{d}_{k}, \hat{F}$.

 A direct inspection of \eqref{e01.26} shows that modulo $\cinfb(I; \cW^{-\infty}(\Sigma; \tV))$, $\tb_{0}(0)$ depends only on the Taylor expansion of $s\mapsto \teps(s)$ at $s=0$, i.e.~on the Taylor expansion of $s\mapsto \ta(s)$ at $s= 0$. The same result holds also  for $\hat{b}_{0}(0)$, replacing $\ta(s)$ by $\hat{a}(s)$.
  
  We claim that
  \begin{equation}
  \label{e01.27} 
  (\p_{s})^{n}\hat{a}(0)= \p_{s}^{n}\ta(0), \ \forall n\in \nn, 
  \end{equation}
  which by the discussion above implies \eqref{e01.28}. 
 We first note that since $a(t)= a^{\star}(t)$ by ${\rm (H1)}$, we have $a^{*}(t)= \tau a(t)\tau^{-1}$, hence 
 \[
 ((\i \p_{t})^{n}a)^{*}(0)= (-1)^{n}\tau (\i \p_{t})^{n}a(0)\tau^{-1}, \ n\in \nn.
 \]
  By \eqref{e01.23b} this implies that
 \beq\label{e01.22c}
 \p_{s}^{n}\ta^{*}(0)= (-1)^{n}\tau\p_{s}^{n}\ta(0)\tau^{-1}, \ n\in \nn,
 \eeq
 and hence \eqref{e01.27}. This completes Step 1.
 
\step{2} In Step 2 we complete the proof of the proposition. Let $b$ be  the operator from Prop.~\ref{prop01.1}. We know that $b= \epsilon+ b_{0}$, where  $b_{0}$ solves the fixed point equation
 \beq\label{e01.23}
b_{0}= c_{0}+ F(c_{0}), \ \ c_{0}= (2\epsilon)^{-1}(\i \p_{t}\epsilon).
\eeq
for
\beq\label{e01.22}
F(d)= (2\epsilon)^{-1}(\i \p_{t}d+ [\epsilon, d]-d^{2}).
 \eeq
The equation \eqref{e01.23} is solved symbolically as before by 
\beq\label{e01.29}
b_{0}= c_{0}+ \sum_{k\geq 1}d_{k}- d_{k-1}, \quad d_{0}= c_{0}, \quad d_{k}= c_{0}+ F(d_{k-1}).
\eeq
 From the definitions of $F, \tilde{F}$ we get that if $d\in \cinfb(I; \Psi^{0}(\Sigma; V))$, $\tilde{d}\in \cinfb(I; \Psi^{0}(\Sigma; \tV))$ satisfy 
\[
(\i \p_{t})^{n}d(0)= (\p_{s})^{n}\tilde{d}(0), \  n\in \nn,
\]
then
\[
(\i \p_{t})^{n}F(d)(0)= (\p_{s})^{n}\tilde{F}(\tilde{d})(0), \  n\in \nn.
\]
Since $(\i \p_{t})^{n}a(0)= \p_{s}^{n}\ta(0) \ \forall n\in \nn$, we obtain by \eqref{e01.2f} that 
\[
(\i \p_{t})^{n}\epsilon(0)= \p_{s}^{n}\teps(0), \ n\in \nn,\hbox{ i.e. }(\i \p_{t})^{n}c_{0}(0)= \p_{s}^{n}\tilde{c}_{0}(0), \  n\in \nn.
\]
Therefore $(\i \p_{t})^{n}d_{k}(0)= (\p_{s})^{n}\tilde{d}_{k}(0) \ \forall n, k\in \nn$, which  
 implies that \[
 b(0)- \tb(0)= b^{+}(0)- \tb^{+}(0)\in \cW^{-\infty}(\Sigma; \tV).
 \]
 We have $b^{-}(0)= - b^{\star}(0)= - \tau^{-1}b^{*}(0)\tau$, see \eqref{e01.21c}, and $\tb^{-}(0)= - \tau^{-1}\tb^{*}(0)\tau$ mod $\cinfb(I; \cW^{-\infty}(\Sigma; \tV))$ by  \eqref{e01.28}. Therefore 
 \[
 b^{-}(0)- \tb^{-}(0)\in \cW^{-\infty}(\Sigma; \tV),
 \] which implies the proposition. \qed
\section{Wick rotation in Gaussian time}\label{sec3}\init

In this section we prove the main result of this paper, namely the existence of Hadamard states for linearized gravity on Einstein spacetimes $(M, \rg)$ satisfying the analyticity hypotheses in Subsect.~\ref{hippopotamus}.

 The idea of constructing states by Wick rotation in Gaussian time $t$ is inspired by the earlier work in the real analytic case \cite{GW2}, see \cite{schapira,wrochnawick} for further analyticity properties in related frameworks; we also remark that these results are largely consistent with the program recently outlined by Kontsevich--Segal \cite{kontsevich}.

\subsection{Framework}
 We assume that the metric $\rg$ satisfies the hypotheses in Subsect.~\ref{hippopotamus} and we apply the framework of Sect.~\ref{sec01} to the reduced operators $\hat{D}_{i}$, $i=1, 2$, constructed in Subsect.~\ref{ss:reduction}.

 \subsubsection{Hyperbolic operators}
From now on  the  operators $\hat{D}_{i}, \hat{d}$,$\hat{I}$ introduced in \ref{sssect.reduct} will simply be denoted by  $D_{i}= \p_{t}^{2}+ a_{i}(t)$, $d$ and $I_{2}$ respectively. 
  We can write
 $d$ as 
\[
d= d_{0}(t)\p_{t}+ d_{1}(t),\ d_{i}\in \cinfb(I; \Diff^{i}(\Sigma; V_{1}, V_{2})).
\]

   In order to have more uniform formulas with respect to  the index $i=1,2$ we set $I_{1}= \one$, which corresponds to the notation introduced in \eqref{def-physical-scalar-product}.
  
  The trace operators $\varrho_{i}$ are defined as in \eqref{e01.1bb}. 
  \subsubsection{Hermitian bundles}
 The  Hermitian  bundles are 
 \[
V_{i}= I\times (\sum_{k=0}^{i}\cc\otimes_{\rm s}^{k}T^{*}\Sigma),
 \]
 see  the identifications \eqref{etiti.0} and \eqref{etiti.-1}, equipped with the Hermitian structures
 \[
(\cdot | \cdot)_{V_{1}}= (\cdot| \cdot)_{\rg_{0}}, \quad (\cdot | \cdot)_{V_{2}}= 2(\cdot| \cdot)_{g^{\otimes 2}_{0}}.
 \]
 The  physical Hermitian forms  are:
\beq\label{e3.-1}
(u| u)_{I, V_{i}}= (u | I_{i} u)_{V_{i}}.
\eeq
\subsubsection{Hilbertian bundles}
 The corresponding Hilbertian bundles are $\tV_{i}$, equal to $V_{i}$ as complex vector bundles but equipped with the Hilbertian structures defined by
 \[
 (u| \tau_{i}v)_{\tV_{i}}= (u| v)_{V_{i}},
 \]
 for
 \begin{equation}
 \label{e3.1b}
 \tau_{0}= 1, \quad \tau_{1}= \mat{-1}{0}{0}{\one}, \quad \tau_{2}= \left(\begin{array}{ccc}
 1&0&0\\
 0&-\one&0\\
 0&0&\one
 \end{array}\right).
 \end{equation}
 As in Subsect.~\ref{sec:lg.3.1} and \ref{sec:lg.3.2}, we write an element of  $\tV_{1}$, resp.~$\tV_{2}$, as
 \[
  w= (w_{s}, w_{\Sig}), \hbox{ resp.~} u= (u_{ss}, u_{s\Sig}, u_{\Sig\Sig}),
 \]
 which gives:
\[
 \begin{array}{l}
  (w| w)_{\tV_{1}}= \bar{w_{s}}w_{s}+ \bar{w}_{\Sig}\dual \rh_{0}^{-1}\bar{w}_{\Sig},\\[2mm]
  (u| u)_{\tV_{2}}= 2(\bar{u_{ss}}u_{ss}+ 2\bar{u}_{s\Sig}\dual \rh_{0}^{-1}\bar{u}_{s\Sig}+ \bar{u}_{\Sig\Sig}\dual (h_{0}^{\otimes 2})^{-1}u_{\Sig\Sig}).
\end{array}
 \]
\def\tK{\tilde{K}}

\subsubsection{Gauge invariance}
  We recall that the gauge operator is $K =I\circ d$, see \eqref{e2.4bb}. It will be convenient to treat the factors $I$ and $d$ in $K$ separately. 
  We recall from \Lemma \ref{lemma2.1} {\it ii)} that $d= d_{0}\p_{t}+ d_{1}$ satisfies $D_{2}d= dD_{1}$. 
  We   can hence define the Cauchy surface version of $d$ by
  \beq\label{defdeTsig}
\varrho_{2}\circ d\eqdef T_{\Sigma}\circ \varrho_{1}\hbox{ on }\Ker D_{1}|_{\cinfsc(M; V_{1})}.
  \eeq
Using  \eqref{e3.1} we obtain that 
\begin{equation}
\label{e3.4}
T_{\Sigma}= \mat{d_{1}(0)}{\i d_{0}(0)}{\frac{\i}{2}(d_{0}a_{1}+ a_{2}d_{0})(0)}{d_{1}(0)}.
\end{equation}

  \subsubsection{Trace reversal}
Recall that the trace reversal  operator $\hat{I}$ in \eqref{etiti.7}   is denoted by $I_{2}$. Explicitly,
\[
I_{2}u= u- \frac{1}{4}(\rg_{0}|u)_{V_{2}}\rg_{0}, \ u\in \cinf(M; V_{2}),
\]
 and  we have also set 
\[
I_{1}u= u, \ u\in \cinf(M; V_{1}).
\]
Since $D_{i}I_{i}= I_{i}D_{i}$, we can define $I_{i\Sigma}: \coinf(\Sigma; V_{i}(\Sigma)\otimes \cc^{2})$ by
\[
\varrho_{i}\circ I_{i}\eqdef I_{i\Sigma}\circ \varrho_{i} \hbox{ on }\Ker D_{i}|_{\cinfsc(M; V_{i})}.
\]
A routine computation shows that
\begin{equation}
\label{e3.7}
I_{i\Sigma}= I_{i}\otimes \one_{\cc^{2}},
\end{equation}
and since $K= I\circ d$ we have
\[
K_{\Sigma}= I_{\Sigma}\circ T_{\Sigma}.
\]
\subsubsection{Physical charge on a Cauchy surface}
The (unphysical) charges for $D_{i}$ on the Cauchy surface $\Sigma$ corresponding to the (unphysical) Hermitian forms $(\cdot| \cdot)_{V_{i}}$ in \eqref{eq:green} are given by
\beq\label{unphysical}
q_{i}= \mat{0}{\tau_{i}}{\tau_{i}}{0}. 
\eeq
For linearized gravity one has to use the physical charges corresponding to $(\cdot| \cdot)_{I, V_{i}}$, see \eqref{e3.-1}.  On the Cauchy surface $\Sigma$ they are given by
\beq\label{e3.-2}
q_{i, {\rm phys}}= q_{i}\circ I_{i\Sigma}= \mat{0}{\tau_{i}I_{i}}{\tau_{i}I_{i}}{0}.
\eeq
The extra subscript in $q_{i, {\rm phys}}$ (absent in previous parts of the paper) is used throughout this section to disambiguate from \eqref{unphysical}.

\subsection{Wick rotation}

By the hypotheses in Subsect.~\ref{hippopotamus}, we know that  the maps $t\mapsto a_{i}(t), d_{0}(t), d_{1}(t)$ 
extend holomorphically to some disk $D_{1}(0, \epsilon)$ with values in differential operators on $\Sigma$, i.e.
\[
a_{i}\in \AT(C_{1}(0, \epsilon); \Diff^{2}(\Sigma; V_{i})), \quad d_{j}\in \AT(C_{1}(0, \epsilon); \Diff^{j}(\Sigma; V_{1}, V_{2})),
\]
$i=1,2, j= 0, 1$. Therefore we can define the \emph{Wick rotated operators}
\beq\label{wick-rotated-ops}
\ta_{i}(s)= a_{i}(\i s), \ i=1, 2, \quad \tilde{d}_{0}(s)= - \i d_{0}(\i s), \quad \tilde{d}_{1}(s)= d_{1}(\i s),
\eeq
which for  $\tilde{I}= \open{-\epsilon, \epsilon}$,  $0<\epsilon\ll 1$ satisfy
\begin{equation}
\label{e3.7b}
\ta_{i}\in \cinfb(\tilde{I}; \Diff^{2}(\Sigma; \tV_{i})), \quad \tilde{d}_{j}\in \cinfb(\tilde{I}; \Diff^{j}(\Sigma; \modif{\tilde V_{1}}, \modif{\tilde V_{2}})), 
\end{equation}
$i=1,2, j= 0, 1$.
  
\subsubsection{Wick rotated operators}
The elliptic operators \[
\tD_{i}= - \p_{s}^{2}+ \ta_{i}(s),
\] satisfy the conditions in Subsect.~\ref{sec01.3}.
 The trace operators $\trho_{i}$ are defined as in \eqref{e01.14b}.

Since  $\ta_{i}(s)= a_{i}(\i s)$ we have
\begin{equation}
\label{e3.2}
(\i\p_{t})^{n}a_{i}(0)= \p_{s}^{n}\ta_{i}(0), \ n\in \nn, \ i= 1,2.
\end{equation}
i.e.~condition \eqref{e01.23b} in Prop.~\ref{prop01.6} is satisfied. This fact will be crucial to establish the Hadamard property.

Note that since $a_{i}(t)= a_{i}^{\star}(t)$ and $a_{i}^{\star}= \tau_{i}^{-1}a^{*}\tau_{i}$ we have 
$a_{i}^{*}(t)= \tau_{i} a_{i}(t)\tau_{i}^{-1}$ hence  by analytic continuation in $t$ we have
\begin{equation}
\label{e3.2b}
\ta_{i}^{*}(s)= \tau_{i}\ta_{i}(-s)\tau_{i}^{-1}, \ s\in I, \ i= 1,2.
\end{equation}

Setting
\begin{equation}
\label{e3.2d}
\kappa_{i}u(s)\defeq \tau_{i}u(-s),
\end{equation}
 \eqref{e3.2b} is equivalent to
\begin{equation}
\label{e3.2e}
\tD_{i}^{*}=  \kappa_{i} \tD_{i}\kappa_{i}^{-1}= - \p_{s}^{2}+ \tau_{i}\ta_{i}(-s)\tau_{i}^{-1}.
\end{equation}
Since the open set $\Omega= \open{-T, T}\times \Sigma$ is invariant under $s\mapsto -s$ we obtain that
\[
\tD_{i\Omega}^{*}= \kappa_{i} \tD_{i\Omega}\kappa_{i}^{-1},
\]
where we recall that $\tD_{i\Omega}$ is the Dirichlet realization of $\tD_{i}$ in $\Omega$.

\subsubsection{Gauge invariance for Wick rotated operators}
Setting
\[
\tilde{d}= \tilde{d}_{0}(s)\p_{s}+ \tilde{d}_{1}(s),
\]
 we deduce from $D_{2}d= d D_{1}$ and analytic continuation in $t$ that
 \beq\label{e3.7c}
 \tD_{2}\tilde{d}= \tilde{d}\tD_{1}.
 \eeq 
 \def\tT{\tilde{T}} Therefore, we can  we define $\tT_{\Sigma}$ by
\[
\trho_{2}^{\pm}\circ \tilde{d}\eqdef \tT_{\Sigma}\trho_{1}^{\pm} \hbox{ on }\Ker \tD_{1}|_{\bar{\cinf}(\Omega^{\pm}; \tV_{1})}.
\]
 \begin{lemma}\label{lemma5.1}
Let  $T_{\Sigma}$ be the operator defined  in \eqref{defdeTsig}. Then $T_{\Sigma}= \tT_{\Sigma}$.

\end{lemma}
\proof  Property \eqref{e3.7c} is equivalent to:
\begin{equation}
\label{e3.3}
\begin{array}{rl}
i)& \p_{s}\tilde{d}_{0}=0,\\[2mm]
ii)& 2 \p_{s}\tilde{d}_{1}- \ta_{2}\tilde{d}_{0}+ \tilde{d}_{0}\ta_{1}=0,\\[2mm]
iii)& \p_{s}^{2}\tilde{d}_{1}- \ta_{2}\tilde{d}_{1}+ \tilde{d}_{1}\ta_{1}+ \tilde{d}_{0}\p_{s}\ta_{1}=0.
\end{array}
\end{equation}
From  \eqref{e3.3} we obtain that:
\begin{equation}
\label{e3.5}
\tT_{\Sigma}= \mat{\tilde{d}_{1}(0)}{- \tilde{d}_{0}(0)}{-\frac{1}{2}(\tilde{d}_{0}\ta_{1}+ \ta_{2}\tilde{d}_{0})(0)}{\tilde{d}_{1}(0)}.
\end{equation}
Using \eqref{e3.2}  for $n=0$ and the fact that $\tilde{d}_{1}(0)= d_{1}(0)$,  $\tilde{d}_{0}(0)=-\i d_{0}(0)$, we obtain $T_{\Sigma}= \tT_{\Sigma}$. \qed

\subsubsection{Trace reversal for Wick rotated operators}
Let us also define a `trace reversal'  for $\tD_{2}$.  We set
\[
\tilde{I}_{2}u= u- \frac{1}{4}(\rg_{0}|u)_{\tV_{2}}\rg_{0}, \ u\in \cinf(M; V_{2}).
\]
Note that $\tilde{I}_{2}= I_{2}$ and the expression above immediately shows that  $\tilde{I}_{2}$ is selfadjoint for $(\cdot| \cdot)_{\tV_{2}}$. For coherence of notation we set  again $\tilde{I}_{1}= 1$.

We deduce from $D_{2}\circ I = I \circ D_{2}$ and analyticity in $t$ that
\beq\label{e3.6c}
\tilde{I}_{2}\tD_{2}= \tD_{2}\tilde{I}_{2},
\eeq
 
We then define $\tilde{I}_{2\Sigma}$ by
\[
\tilde\varrho_{2}^{\pm} \circ \tilde{I}_{2}\defeq \tilde{I}_{2\Sigma}\circ \tilde{\varrho}_{2}^{\pm},
\]
and a routine computation shows that  $\tilde{I}_{2\Sigma}= \tilde{I}_{2}\otimes \one_{\cc^{2}}$ hence
\begin{equation}
\label{e3.6d}
\tilde{I}_{2\Sigma}= I_{2\Sigma}.
\end{equation}
In a similar vein we set 
\[
\tK= \sqrt{2}\, \tilde{I}\circ\tilde{d},
\]
so that
\[
\tD_{2}\circ \tK= \tK\circ \tD_{1}.
\]
 The Cauchy surface version of $\tK$ is 
\beq\label{e3.6e}
\tK_{\Sigma}= \tilde{I}_{\Sigma}\circ \tilde{T}_{\Sigma}= K_{\Sigma}
\eeq
by \Lemma \ref{lemma5.1} and \eqref{e3.6d}.
\subsection{\calde projectors} \label{ss:calde}

We now \modif{show a list of} properties of the \calde projectors $\tilde{c}_{i}^{\pm}$ constructed in \ref{sec01.3.7} for  the \modif{Dirichlet realizations of the}  operators $\tD_{i}$.

In Prop.~\ref{prop:calde} below, the adjoints $(\tilde{c}_i^\pm)^\dagger$ are computed with respect to the physical charges $q_{i, {\rm phys}}$, see \eqref{e3.-2}.  

\begin{proposition}\label{prop:calde}  The Calder\'on projectors $\tilde{c}_{i}^{\pm}$ for  $i=1,2$ satisfy:
\ben 
\item\label{item:cald0} $\tilde{c}_i^\pm\in\Psi^1(\Sigma;\tV_{i}\otimes \cc^{2})$;
\item\label{item:cald1} $\tilde{c}^{\pm}_{i}= (\tilde{c}^{\pm}_{i})^{2}$ and $\tilde{c}^{+}_i + \tilde{c}^{-}_i= \one$ on  $\cH^{s}(\Sigma;\tV_{i}\otimes \cc^{2})$;
\item\label{item:caldnew1}   $(I_i\otimes \one)\tilde{c}_i^\pm = \tilde{c}_i^\pm (I_i\otimes \one)$ on  $\coinf(\Sigma;\tilde V_{i}\otimes\cc^2)$;
\item\label{item:caldnew2}   $\tilde{c}_i^\pm = (\tilde{c}_i^\pm)^\dagger$;
\item\label{item:cald3} $\tilde{c}_2^\pm K_\Sig={K_\Sig} \tilde{c}_1^\pm\newmodif{\pm K_{-\infty}}$ on $\coinf(\Sigma;\tilde V_{1}\otimes\cc^2)$\newmodif{, where} $$\newmodif{K_{-\infty}\in \cW^{-\infty}(\Sigma; \tV_1 \otimes \cc^2,\tV_2 \otimes \cc^2).}$$

\een
\end{proposition}
\proof 
\eqref{item:cald0} is proved in Prop.~\ref{prop01.4}, \eqref{item:cald1} in Prop.~\ref{prop01.4b}. \eqref{item:caldnew1} follows from the definition of $\tilde{c}_{i}^{\pm}$ and the fact that $\tilde{I}_{i}\tD_{i\Omega}= \tD_{i \Omega}\tilde{I}_{i}$, since $\tilde{I}_{i}$ preserves the Dirichlet boundary condition on $\p\Omega$.

Let us now prove \eqref{item:caldnew2}. We drop the $i$ index.
 Let us set for $f, g\in \coinf(\Sigma; \tV\otimes \cc^{2})$:
 \[
 u= -r^{+}\tD_{\Omega}^{-1} \trho^{*}\tsig f, \ v= -r^{+}\kappa \tD_{\Omega}^{-1} \tilde\varrho^{*}\tsig g,
 \]
 where $r^{+}$ is the operator of restriction to $\Omega^{+}$. We use now  \cite[Lem.~A.1]{GW2} which extends to the case of vector bundles, and we obtain that $u, v\in \bar{H^{1}_{0}}(\Omega^{+}; \tV)$. Since $\tD u=\tD^{*}v=0$ in $\Omega^{+}$ we deduce from \modif{the Green's formula \eqref{e01.green1} and \eqref{greendir}} that
 \begin{equation}
 \label{e3.8}
 (\trho^{+}v| \tsig \trho^{+}u)_{\tV(\Sigma)\otimes \cc^{2}}=0.
 \end{equation}
Using that  $\trho^{+}\kappa= \mat{\tau}{0}{0}{-\tau}\trho^{-}$ we have\vspace{-0.2cm}
 \[
 \trho^{+}u=  \tilde{c}^{+}f, \quad \trho^{+}v=  \mat{\tau}{0}{0}{-\tau}\tilde{c}^{-}g. \]
Since $\tsig^{*}\mat{\tau}{0}{0}{-\tau}= -\mat{0}{\tau}{\tau}{0}$, we obtain that $\tilde{c}^{-*}\mat{0}{\tau}{\tau}{0}\tilde{c}^{+}=0$ hence $\tilde{c}^{+*}\mat{0}{\tau}{\tau}{0}= \mat{0}{\tau}{\tau}{0}\tilde{c}^{+}$.  Since $I_{\Sigma}^{*}= I_{\Sigma}$ and $\tilde{c}^{+}I_{\Sigma}= I_{\Sigma}\tilde{c}^{+}$, we obtain that 
$\tilde{c}^{+*}I_{\Sigma}\mat{0}{\tau}{\tau}{0}= I_{\Sigma}\mat{0}{\tau}{\tau}{0}\tilde{c}^{+}$, i.e.~$\tilde{c}^{+*}q_{\rm phys}= q_{\rm phys}\tilde{c}^{+}$, which proves  \eqref{item:caldnew2}.

We now prove \eqref{item:cald3}.
Let us set for $f_{i}\in \coinf(\Sigma; \tV_{i}\otimes \cc^{2})$:
\[
 u_{2}= r^{+}\modif{\tD_{2\Omega}^{-1}\tK}\trho_{1}^{*}\tsig_{1}f_{1}, \quad v_{2}= r^{+}\kappa_{2} \tD_{2\Omega}^{-1}\trho_{2}^{*}\tsig_{2}f_{2}.
\]
We have $\tD_{2}^{*}v_{2}= 0$ in $\Omega^{+}$,   $\tD_{2}u_{2}=0$ in $\Omega^{+}$. 

As before we can apply \eqref{e01.green1} and obtain
\beq\label{eq:mlml}
(\trho_{2}^{+}v_{2}| \tsig_{2} \trho_{2}^{+}u_{2})_{\tV_{2}(\Sigma)\otimes \cc^{2}}=0.
\eeq
\newmodif{Since $\tD_{2}\tK= \tK \tD_{1}$ as differential operators, we have  
\[
\bea
\trho_{2}^{+}u_{2}&= \trho_{2}^{+}\tD_{2\Omega}^{-1}\tK\trho_{1}^{*}\tsig_{1}f_{1}= \trho_{2}^{+}\tD_{2\Omega}^{-1}\tK\tD_{1}\tD_{1\Omega}^{-1}\trho_{1}^{*}\tsig_{1}f_{1}\\
&=\trho_{2}^{+}\tD_{2\Omega}^{-1}\tD_{2}\tK\tD_{1\Omega}^{-1}\trho_{1}^{*}\tsig_{1}f_{1}\\
&=\trho_{2}^{+}\tK\tD_{1\Omega}^{-1}\trho_{1}^{*}\tsig_{1}f_{1}+ \trho_{2}^{+}\tD_{2\Omega}^{-1}\tD_{2}e(K\tD_{1\Omega}^{-1}\trho_{1}^{*}\tsig f_{1})\traa{\p \Omega}
\eea
\]
where we have used Lem.~\ref{minito}  and $\tK_{\Sigma}= K_{\Sigma}$ in the last line. Using the parametrix in \ref{sec01.3.6}, see Prop.~\ref{prop01.3b}, we obtain that $(K\tD_{1\Omega}^{-1}\trho_{1}^{*}\tsig f_{1})\traa{\p \Omega}\defeq T_{-\infty}f_{1}$ is smoothing, i.e.~ $T_{-\infty}\in \cW^{-\infty}(\Sigma; \tV)$. Therefore $r_{-\infty}^{+}:= \trho_{2}^{+}\tD_{2\Omega}^{-1}\tD_{2}eT_{-\infty}$ belongs  to $\cW^{-\infty}(\Sigma; \tV\otimes \cc^{2})$.}

\newmodif{ Furthermore,
\[
\trho^{+}_{2}v_{2}= \mat{\tau_{2}}{0}{0}{-\tau_{2}}\tilde{c}_{2}^{-}f_{2}.
\]
Thus,  \eqref{eq:mlml} gives the identity
\[
\tilde{c}_{2}^{-*}\mat{0}{\tau_{2}}{\tau_{2}}{0} (K_{\Sigma}\tilde{c}_{1}^{+}- r_{-\infty}^{+})=0,
\]
which implies that $\tilde{c}_{2}^{-}(K_{\Sigma}\tilde{c}_{1}^{+}- r_{-\infty}^{+})=0$ since $\tilde{c}_{2}^{-*}\mat{0}{\tau_{2}}{\tau_{2}}{0}= \mat{0}{\tau_{2}}{\tau_{2}}{0}\tilde{c}_{2}^{-}$ by the proof of \eqref{item:caldnew2}. In a similar vein we get
\[
\tilde{c}_{2}^{+}(K_{\Sigma}\tilde{c}_{1}^{-}- r_{-\infty}^{-})=0,
\]
where $r_{-\infty}^{-}\in\cW^{-\infty}(\Sigma; \tV\otimes \cc^{2})$. Using (2) we obtain
\[
\bea
\tilde{c}_{2}^{+}K_{\Sigma}&= K_{\Sigma}\tilde{c}_{1}^{+}- \tilde{c}_{2}^{-}K_{\Sigma}\tilde{c}_{1}^{+}+ \tilde{c}_{2}^{+}K_{\Sigma}\tilde{c}_{1}^{-}\\[2mm]
&= K_{\Sigma}\tilde{c}_{1}^{+}- \tilde{c}_{2}^{-}r_{-\infty}^{+}+ \tilde c_{2}^{+}r_{-\infty}^{-}\\[2mm]
&=K_{\Sigma}\tilde{c}_{1}^{+}+ K_{-\infty},
\eea
\]
where $K_{-\infty}\in\cW^{-\infty}(\Sigma; \tV_1 \otimes \cc^2,\tV_2 \otimes \cc^2)$}. \qed\smallskip

 \modif{\begin{remark}We remark that one can attempt to modify $\tilde c_i^\pm$ (and $\tilde \pi_i^\pm$ at the same time) by changing  the boundary conditions for $\tilde D_1$ and $\tilde D_2$ at $s=T$ and $s=-T$ (in particular a well-motivated candidate is provided by boundary conditions studied in \cite{Anderson2008,witten}). If this gives a pair of \emph{invertible} operators $\tilde D_{1,\rm mod}$ and $\tilde D_{2,\rm mod}$ such that \beq\label{eq:dd}
\tilde K\Dom \tilde D_{1,\rm mod}\subset \Dom  \tilde D_{2,\rm mod},
\eeq
then one can show (under some further assumptions on the boundary conditions) that the modified operators $\tilde c_{i,{\rm mod}}^\pm$ satisfy gauge invariance $\tilde c_{2,{\rm mod}}^\pm K_\Sigma =K_\Sigma \tilde c_{1,{\rm mod}}^\pm$. This motivates a broader study of boundary conditions in Wick-rotated linearized gravity. One of the main difficulties is that it is not clear how to obtain operators satisfying  \eqref{eq:dd} and which are at the same time invertible.
\end{remark}}

The next lemma states that there exists a gauge transformation in the Wick-rotated setting mapping to tensors which have no mixed components and are symmetric with respect to the trace reversal at $\Sigma$. Furthermore, the gauge transformation preserves boundary conditions and is given in terms of a solution $w$ rather than some arbitrary $(0,1)$-tensor.

\begin{lemma}\label{lem:alt2}
 Let $u\in H^{s}(\Omega^{\pm}; \tV)\cap\modif{\bar{H^{1}}(\Sigma^{\pm}; \tV_{2})}$ for some $s>\12$. Then there exists $w\in H^{s+1}(\Omega^{\pm}; \tV_{1})\cap \bar{H^{1}_{0}}(\Omega^{\pm}; \tV_{1})$ such that $\tD_{1}w=0$ in $\Omega^{\pm}$ and $v= u- \tilde{K}w$ satisfies
 \begin{equation}
 \label{e5.1}
\begin{cases}
 (v_{s\Sigma})\traa{\Sigma}=0, \\[1mm]
 (\tilde{I}_{2}v)\traa{\Sigma}= v\traa{\Sigma}.
\end{cases}
\end{equation}
 \end{lemma}
\proof
It suffices to prove the `$+$' case. We rewrite \eqref{e5.1} as
\[
\begin{cases}
(\tilde{K}w)_{s\Sigma}\traa{\Sigma}=u_{\Sigma}, \\[1mm]
 (\rg_{0}| \tilde{K}w)_{\tV_{2}}= u_{s},
\end{cases}
\]
for $u_{\Sigma}= u_{s\Sigma}\traa{\Sigma}\in H^{s-\12}(\Sigma; T^{*}\Sigma)$ and $u_{s}= (\rg_{0}| u)_{\tV_{2}}\traa{\Sigma}\in H^{s-\12}(\Sigma; \cc)$.

Using that $\tilde{K}= \tilde{I}_{2}\tilde{d}$, $\tilde{I}_{2}= \tilde{I}_{2}^{*}$, $\tilde{I}_{2}\rg_{0}= - \rg_{0}$ and the fact that  $\tilde{I}$ does not act on the $_{s\Sigma}$ components, this is equivalent to
\begin{equation}
\label{e5.2}
 \begin{cases}
(\tilde{d}w)_{s\Sigma}\traa{\Sigma}=u_{\Sigma},\\[1mm]
(\rg_{0}|\tilde{d}w)_{\tV_{2}}\traa{\Sigma}=-u_{s}.
\end{cases}
\end{equation}
\def\DtN{Dirichlet-to-Neumann }

Let $w\in H^{s+1}(\Omega^{\pm}; \tV_{1})\cap \bar{H^{1}_{0}}(\Omega^{\pm}; \tV_{1})$ such that $\tD_{1}w=0$ in $\Omega^{+}$. Let us compute $\tilde{d}w\traa{\Sigma}$.   We replace $\p_{t}$ by $-\i\p_{s}$ and set $t=0$ in the expression of $d$ in Prop.~\ref{prop4.2}, which corresponds to the   relationship between $d$ and $\tilde{d}$ at $s=0$ given by \eqref{wick-rotated-ops}.  Keeping in mind that $\ru(0)= \one$ and $\bs(0)= 1$, we obtain
\beq\label{e5.3b}
\bea
 (\tilde{d}w)_{ss}\traa{\Sigma}&= -\i\p_{s}w_{s}\traa{\Sigma}- \12 \trace(r_{0}) w_{s}\traa{\Sigma},\\
 (\tilde{d}w)_{s\Sigma}\traa{\Sigma}&= \12 ( -\i\p_{s}w_{\Sigma}\traa{\Sigma}- \12  \trace(r_{0})w_{\Sigma}\traa{\Sigma}- \rer_{0}w_{\Sigma}\traa{\Sigma}+ d_{\Sigma}w_{s}\traa{\Sigma}),\\
 (\tilde{d}w)_{\Sigma\Sigma}\traa{\Sigma}&=  d_{\Sigma}w_{\Sigma}\traa{\Sigma}- \12 \p_{t}\rh_{0}w_{s}\traa{\Sigma},
\eea
 \eeq
 and hence
 \[
\bea
  \12(\rg_{0}|\tilde{d}w)_{\tV_{2}}\traa{\Sigma}&= - (\tilde{d}w)_{ss}\traa{\Sigma}+ (\rh_{0}|  (\tilde{d}w)_{\Sigma\Sigma}\traa{\Sigma})_{_{h_{0}^{\otimes 2}}}\\
  &=\i \p_{s}w_{s}\traa{\Sigma}+ \12  \trace(r_{0})w_{s}\traa{\Sigma}+ (\rh_{0}| d_{\Sigma}w_{\Sigma})_{_{h_{0}^{\otimes 2}}}- \12 (\rh_{0}| \p_{t}\rh_{0})_{_{h_{0}^{\otimes 2}}}w_{s}\traa{\Sigma}\\
  &=\i \p_{s}w_{s}\traa{\Sigma}- \12 \trace(r_{0})w_{s}\traa{\Sigma}- \delta_{\Sigma}w_{\Sigma}\traa{\Sigma}
\eea
 \]
 
 Since $\tD_{1}w=0$ in $\Omega^{+}$ and $w\traa{\p\Omega^{+}\setminus \Sigma}=0$, we have $(-\p_{s}w)\traa{\Sigma}= N_{\Omega^{+}}w\traa{\Sigma}$ where $N_{\Omega^{+}}$ is the \DtN map for $\tD_{1}$, see \ref{sec01.3.8}.  Setting 
 \[
 y\defeq w\traa{\Sigma}, \quad H= N_{\Omega^{+}}+\i \mat{-\12  \trace(r_{0})}{-\delta_{\Sigma}}{-d_{\Sigma}}{\12  \trace(r_{0})+ \rer_{0}},
 \]
  we can rewrite \eqref{e5.2} as
 \begin{equation}
 \label{e5.3}
 Hy=\col{-\frac{\i}{2}u_{s}}{2\i u_{\Sigma}}.
 \end{equation}
Next, $H$  belongs to $\Psi^{1}(\Sigma; \tV_{1})$. By Prop.~\ref{prop01.5} its principal symbol is
 $$\sigma_{\rm pr}(H)= ({\rm k}\dual \rh_{0}^{-1}{\rm k})^{\12}\one+ \i \mat{0}{-({\rm k}| \cdot)_{\rh_{0}}}{- {\rm k}}{0},$$ where we recall that  $(\rx, {\rm k})$  are the variables in $T^{*}\Sigma$. It follows that $H$ is elliptic in $\Psi^{1}$. 
 
 We claim that $H: H^{s+ \12}(\Sigma; \tV_{1})\to H^{s-\12}(\Sigma; \tV_{1})$ is invertible. Let  $Q_{H}$ be the sesquilinear form $(\cdot| H\cdot)_{\tV_{1}(\Sigma)}$ with domain $H^{\12}(\Sigma; \tV_{1})$. Since $H$ is elliptic, $Q_{H}$ is closed, and in view of Prop.~
 \ref{prop01.5b} $Q_{H}$ is coercive. By the Lax--Milgram theorem this implies that $H: H^{\12}(\Sigma, \tV_{1})\tosim H^{-\12}(\Sigma; \tV)$.  which proves our claim for $s= 0$. For arbitrary $s$ we use the standard argument of commutation with tangential derivatives.

 Therefore there exists a unique solution $y\in H^{s+ \12}(\Sigma; \tV_{1})$ of \eqref{e5.3}. 
We choose  $w= P_{\Omega^{+}}y$, see \ref{sec01.3.8}, i.e.~we take $w$ to be the unique solution of the boundary value problem
\[
\begin{cases}
\tD_{1}w= 0\hbox{ in }\Omega^{+},\\
w\traa{\p \Omega^{+}\setminus \Sigma}=0, \\
w\traa{\Sigma}= v.
\end{cases}
\]
By the arguments in \ref{sec01.3.8}, $w$ belongs to $H^{s+1}(\Sigma^{+}; \tV_{1})\cap \bar{H^{1}_{0}}(\Omega^{+}; \tV_{1})$. This completes the proof of the lemma. \qed

We now give a version of \Lemma \ref{lem:alt2} expressed in terms of Cauchy data. 
 Let us introduce the operator
$$
J_2 = \begin{pmatrix} I_2 \tau_2 & 0 \\ 0 & \one
\end{pmatrix},
$$
 which satisfies
\beq\label{e4.78}
J_2^*=J_2,  \quad J_2^2=\one, \quad  
J_2 q_{2, {\rm phys}} J_2 = \mat{0}{\one}{\one}{0}= \tilde{q}_2.
\eeq
In particular, $J_2$ transforms the physical charge $q_{2, {\rm phys}}$ into the Euclidean charge $\tilde{q}_2$.

Below, we set  $H^{\infty}(\Omega^{\pm}; \tV_{i} )= \bigcap_{s\in\rr}H^{s}(\Omega^{\pm}; \tV_{i} )$ and we define $\cH^{\infty}(\Sigma; \tV_{i}\otimes \cc^{2})$ similarly.

\begin{lemma}\label{lem:synchronous2}  
For all $f\in\Ker K_\Sig^\dagger|_{\coinf}$ there exists $h\in \cH^{\infty}(\Sigma; \tV_{1}\otimes \cc^{2})$ and $k\in \cH^{\infty}(\Sigma; \tV_{2}\otimes \cc^{2})$ such that  
\[ \tilde{c}^\pm_2 f= k + K_\Sig \tilde{c}_{1}^{\pm}h, \quad J_2 k=k.
\] 
\end{lemma}
\proof  We prove only the $+$ case. Let  $f\in\Ker K_\Sig^\dagger|_{\coinf}$ and
$$
u = -r^+ \tD_{2\Omega}^{-1}\trho_{2}^{*} \tsig_{2} f \in  \overline{H^{1}_{0}}(\Omega^{+};\tV_2).  
$$
We see that $u$ is the solution of 
\[
\begin{cases}
\tD_{2}u= 0\hbox{ in }\Omega^{+},\\
u\traa{\p\Omega^{+}\setminus \Sigma}=0,\\
u\traa{\Sigma}= \pi_{0}\tilde{c}_{2}^{+}f,
\end{cases}
\]
where $\pi_{0}: \col{f_{0}}{f_{1}}\mapsto  f_{0}$ is the projection on the first component.
Commuting tangential derivatives  and using that $\tilde{c}_{2}^{+}f\in \cH^{\infty}(\Sigma; \tV\otimes \cc^{2})$, we obtain that $u\in H^{\infty}(\Omega^{+}; \tV)$. 
 Let $w$ be as in \Lemma \ref{lem:alt2} and $v= u- \tK w$,  $h= \trho_{1}^{+}w, k=\trho_{2}^{+}v$. We have 
 \[
 \tilde{c}_{2}^{+}f= \trho_{2}^{+}u= k+ K_{\Sigma}h
 \]
since $K_{\Sigma}= \tK_{\Sigma}$, see \eqref{e3.6e}. The conditions on $v\traa{\Sigma}$ in \Lemma \ref{lem:alt2} mean that $J_{2}k=k$.  Moreover 
 $h\in \cH^{\infty}(\Sigma; \tV_{1}\otimes \cc^{2})$ hence $k\in \cH^{\infty}(\Sigma; \tV_{2}\otimes \cc^{2})$.  Finally we claim that $h= \tilde{c}_{1}^{+}h$. In fact let $ew$ the extension of $w$ by $0$ in $\Omega^{-}$. Then $\tD_{1}ew= \trho_{1}^{*}\tsig_{1}h$ and $ew\traa{\p \Omega}=0$ so $ew= \tD_{1\Omega}^{-1}\trho_{1}^{*}\tsig_{1}h$, hence  $h= \trho_{1}^{+}w= \trho_{1}^{+}ew= \tilde{c}_{1}^{+}h$. This completes the proof of the lemma. \qed
 
 \smallskip
 
 We now prove the positivity of \calde projectors for the physical charge $q_{2,\rm phys}$ on  \modif{$\Ker K_{\Sigma}^{\dag}|_{\coinf}$, modulo a smoothing term arising}.
 
\begin{proposition}\label{prop:positivity}
Let $\tilde{c}_{2}^{\pm}$ be the \calde projectors for $\tilde D_2$. Then
\[
\pm (f| q_{2, {\rm phys}} \modif{(}\tilde c_2^\pm + \modif{\tr_{2,-\infty}^\pm)} f)_{\tV(\Sigma)\otimes \cc^{2}}\geq 0 \ \  \forall  f\in \Ker K_{\Sigma}^{\dag}|_{\coinf}
\]
\modif{where ${\tr_{2,-\infty}^\pm}\in \cW^{-\infty}(\Sigma; \tV_2 \otimes \cc^2)$.}
\end{proposition}
\proof We only prove  the `$+$' case. Let $f\in \Ker K_{\Sigma}^{\dag}|_{\coinf}$ and let $k, h$ be as in \Lemma \ref{lem:synchronous2}.   We recall that \[
\begin{array}{l}
\tilde{c}_{2}^{+}f= k + K_{\Sigma}\tilde{c}_{1}^{+}h,\\[2mm]
K_{\Sigma}^{\dag}\tilde{c}_{2}^{+}f= \tilde{c}_{1}^{+}K_{\Sigma}^{\dag}f \modif{ + K^{\dagger}_{-\infty} f}= \modif{ K^{\dagger}_{-\infty} f},\\[2mm]
K_{\Sigma}^{\dag}k= K_{\Sigma}^{\dag}(\tilde{c}_{2}^{+}f- K_{\Sigma}\tilde{c}_{1}^{+}h)=\modif{ K^{\dagger}_{-\infty} f}.
\end{array}
\]
 Since $(\tilde{c}_{2}^{+})^{2}= \tilde{c}_{2}^{+}$ on $\cH^{\infty}(\Sigma; \tV_{2}\otimes \cc^{2})$ and $\tilde{c}_{2}^{+}= (\tilde{c}_{2}^{+})^{\dag}$ we obtain
\[
\bea
&(f| q_{2, {\rm phys}}\tilde{c}_{2}^{+}f)_{\tV_{2}(\Sigma)\otimes \cc^{2}}= (\tilde{c}_{2}^{+}f| q_{2, {\rm phys}}\tilde{c}_{2}^{+}f)_{\tV_{2}(\Sigma)\otimes \cc^{2}}\\
&=(k| q_{2, {\rm phys}}\tilde{c}_{2}^{+}f)_{\tV_{2}(\Sigma)\otimes \cc^{2}} \modif{+   (h| q_{1, {\rm phys}}\tilde{c}_{1}^{+} K^{\dagger}_{-\infty} f)_{\tV_{1}(\Sigma)\otimes \cc^{2}} }\\
&= (k| q_{2, {\rm phys}}k)_{\tV_{2}(\Sigma)\otimes \cc^{2}}  \modif{+ (\tilde{c}_{1}^{+} K^{\dagger}_{-\infty} f | q_{1, {\rm phys}}h)_{\tV_{1}(\Sigma)\otimes \cc^{2}}+  (h| q_{1, {\rm phys}}\tilde{c}_{1}^{+} K^{\dagger}_{-\infty} f)_{\tV_{1}(\Sigma)\otimes \cc^{2}} }
\eea
\]
Since $J_{2}k= k$ we have by \eqref{e4.78}:
$$
 (k| q_{2, {\rm phys}}k)_{\tV_{2}(\Sigma)\otimes \cc^{2}}= (k | \tilde{q}_{2}k)_{\tV_{2}(\Sigma)}.
$$
Next, $\modif{\tilde k}\defeq \tilde{c}_{2}^{+}(f- K_{\Sigma}\tilde{c}_{1}^{+}h)= \tilde{c}_{2}^{+}( f- K_{\Sigma}h)$, hence $\modif{\tilde k}= \trho_{2}^{+}v$ for $v= \tD_{2\Omega}^{-1} \trho_{2}^{*}\tsig_{2}(f- K_{\Sigma} h)$. Since $\tD_{2}v= 0$ in $\Omega^{+}$ with $v\traa{\p \Omega^{+}\setminus \Sigma}=0$ we obtain by  Green's formula 
\eqref{e01.green2} that
\beq\label{eq:kaku}
(\modif{\tilde k} | \tilde{q}_{2}\modif{\tilde k})_{\tV_{2}(\Sigma)}= 2\Re Q_{\Omega^{+}}(v, v)\geq 0,
\eeq
where the positivity follows from coercivity of $Q_{\Omega^{+}}$, see \eqref{eq:coercive}. \modif{Finally, }
$$
\modif{k=\tilde k - \tilde c_2^- K_\Sigma \tilde c_1^+ h = \tilde k - \tilde c_2^- K_{-\infty} h.}
$$
\modif{In conclusion, $
(f| q_{2, {\rm phys}}\tilde{c}_{2}^{+}f)_{\tV_{2}(\Sigma)\otimes \cc^{2}}
$
 equals \eqref{eq:kaku} modulo 
 $$
 (f| q_{2, {\rm phys}}\tilde{r}_{2,-\infty}^{+}f)_{\tV_{2}(\Sigma)\otimes \cc^{2}},
 $$ where $\tilde{r}_{2,-\infty}^{+}\in \cW^{-\infty}(\Sigma; \tV_2 \otimes \cc^2)$ is the sum of terms obtained by composition of $K_{-\infty}$ with $\tilde c_i^+$, $i=1,2$, $K_\Sigma$, $J_2$, and the Poisson operator  from \Lemma \ref{lem:alt2}.  } \qed

\appendix

  \section{}
  
  \label{app}\init
  
  \subsection{Symbolic fixed point}
  We recall a useful way to solve  recursive equations that are often encountered in symbolic calculus. We refer the reader to \cite[Lem.~A.1]{GW0} for the proof.
\begin{proposition}\label{prop:fixed-point}
 Suppose $F: \cinfb(I; \Psi^{\infty}(M; V))\to \cinfb(I; \Psi^{\infty}(M; V))$ is a map such that:
 \[
 \begin{array}{rl}
 i)&F: \cinfb(I; \Psi^{0}(M; V))\to \cinfb(I; \Psi^{-1}(M; V)),\\[2mm]
 ii)&b_{1}- b_{2}\in \cinfb(I; \Psi^{-j}(M; V)) \implies F(b_{1})- F(b_{2})\in \cinfb(I; \Psi^{-j-1}(M; V)) \ \forall j\in \nn.
 \end{array}
 \]
 Moreover, let  $a\in \cinfb(I; \Psi^{0}(M; V))$. Then there exists a solution $b\in \cinfb(I; \Psi^{0}(M; V))$, unique modulo $\cinfb(I; \Psi^{-\infty}(M; V))$, of the equation
 \[
 b= a+ F(b)\hbox{ mod }\cinfb(I; \Psi^{-\infty}(M; V)).
 \]
 \end{proposition}
  \subsection{Bounded analytic geometry}\label{app.1}
Let us first prove the claim in \ref{an-bg.2.3} about the independence of the spaces $\AT^{p}_{q}(M, \hat{g})$ in Def.~\ref{defp0.2} on the choice of the bounded analytic atlas $\{(U_{x}, \psi_{x})\}_{x\in M}$.  

Clearly  to define an manifold $(M, \hat{g})$ of bounded analytic geometry, there is some freedom in the ranges of $\psi_{x}$ and the holomorphy domain of $\hat{g}_{x}$.   In fact we can require equivalently in Def.~\ref{def1.1} that 
$\psi_{x}: U_{x}\tosim B_{n}(0, 1)$ and 
\ben
\item $\{\hat{g}_{x}\}_{x\in M}$ is bounded in $\BT(B_{n}(0, 1), \delta)$,
\item $c^{-1}\delta\leq \hat{g}_{x}\leq c \delta$ uniformly in $x\in M$,
\item $\{\hat{g}_{x}\}_{x\in M}$ is bounded in $\AT(C_{n}(0, \epsilon), \delta)$ for some $\epsilon>0$.
\een

The next proposition is the analog of \cite[Thm.~2.4]{GOW}. It implies that as $\{(U_{x}, \psi_{x})\}_{x\in M}$ one can take $U_{x}= B^{{\hat g}}_{M}(x, \epsilon)$,  for some $\epsilon>0$ small enough, and $\psi_{x}^{-1}= \exp^{{\hat g}}_{x}\circ e_{x}$, where $e_{x}: (\rr^{n}, \delta)\to (T_{x}M, {\hat g}(x))$ is a linear isometry.
It follows  from Prop.~\ref{propadd1.1} that the definition of the spaces $\AT^{p}_{q}(M, \hat{g})$ in Def.~\ref{defp0.2} is indeed independent on the choice of the bounded analytic atlas $\{(U_{x}, \psi_{x})\}_{x\in M}$, since we can  always take the geodesic maps $\exp^{\hat{g}}_{x}\circ e_{x}$
instead of $\psi_{x}^{-1}$. 

\begin{proposition}\label{propadd1.1}
 There exists $\epsilon, \epsilon_{1}, \epsilon_{2}>0$ such that if 
 \[
 \chi_{x}\defeq \psi_{x}\circ \exp^{{\hat g}}_{x}\circ e_{x}: B_{n}(0, \epsilon)\to \psi_{x}(B^{{\hat g}}_{M}(x, \epsilon))
 \]
 then 
 the family $\{\chi_{x}\}_{x\in M}$ is bounded in $\AT(C_{n}(0, \epsilon_{1}))$, $C_{n}(0, \epsilon_{2})\subset \chi_{x}(C_{n}(0, \epsilon_{1}))$ and the family $\{\chi_{x}^{-1}\}_{x\in M}$ is bounded in $\AT(C_{n}(0, \epsilon_{2}))$.
\end{proposition}
\proof
We choose $\epsilon$ as in the proof of \cite[Thm.~2.4]{GOW}, and $\epsilon_{1}$ such that  $C_{n}(0, 2\epsilon_{1})\subset B_{n}(0, \epsilon)$  and (3) above holds for $2\epsilon_{1}$.

We have $\psi_{x}\circ \exp_{x}^{\hat g}\circ e_{x}= \exp^{{\hat g}_{x}}_{0}\circ T_{x}$, for $T_{x}= D_{x}\psi_{x}\circ e_{x}\in L(\rr^{n})$. Using that $T_{x}^{\rm t}{\hat g}_{x}(0)T_{x}= e_{x}^{\rm t} {\hat g}(x)e_{x}= \delta$ and (2)  we obtain that $\{T_{x}\}_{x\in M}$ is bounded in $L(\rr^{n})$.

Denoting by $\Gamma_{ij, x}^{k}$ the Christoffel symbols of the metric ${\hat g}_{x}$,  $x(t)= \exp^{{\hat g}_{x}}_{0}(tv)$ solves the geodesic equations:
\[
\begin{cases}
\ddot{x}^{k}(t)= \Gamma^{k}_{ij, x}(x(t))\dot{x}^{i}(t)\dot{x}^{j}(t), \\
x(0)=0, \   \dot{x}(0)=v.
\end{cases}
\]
Since $\{\hat{g}_{x}\}_{x\in M}$ is bounded in  $\AT(C_{n}(0, \epsilon_{1}), \delta)$,  the family $\{\Gamma^{k}_{ij, x}\}_{x\in M}$ is bounded in $\AT(C_{n}(0, 2\epsilon_{1}))$ and we deduce from the  Cauchy--Kowalevski  theorem that   the family of diffeomorphisms $\{\exp^{\hat{g}_{x}}_{0}\}_{x\in M}$ is bounded in $\AT(C_{n}(0, \epsilon_{1}))$. The same is true for $\chi_{x}= \exp^{\hat{g}_{x}}_{0}\circ T_{x}$, since $\{T_{x}\}_{x\in M}$ is bounded in $L(\rr^{n})$.  Since by (2) $D\chi_{x}(0)$ is uniformly bounded in $L(\rr^{n})$, we can find $\epsilon_{2}>0$ such that  $C_{n}(0, \epsilon_{2})\subset \chi_{x}(C_{n}(0, \epsilon_{1}))$ and the family $\{\chi_{x}^{-1}\}_{x\in M}$ is bounded in $\AT(C_{n}(0, \epsilon_{2}))$. \qed

\subsubsection{Proof of Thm.~\ref{th-omar-analytic}}
 Without loss of generality we can assume that $U= M$. All statements are already proved in Thm.~\ref{th-omar}, except for the analyticity. 

For $x\in \Sigma$ we choose $U_{x}, \psi_{x}$ such that if $\Sigma_{x}= \psi_{x}(\Sigma\cap U_{x})$, then $\Sigma_{x}\cap C_{n}(0, \epsilon)= C_{n}(0, \epsilon)\cap \{v_{n}=0\}$. If $\rg_{x}= (\psi_{x}^{-1})^{*}\rg$ and $n_{x}$ is the future unit normal vector field to $\Sigma_{x}$ for $\rg_{x}$, then we can decompose $n_{x}$ as $n_{x}= n_{x}'+ \lambda_{x}e_{n}$. The families $\{\rg_{x}\}_{x\in \Sigma}$, $\{\rg_{x}^{-1}\}_{x\in \Sigma}$,  $\{n_{x}'\}_{x\in \Sigma}$ and  $\{\lambda_{x}\}_{x\in \Sigma}$ are bounded in $\AT^{0}_{2}(C_{n-1}(0, \epsilon), \delta)$, $\AT^{2}_{0}(C_{n-1}(0, \epsilon), \delta)$, $\AT^{1}_{0}(C_{n-1}(0, \epsilon), \delta)$ and $\AT^{0}_{0}(C_{n-1}(0, \epsilon), \delta)$ respectively.  We introduce the normal geodesic flow to $\Sigma_{x}$ for $\rg_{x}$:
\[
\chi_{x}:\ \begin{array}{l}
\open{-\epsilon_{1}, \epsilon_{1}}\times C_{n-1}(0,\epsilon_{1})\to C_{n}(0,\epsilon)\\
(t, v')\mapsto \exp_{(v',0)}^{g_{x}}(tn_{x}(v', 0)).
\end{array}
\]
By the Cauchy--Kowalevski theorem we obtain that  for $\epsilon_{1}$ small enough $\{\chi_{x}\}_{x\in \Sigma}$ is bounded in $\AT(C_{n}(0, \epsilon_{1}), \delta)$. We have $\chi_{x}^{*}\rg_{x}= - dt^{2}+ \rh_{x}(t, v')dv'^{2}$, and  $\{\rh_{x}\}_{x\in \Sigma}$, $\{\rh_{x}^{-1}\}_{x\in \Sigma}$ are bounded in $\AT^{0}_{2}(C_{n}(0, \epsilon_{1}), \delta)$ and $\AT^{2}_{0}(C_{n}(0, \epsilon_{1}), \delta)$ respectively, using the same properties of $\chi_{x}, \rg_{x}, \rg_{x}^{-1}$ recalled above. \qed

\bigskip

{\small
\subsubsection*{Acknowledgments} S.M.~was supported by the DFG research grant MU 4559/1-1 
``Hadamard States in Linearized Quantum Gravity''. Support from the grant
ANR-20-CE40-0018 is gratefully acknowledged. The research leading to these results has received funding from the European Union's Horizon 2020 research and innovation programme under the Marie Sk{\l}odowska-Curie grant agreement No 754340.   \medskip }

 \bibliographystyle{abbrv}
 \bibliography{linearizedgravity}

\begin{thebibliography}{10}

\bibitem{AFO}
B.~Allen, A.~Folacci, and A.~C. Ottewill.
\newblock {Renormalized graviton stress-energy tensor in curved vacuum
  space-times}.
\newblock {\em Phys. Rev. D}, 38(4):1069--1082, 1988.

\bibitem{Anderson2008}
M.~T. Anderson.
\newblock {On boundary value problems for Einstein metrics}.
\newblock {\em Geom. Topol.}, 12(4):2009--2045, 2008.

\bibitem{AA}
A.~Ashtekar and A.~Magnon-Ashtekar.
\newblock {On the symplectic structure of general relativity}.
\newblock {\em Commun. Math. Phys.}, 86(1):55--68, 1982.

\bibitem{Bar2012}
C.~B{\"{a}}r and N.~Ginoux.
\newblock {Classical and Quantum Fields on Lorentzian Manifolds}.
\newblock In C.~B{\"{a}}r, J.~Lohkamp, and M.~Schwarz, editors, {\em Glob.
  Differ. Geom. Springer Proc. Math. 17}, pages 359--400. Springer-Verlag,
  2012.

\bibitem{BGP}
C.~B{\"{a}}r, N.~Ginoux, and F.~Pf{\"{a}}ffle.
\newblock {\em {Wave Equations on Lorentzian Manifolds and Quantization}}.
\newblock European Mathematical Society Publishing House, Z{\"{u}}rich, 2007.

\bibitem{BDM}
M.~Benini, C.~Dappiaggi, and S.~Murro.
\newblock {Radiative observables for linearized gravity on asymptotically flat
  spacetimes and their boundary induced states}.
\newblock {\em J. Math. Phys.}, 55(8):082301, 2014.

\bibitem{borthwick}
J.~Borthwick.
\newblock {Maximal Kerr–de Sitter spacetimes}.
\newblock {\em Class. Quantum Gravity}, 35(21):215006, 2018.

\bibitem{BFH}
R.~Brunetti, K.~Fredenhagen, T.-P. Hack, N.~Pinamonti, and K.~Rejzner.
\newblock {Cosmological perturbation theory and quantum gravity}.
\newblock {\em J. High Energy Phys.}, 2016(8):32, 2016.

\bibitem{BFR}
R.~Brunetti, K.~Fredenhagen, and K.~Rejzner.
\newblock {Quantum gravity from the point of view of locally covariant quantum
  field theory}.
\newblock {\em Commun. Math. Phys.}, 345(3):741--779, 2016.

\bibitem{DMP}
C.~Dappiaggi, V.~Moretti, and N.~Pinamonti.
\newblock {\em {Hadamard States from Light-like Hypersurfaces}}, volume~25 of
  {\em SpringerBriefs in Mathematical Physics}.
\newblock Springer International Publishing, Cham, 2017.

\bibitem{DS}
C.~Dappiaggi and D.~Siemssen.
\newblock {Hadamard states for the vector potential on asymptotically flat
  spacetimes}.
\newblock {\em Rev. Math. Phys.}, 25(01):1350002, 2013.

\bibitem{DH}
J.~J. Duistermaat and L.~H{\"{o}}rmander.
\newblock {Fourier integral operators. II}.
\newblock {\em Acta Math.}, 128:183--269, 1972.

\bibitem{FH}
C.~J. Fewster and D.~S. Hunt.
\newblock {Quantization of linearized gravity in cosmological vacuum
  spacetimes}.
\newblock {\em Rev. Math. Phys.}, 25(02):1330003, 2013.

\bibitem{FP}
C.~J. Fewster and M.~J. Pfenning.
\newblock {A quantum weak energy inequality for spin-one fields in curved
  space–time}.
\newblock {\em J. Math. Phys.}, 44(10):4480, 2003.

\bibitem{FV}
C.~J. Fewster and R.~Verch.
\newblock {The necessity of the Hadamard condition}.
\newblock {\em Class. Quantum Gravity}, 30(23):235027, 2013.

\bibitem{FS}
F.~Finster and A.~Strohmaier.
\newblock {Gupta–Bleuler quantization of the Maxwell field in globally
  hyperbolic space-times}.
\newblock {\em Ann. Henri Poincar{\'{e}}}, 16(8):1837--1868, 2015.

\bibitem{FR}
K.~Fredenhagen and K.~Rejzner.
\newblock {Batalin-Vilkovisky formalism in perturbative Algebraic Quantum Field
  Theory}.
\newblock {\em Commun. Math. Phys.}, 317(3):697--725, 2013.

\bibitem{Fritzsch2023}
K.~Fritzsch, D.~Grieser, and E.~Schrohe.
\newblock {The Calder{\'{o}}n projector for fibred cusp operators}.
\newblock {\em J. Funct. Anal.}, 285(10):110127, 2023.

\bibitem{FNW}
S.~Fulling, F.~Narcowich, and R.~M. Wald.
\newblock {Singularity structure of the two-point function in quantum field
  theory in curved spacetime, II}.
\newblock {\em Ann. Phys. (N. Y).}, 136(2):243--272, 1981.

\bibitem{furlani}
E.~P. Furlani.
\newblock {Quantization of the electromagnetic field on static space–times}.
\newblock {\em J. Math. Phys.}, 36(3):1063--1079, 1995.

\bibitem{HHI}
C.~G{\'{e}}rard.
\newblock {The Hartle–Hawking–Israel state on spacetimes with stationary
  bifurcate Killing horizons}.
\newblock {\em Rev. Math. Phys.}, 33(08):2150028, 2021.

\bibitem{GHW}
C.~G{\'{e}}rard, D.~H{\"{a}}fner, and M.~Wrochna.
\newblock {The Unruh state for massless fermions on Kerr spacetime and its
  Hadamard property}.
\newblock {\em arXiv:2008.10995}, 2020.

\bibitem{GOW}
C.~G{\'{e}}rard, O.~Oulghazi, and M.~Wrochna.
\newblock {Hadamard states for the Klein–Gordon equation on Lorentzian
  manifolds of bounded geometry}.
\newblock {\em Commun. Math. Phys.}, 2017.

\bibitem{GS}
C.~G{\'{e}}rard and T.~Stoskopf.
\newblock {Hadamard states for quantized Dirac fields on Lorentzian manifolds
  of bounded geometry}.
\newblock {\em Rev. Math. Phys.}, 34(04), 2022.

\bibitem{GW0}
C.~G{\'{e}}rard and M.~Wrochna.
\newblock {Construction of Hadamard states by pseudo-differential calculus}.
\newblock {\em Commun. Math. Phys.}, 325(2):713--755, 2014.

\bibitem{GW1}
C.~G{\'{e}}rard and M.~Wrochna.
\newblock {Hadamard states for the linearized Yang–Mills equation on curved
  spacetime}.
\newblock {\em Commun. Math. Phys.}, 337(1):253--320, 2015.

\bibitem{GW2}
C.~G{\'{e}}rard and M.~Wrochna.
\newblock {Analytic Hadamard states, Calder{\'{o}}n projectors and Wick
  rotation near analytic Cauchy surfaces}.
\newblock {\em Commun. Math. Phys.}, 366(1):29--65, 2019.

\bibitem{GN}
J.~Guven and D.~N{\'{u}}{\~{n}}ez.
\newblock {Schwarzschild-de Sitter space and its perturbations}.
\newblock {\em Phys. Rev. D}, 42(8):2577--2584, 1990.

\bibitem{HS}
T.-P. Hack and A.~Schenkel.
\newblock {Linear bosonic and fermionic quantum gauge theories on curved
  spacetimes}.
\newblock {\em Gen. Relativ. Gravit.}, 45(5):877--910, 2013.

\bibitem{HHV}
D.~H{\"{a}}fner, P.~Hintz, and A.~Vasy.
\newblock {Linear stability of slowly rotating Kerr black holes}.
\newblock {\em Invent. Math.}, 223(3):1227--1406, 2021.

\bibitem{dS1}
A.~Higuchi, D.~Marolf, and I.~A. Morrison.
\newblock {de Sitter invariance of the dS graviton vacuum}.
\newblock {\em Class. Quantum Gravity}, 28(24):245012, 2011.

\bibitem{Hintz2018a}
P.~Hintz and A.~Vasy.
\newblock {The global non-linear stability of the Kerr–de Sitter family of
  black holes}.
\newblock {\em Acta Math.}, 220(1):1--206, 2018.

\bibitem{hollands}
S.~Hollands.
\newblock {Renormalized quantum Yang-Mills fields in curved spacetime}.
\newblock {\em Rev. Math. Phys.}, 20(09):1033--1172, 2008.

\bibitem{H}
L.~H{\"{o}}rmander.
\newblock {\em {The Analysis of Linear Partial Differential Operators I.
  Distribution Theory and Fourier Analysis}}.
\newblock Springer Verlag, Berlin, second edition, 1990.

\bibitem{junker}
W.~Junker.
\newblock {Hadamard states, adiabatic vacua and the construction of physical
  states for scalar quantum fields on curved spacetime}.
\newblock {\em Rev. Math. Phys.}, 08(08):1091--1159, 1996.

\bibitem{K2}
T.~Kato.
\newblock {Integration of the equation of evolution in a Banach space.}
\newblock {\em J. Math. Soc. Japan}, 5(2):208--234, 1953.

\bibitem{K1}
T.~Kato.
\newblock {\em {Perturbation Theory for Linear Operators}}.
\newblock Springer-Verlag, Berlin Heidelberg, 2nd edition, 1995.

\bibitem{Kh}
I.~Khavkine.
\newblock {Characteristics, conal geometry and causality in locally covariant
  field theory}.
\newblock {\em arXiv:1211.1914}, 2012.

\bibitem{kontsevich}
M.~Kontsevich and G.~Segal.
\newblock {Wick rotation and the positivity of energy in quantum field theory}.
\newblock {\em Q. J. Math.}, 72(1-2):673--699, 2021.

\bibitem{dS2}
S.~P. Miao, P.~J. Mora, N.~C. Tsamis, and R.~P. Woodard.
\newblock {Perils of analytic continuation}.
\newblock {\em Phys. Rev. D}, 89(10):104004, 2014.

\bibitem{moncrief}
V.~Moncrief.
\newblock {Decompositions of gravitational perturbations}.
\newblock {\em J. Math. Phys.}, 16(8):1556--1560, 1975.

\bibitem{N2}
B.~O'Neill.
\newblock {\em {The Geometry of Kerr Black Holes}}.
\newblock Dover Publications, 2014.

\bibitem{rejzner}
K.~Rejzner.
\newblock {\em {Perturbative Algebraic Quantum Field Theory}}.
\newblock Mathematical Physics Studies. Springer International Publishing,
  Cham, 2016.

\bibitem{R}
H.~Ringstr{\"{o}}m.
\newblock {\em {The Cauchy Problem in General Relativity}}.
\newblock European Mathematical Society Publishing House, Z{\"{u}}rich, 2009.

\bibitem{SV}
H.~Sahlmann and R.~Verch.
\newblock {Microlocal spectrum condition and Hadamard form for vector-valued
  quantum fields in curved spacetime}.
\newblock {\em Rev. Math. Phys.}, 13(10):1203--1246, 2001.

\bibitem{schapira}
P.~Schapira.
\newblock {Wick rotation for D-modules}.
\newblock {\em Math. Physics, Anal. Geom.}, 20(3):21, 2017.

\bibitem{SG}
J.~Schmid and M.~Griesemer.
\newblock {Kato's theorem on the integration of non-autonomous linear evolution
  equations}.
\newblock {\em Math. Physics, Anal. Geom.}, 17(3-4):265--271, 2014.

\bibitem{Sh}
M.~A. Shubin.
\newblock {Spectral theory of elliptic operators on non-compact manifolds}.
\newblock In {\em M{\'{e}}thodes semi-classiques Vol. 1 - {\'{E}}cole
  d'{\'{E}}t{\'{e}} (Nantes, juin 1991)}, number 207 in Ast{\'{e}}risque.
  Soci{\'{e}}t{\'{e}} math{\'{e}}matique de France, 1992.

\bibitem{witten}
E.~Witten.
\newblock {A note on boundary conditions in Euclidean gravity}.
\newblock {\em Rev. Math. Phys.}, 33(10), 2021.

\bibitem{wrochnawick}
M.~Wrochna.
\newblock {Wick rotation of the time variables for two-point functions on
  analytic backgrounds}.
\newblock {\em Lett. Math. Phys.}, 110(3):585--609, 2020.

\bibitem{WZ}
M.~Wrochna and J.~Zahn.
\newblock {Classical phase space and Hadamard states in the BRST formalism for
  gauge field theories on curved spacetime}.
\newblock {\em Rev. Math. Phys.}, 29(04):1750014, 2017.

\bibitem{Wuensch}
V.~W{\"{u}}nsch.
\newblock {Cauchy's problem and Huygens' principle for relativistic higher spin
  wave equations in an arbitrary curved space-time}.
\newblock {\em Gen. Relativ. Gravit.}, 17(1):15--38, 1985.

\end{thebibliography}
\end{document}